\documentclass[11pt]{article}
\usepackage{authblk}

\usepackage[utf8]{inputenc}
\usepackage[margin=1.0in]{geometry}
\usepackage{amsmath,amssymb,amsthm,dsfont}
\usepackage[inline, shortlabels]{enumitem}

\usepackage{mathtools}

\usepackage{thmtools}
\usepackage{array}
\usepackage{comment}
\usepackage{cite}
\usepackage{xspace}
\usepackage{caption}
\usepackage[backref,colorlinks=true]{hyperref}
\usepackage{makecell}
\usepackage{tcolorbox} 

\usepackage[capitalise]{cleveref}
\crefname{enumi}{}{}

\allowdisplaybreaks
\newtheorem{theorem}{Theorem}[section]

\newtheorem{lemma}[theorem]{Lemma}
\newtheorem{corollary}[theorem]{Corollary}

\theoremstyle{definition} \newtheorem{definition}[theorem]{Definition}
\theoremstyle{remark}
\newtheorem{remark}[theorem]{Remark}

\usepackage[textsize=tiny,disable]{todonotes}

\usepackage{soul}
\colorlet{colGeorge}{yellow!50!black}
\colorlet{colGeorgeBack}{yellow!20!white}
\newtcolorbox{boxGeorge}{colback=colGeorgeBack,colframe=colGeorge}

\colorlet{colAndy}{purple!50!black}
\colorlet{colAndyBack}{purple!20!white}
\newtcolorbox{boxAndy}{colback=colAndyBack,colframe=colAndy}
\newcommand{\andy}[1]{\begin{boxAndy} A: #1 \end{boxAndy}}

\colorlet{colFra}{green!50!black}
\colorlet{colFraBack}{green!20!white}
\newtcolorbox{boxFra}{colback=colFraBack,colframe=colFra}

\colorlet{colEmanuele}{blue!50!black}
\colorlet{colEmanueleBack}{blue!20!white}
\newtcolorbox{boxEmanuele}{colback=colEmanueleBack,colframe=colEmanuele}

\newcommand{\levy}{L\'evy\xspace}
\newcommand{\levywalk}{L\'evy walk\xspace}
\newcommand{\levyWalk}{L\'evy Walk\xspace}
\newcommand{\levywalks}{L\'evy walks\xspace}
\newcommand{\levyWalks}{L\'evy Walks\xspace}
\newcommand{\levyflight}{L\'evy flight\xspace}
\newcommand{\levyFlight}{L\'evy Flight\xspace}
\newcommand{\levyflights}{L\'evy flights\xspace}

\newcommand{\search }{{\mathrm{ANTS}}}
\newcommand{\algsearch}{{$\mbox{L\'evy-Alg}$}\xspace}
\newcommand{\work}{{work}\xspace}

\newcommand{\treasure}{{target}\xspace}
\newcommand{\trespoint}{{u^\ast}}

\newcommand{\variance}[1]{\mathrm{Var}\left( #1 \right)}

\DeclareMathOperator{\polylog}{polylog}
\DeclareMathOperator{\poly}{poly}
\newcommand{\bigo}{\mathcal{O}}
\newcommand{\mybigo}[1]{\mathcal{O}\left(#1\right)}
\newcommand{\myOmega}[1]{\Omega\left(#1\right)}
\newcommand{\myTheta}[1]{\Theta\left(#1\right)}

\newcommand{\mylittleo}[1]{o \left (#1 \right )}
\newcommand{\mylittleomega}[1]{ \omega \left(#1\right)}

\newcommand{\mean}{\mathbb{E}}
\newcommand{\expect}[1]{\mathbb{E}\left[ #1 \right]}

\newcommand{\abs}[1]{\left\lvert #1 \right\rvert}
\newcommand{\manhattan}[1]{\left\| #1 \right\|_1}

\newcommand{\supdist}[1]{\left \| #1 \right \|_{\infty}}

\newcommand{\cardinality}[1]{\left\lvert #1 \right\rvert}

\newcommand{\integer}{\mathbb{Z}}
\newcommand{\nat}{\mathbb{N}}
\newcommand{\real}{\mathbb{R}}

\newcommand{\prob}{\mathbb{P}}

\newcommand{\pr}[1]{\mathbb{P}\left( #1 \right)}
\newcommand{\myfloor}[1]{\left\lfloor #1 \right\rfloor}
\newcommand{\myceil}[1]{\left\lceil #1 \right\rceil}

\newcommand{\Mid}{\ \middle|\ }

\newcommand{\ind}{\mathds{1}}

\newcommand{\levywalkvisits}[2]{Z^w_{#1}\left(#2\right)}
\newcommand{\levyflightvisits}[2]{Z^f_{#1}\left(#2\right)}
\newcommand{\Zf}[2]{Z^f_{#1}\left(#2\right)}

\newcommand{\levyrand}[1]{L^w_{#1}}

\newcommand{\levyflightrand}[1]{L^f_{#1}}
\newcommand{\Lf}[1]{L^f_{#1}}

\newcommand{\genericrand}[1]{\mathbf{X}_{#1}}

\newcommand{\dmax}{d_{\text{max}}}

\newcommand{\levyconst}{c_{\alpha}}
\newcommand{\dist}{\ell}
\newcommand{\distu}{d_u}

\newcommand{\firstregion}{\mathcal{A}_1}
\newcommand{\secondregion}{\mathcal{A}_2}
\newcommand{\thirdregion}{\mathcal{A}_3}
\renewcommand{\AA}{\mathcal{A}}

\newcommand{\rwfirstregion}{\mathcal{A}_1}
\newcommand{\rwsecondregion}{\mathcal{A}_2}

\newcommand{\intervalsquare}[1]{\left[ #1 \right]}

\newcommand{\EE}{\mathcal{E}}

\newcommand{\origin}{\mathbf{0}}

\newcommand{\diff}{\mathrm{d}}
 
\title{Search via Parallel Lévy Walks on ${\mathbb Z}^2$}

\author[1]{Andrea Clementi}
\author[2]{Francesco d'Amore}
\author[3]{George Giakkoupis}
\author[2]{Emanuele Natale}

\affil[1]{Universit\`a di Roma ``Tor Vergata'', Roma, Italia \protect\\
clementi@mat.uniroma2.it}
\affil[2]{Universit\'e C\^ote d'Azur, Inria, CNRS, I3S, France\protect\\
francesco.d-amore@inria.fr\protect\\
emanuele.natale@inria.fr}
\affil[3]{Inria, Univ Rennes, CNRS, IRISA
Rennes, France\protect\\
george.giakkoupis@inria.fr}
\date{}

\begin{document}
\def\mystrut(#1,#2){\vrule height #1 depth #2 width 0pt}

\newcolumntype{C}[1]{>{\mystrut(3ex,3ex)\centering}p{#1}<{}}

\maketitle

\begin{abstract}
Motivated by the \emph{Lévy foraging hypothesis} -- the premise that various animal species have adapted to follow \emph{Lévy walks} to optimize their search efficiency -- we study the parallel hitting time of Lévy walks on the infinite two-dimensional grid.
We consider $k$ independent discrete-time Lévy walks, with the same exponent $\alpha \in(1,\infty)$, that start from the same node, and analyze the number of steps until the first walk visits a given target at distance $\ell$.
We show that for any choice of $k$ and $\ell$ from a large range, there is a unique optimal exponent $\alpha_{k,\ell} \in (2,3)$, for which the hitting time is $\tilde O(\ell^2/k)$ w.h.p., while modifying the exponent by an $\epsilon$ term increases the hitting time by a polynomial factor, or the walks fail to hit the target almost surely.
Based on that, we propose a surprisingly simple and effective parallel search strategy, for the setting where $k$ and $\ell$ are unknown:
the exponent of each Lévy walk is just chosen independently and uniformly at random from the interval $(2,3)$.
This strategy achieves optimal search time (modulo polylogarithmic factors) among all possible algorithms (even centralized ones that know $k$).
Our results should be contrasted with a line of previous work showing that the exponent $\alpha = 2$ is optimal for various search problems.
In our setting of $k$ parallel walks, we show that the optimal exponent depends on $k$ and $\ell$, and that randomizing the choice of the exponents works simultaneously for all $k$ and $\ell$.
\end{abstract}

\maketitle

\section{Introduction}
\label{sec:intro}

\subsection{Background}

A \emph{\levywalk} is a random walk process in which jump lengths are drawn from a power-law distribution.
Thus, the walk consists of a mix of long trajectories and short, random movements.
Over the last two decades, \levywalks have attracted significant attention, as a result of increasing empirical evidence that the movement patterns of various animal species resemble \levywalks.
Examples of such species range from snails~\cite{reynolds2017}, bees~\cite{Reynolds3763}, and albatross birds~\cite{viswanathan_levy_1996}, to sharks~\cite{humphries2010,sims2008}, deers~\cite{focardi2009}, and humans~\cite{boyer2006,raichlen2014}, among others \cite{reynolds_current_2018}.
Nowadays, \levywalks are
the most prominent
movement model in biology~\cite{reynolds_current_2018}, at least among models with comparable mathematical simplicity and elegance~\cite{viswanathan_physics_2011}.

The \emph{L\'evy foraging hypothesis}, put forward by Viswanathan et al.~\cite{viswanathan_optimizing_1999,viswanathan_levy_2008}, stipulates that the observed \levywalk movement patterns in animals must have been induced by natural selection, due to the optimality of \levywalks in searching for food.
Indeed, it has been shown that \levywalks achieve (near) optimal search time in certain settings.
In particular, Levy walks with exponent parameter $\alpha=2$ are optimal for searching sparse randomly distributed revisitable targets~\cite{viswanathan_optimizing_1999}.
However, these results were formally shown just for one-dimensional spaces~\cite{Buldyrev2001}, and do not carry over to higher-dimensions~\cite{levernier2020}.
Currently, there is no strong analytical evidence supporting the optimality of \levywalks in $d$-dimensional spaces for $d>1$.

In this paper we focus on \levywalks on two-dimensional spaces.
Concretely, we assume the infinite lattice $\integer^2$, and consider the Manhattan distance as the underlying metric.
To determine each jump of the \levywalk, an integer distance $d$ is chosen independently at random such that 
the probability of $d = i$ is inversely proportional to $i^\alpha$,
where $\alpha\in(1,\infty)$ is the exponent parameter of the walk.
A destination $v$ is then chosen uniformly at random among all nodes at distance $d$ from the current node $u$, and in the next $d$ steps, the process moves from $u$ to $v$ along a shortest lattice path approximating the straight line segment $\overline{uv}$.

We evaluate the search efficiency of \levywalks on $\integer^2$, by analysing the parallel hitting time of multiple walks originating at the same node.
Precisely, we assume that $k \geq 1$ independent \levywalks start simultaneously from the origin $(0,0)$ of the lattice.
Then the parallel hitting time for any given target node $\trespoint$ is the first step when some walk visits $\trespoint$.
This very basic setting can be viewed as a model of natural cooperative foraging behavior,
such as the behavior of ants around their nest.
In fact, our setting is as a special instance of the more general ANTS problem introduced by Feinerman and Korman~\cite{feinerman_ants_2017}.
The ANTS problem asks for a search strategy for $k$ independent agents that minimizes the parallel hitting time for an unknown target, subject to limited communication before the search starts.

\subsection{Our Results}

\subsubsection{Hitting Time Bounds for a Single \levyWalk}

A main technical contribution of our work is an analysis of the hitting time $\tau_\alpha(\trespoint)$ of a single \levywalk with exponent $\alpha\in(1,\infty)$, for an arbitrary target node $\trespoint$.
We show the following bounds on $\tau_\alpha(\trespoint)$,
assuming the \levywalk starts at the origin $(0,0)$, and $\trespoint$'s distance to the origin is $\dist = \|\trespoint\|_1$.

Consider first the \emph{super-diffusive} regime, where $\alpha\in (2,3)$.
In this regime, jump lengths have bounded mean and unbounded variance.
Roughly speaking, we show that in the first $t_\dist = \Theta(\dist^{\alpha - 1})$ steps,\footnote{Note that $t_\dist$ is of the same order as the expected number of steps before the first jump of length greater than $\dist$.}
the walk stays inside a ball of radius $t_\dist\cdot\polylog\dist$ with significant probability, while only a constant fraction of those steps are inside the smaller ball of radius~$\dist$.
We also show a monotonicity property, which roughly implies that the probability of visiting a node decreases as the node's distance from the origin increases.
Therefore, a constant fraction of the $t_\dist$ steps visits nodes at distances between $\ell$ and $\ell\cdot\polylog\dist$, and the visit probability of each of these nodes is upper bounded by that of node~$\trespoint$.
We thus obtain that the probability of visiting $\trespoint$ within $t_\ell$ steps
is $\myOmega{ {t_\dist}/{\dist^2 \polylog\dist}}$.

If we consider a smaller number of steps, $t = \mybigo{t_\dist/ \polylog \dist}$, then it is very likely that the walk stays in a ball of radius smaller than $\dist$, and we show a simple bound of $\mybigo{(t/t_\dist)^2\cdot {t_\dist}/{\dist^2}}$ for the probability of $\tau_\alpha(\trespoint) \leq t$, i.e., ignoring $\polylog\dist$ factors, the probability decreases by a factor of $\mybigo{(t/t_\dist)^2}$. 

On the other hand, if we consider a larger number of steps (even if $t \to \infty$), the probability that $\trespoint$ is hit does not increase significantly, just by at most a $\polylog\dist$ factor.

Therefore, in regime $\alpha\in (2,3)$, $\Theta(\dist^{\alpha - 1})$ steps suffice to maximize the hitting probability (within $\polylog\dist$ factors), while reducing this time reduces the probability super-linearly.

The \emph{diffusive} regime, $\alpha\in (3,\infty)$, is similar to the case of a simple random walk, as jump lengths have bounded mean and bounded variance.
We show that $\mybigo{\dist^2\polylog\dist}$ steps suffice to hit the target with probability $\myOmega{1/\polylog\dist}$, while for a smaller number of steps $t$, the probability decreases by a factor of $\mybigo{(t/\dist^2)^2}$. 
The behavior is similar also in the threshold case of $\alpha = 3$, even though the variance of the jump length is unbounded in this case. 

Finally, in the \emph{ballistic} regime, $\alpha\in (1,2]$, where jump lengths have unbounded mean and unbounded variance, the behavior is similar to that of a straight walk along a random direction.
We show that the target is hit with probability $\myOmega
{1/\dist\polylog\dist}$ in the first  $\myTheta{\dist}$ steps, while increasing the number of steps does not increase this probability significantly.

Below we give formal statements of these results, for the case where $\alpha$ is independent of $\dist$, as $\dist\to \infty$.
More refined statements and their proofs are given in \cref{sec:levywalk,sec:equivpwbw,sec:equivpwrw}.

\begin{theorem}
    \label{thm:lw23-intro}
    Let $\alpha$ be any real constant in  $(2,3)$  and $\trespoint$  any node in $\integer^2$ with    $\dist = \|\trespoint\|_1$. Then:
    \begin{enumerate}[(a)]
        \item\label{thm:lw23-intro:a}
            $\pr{\tau_\alpha(\trespoint)  = \mybigo{\dist^{\alpha - 1}}} = \myOmega{ {1}/{\dist^{3-\alpha}\log^2 \dist}}$;

        \item
            $\pr{\tau_\alpha(\trespoint)  \leq t } = \mybigo{{t^2}/{\dist^{\alpha+1}} }$,
            for any step $\dist \le t =\mybigo{\dist^{\alpha - 1}}$;

        \item
            $\pr{\tau_\alpha(\trespoint)  < \infty} = \mybigo{{{\log\dist}}/  {\dist^{3-\alpha}}}$.
    \end{enumerate}
\end{theorem}

\begin{theorem}\samepage
    \label{thm:lw3infty-intro}
    Let  $\alpha$ be any real constant in $ [3,\infty)$  and $\trespoint$ any node in $\integer^2$ with  $\dist = \|\trespoint\|_1$.  Then:

    \begin{enumerate}[(a)~]
        \item
            $\pr{\tau_\alpha(\trespoint)  = \mybigo{\dist^2 {\log^2\dist}}} =
            \myOmega{ 1/{\log^4\dist}}$;

        \item
            $\pr{\tau_\alpha(\trespoint)  \leq t } = ~ \mybigo{t^2{\log\dist} / \dist^4}$, for any step $t$ with $ \dist \le t = \mybigo{\dist^{2}}$.
    \end{enumerate}
\end{theorem}

\begin{theorem}
    \label{thm:lw12-intro}
    Let  $\alpha$ be any real constant in $(1,2]$ and $\trespoint$ any node in $\integer^2$ with  $\dist = \|\trespoint\|_1$. Then:
    \begin{enumerate}[(a)~]
        \item
            $\pr{\tau_\alpha(\trespoint)  = \mybigo{\dist}} = \myOmega{ 1/\dist{\log\dist}}$;
\item
            $\pr{\tau_\alpha(\trespoint)  < \infty} = \mybigo{{\log^2\dist} /\dist }$.
    \end{enumerate}
\end{theorem}

\subsubsection{Parallel \levyWalks with Common Exponent}

Consider $k\geq1$ independent identical \levywalks with exponent $\alpha\in (1,\infty)$, that start simultaneously at the origin.
Let $\tau_a^k(\trespoint)$ denote the parallel hitting time for node $\trespoint$, i.e., the first step when some walk visits $\trespoint$.
It is straightforward to derive upper and lower bounds on $\tau_a^k(\trespoint)$ from the corresponding bounds on the hitting time of a single \levywalk.
For example, the next statement is a direct corollary of \cref{thm:lw23-intro}\cref{thm:lw23-intro:a}.

\begin{corollary}
    Let  $\alpha$ be any real constant in  $(2,3)$  and $\trespoint$  any node in $\integer^2$ with    $\dist = \|\trespoint\|_1$. Then
    $$\pr{\tau^k_\alpha(\trespoint)  =
    \mybigo{\dist^{\alpha - 1}}}
    = 1- e^{-\myOmega{ {k}/{\dist^{3-\alpha}
    {\log^2 \dist}}}}.$$
\end{corollary}

From the bounds we obtain for $\tau^k_\alpha$,
it follows\footnote{In fact, we use more refined versions of \cref{thm:lw23-intro,thm:lw12-intro,thm:lw3infty-intro}, to obtain bounds on $\tau^k_\alpha$ which allow $\alpha$ to be a function of $\dist$ and $k$.}
that, for each pair of $k$ and $\dist = \|\trespoint\|_1$ with $\polylog\dist \leq k \leq \dist\polylog\dist$,
    there is a unique optimal exponent $\alpha = 3 - \frac{\log k}{\log\dist} + \mybigo{\frac{\log\log\dist}{\log\dist}}$, which minimizes $\tau^k_\alpha(\trespoint)$, w.h.p.
Moreover, increasing or decreasing this exponent by an arbitrarily small constant term, respectively increases the hitting time by a $\poly(\dist)$ factor, or the walks never hit $\trespoint$ with probability $1-o(1)$.
For the case of $k\leq \polylog\dist$  or $k\geq \dist \polylog\dist$, all exponents $\alpha\in [3,\infty)$ or $\alpha\in (1,2]$, respectively, achieve the same optimal value of $\tau^k_\alpha(\trespoint)$ (within $\polylog\dist$ factors).
Formal statements of these results are given in \cref{sec:levywalk,sec:equivpwbw,sec:equivpwrw}.
The theorem below bounds the parallel hitting time for (near) optimal choices of $\alpha$.

\begin{theorem}
    \label{thm:plw-intro}
      Let $\trespoint$ be any node in  $\integer^2$, and $\dist = \|\trespoint\|_1$.
    \begin{enumerate}[(a)]
        \item
            If ${\log^6  \dist} \leq k \leq \dist {\log^4\dist}$, then for $\alpha = 3 - \frac{\log k}{\log\dist} + {5}\frac{\log \log\dist}{\log\dist}$,
            $\pr{\tau^k_\alpha(\trespoint)  =
            \mybigo{\frac{\dist^2 {\log^6\dist}}{k}}}
            =
            1- e^{-\omega(\log\dist)}$;

        \item
            If $k = \omega({\log^5\dist})$, then
            $\pr{\tau^k_3(\trespoint)
            = \mybigo{\dist^2}}= 1- e^{-\omega(\log\dist)}$;

        \item
            If $k = \omega(\dist {\log^2\dist})$, then
            $\pr{\tau^k_2(\trespoint)  = \mybigo{\dist}} = 1- e^{-\omega(\log\dist)}$.
    \end{enumerate}
\end{theorem}

Observe that for any given $k,\dist$ with $k = \omega(\log^5\dist)$, if we choose the exponent $\alpha$ as in \cref{thm:plw-intro}, then
\begin{equation}
    \label{eq:work-opt-alpha}
    \pr{
    \tau_\alpha^k(\trespoint)
    =
    \mybigo{(\dist^2/k)\log^6\dist + \dist}
    }
    =
    1- e^{-\omega(\log\dist)}
    .
\end{equation}

\subsubsection{Parallel \levyWalks with Random Exponents}

The right choice of $\alpha$, according to \cref{thm:plw-intro}, requires knowledge of the values of $k$ and $\dist$ (at least within polylogarithmic factors).
We propose a very simple randomized strategy for choosing the exponents of the $k$ \levywalks, which almost matches the parallel hitting time bounds of \cref{thm:plw-intro}, for all distances $\dist$ simultaneously!
The strategy does not require knowledge of $\dist$, and works as long as $k \geq \polylog\dist$.
Interestingly, it does not require knowledge of $k$ either.
The strategy is the following:
\begin{quote}
    \emph{The exponent of each walk is sampled independently and uniformly at random from the real interval $(2,3)$.}
\end{quote}
The next theorem bounds the resulting parallel hitting time $\tau^k_{\mathit rand}(\trespoint)$, for an arbitrary node $\trespoint$.
Its proof is given in \cref{sec:algorithm_analysis}.

\begin{theorem}
	\label{thm:uniform-algo-intro}
    Let $\trespoint$ by any node in  $\in\integer^2$, $\dist = \|\trespoint\|_1$, and $k  \geq \log^8 \dist$.  Then
    \begin{equation}
        \label{eq:work-rand}
        \pr{\tau^k_{\mathit rand}(\trespoint)
        =
        \mybigo{(\dist^2/k)\log^7\dist + \dist\log^3\dist}}
        =
        1-e^{-\omega(\log\dist)}
        .
    \end{equation}
\end{theorem}

By comparing \cref{eq:work-opt-alpha} and \cref{eq:work-rand}, we observe that indeed the hitting time of the randomized strategy is only by a $\polylog\dist$ factor worse than that of the deterministic strategy based on \cref{thm:plw-intro}, which knows $\dist$ and $k$.
Moreover, this hitting time is optimal within a $\polylog\dist$ factor among \emph{all} possible search strategies (deterministic or randomized) that do not know $\dist$ within a constant factor, since a universal lower bound of $\myOmega{\dist^2/k + \dist}$ with constant probability applies to all such strategies, as observed in \cite{feinerman_ants_2017}.

\subsubsection{Implications on L\'evy Hypothesis and Distributed Search}

As already mentioned, our setting of $k$ independent walks starting from the same location, aiming to hit an unknown target, can be viewed as a basic model of animals' foraging behavior around a central location, such as a nest, a food storage area, or a sheltered environment.
The assumption that walks are independent is approximately true for certain animal species such as ants Cataglyphis, which lack pheromone-based marking mechanisms~\cite{razin2013desert}.
Our results suggest that if the typical or maximum distance $\dist$ of the food (target) from the nest (source) is fixed, then a group of animals executing parallel \levywalks with the same exponent can optimize search efficiency by tuning the exponent value and/or the number $k$ of animals participating in the foraging.
In that setting, no universally optimal exponent value exists, as the optimal exponent depends on $k$ and $\dist$.
An alternative, novel approach suggested by our last result is that each animal performs a \levywalk with a randomly chosen exponent.
This strategy, which surprisingly achieves near optimal search efficiency for all distance scales, implies that different members of the same group follow different search patterns.
The existence of such variation in the search patterns among individuals of the same species requires empirical validation.

In the context of the related ANTS problem~\cite{feinerman_ants_2017}, our result on parallel \levywalks with randomly selected exponents directly implies a uniform solution to the problem (i.e., independent of $k$ and $\dist$), which is extremely simple and natural, and is optimal within $\polylog\dist$ factors, w.h.p.

\section{Related Work}

\levywalks (also referred to as \levyflights) have been studied mostly by physicists, and mainly in continuous spaces~\cite{Zaburdaev2015lw,reynolds_current_2018}.
The idea that biological organisms could perform \levywalks was first suggested in the mid 80s~\cite{shlesinger_levy_1986}, as a potentially more efficient search strategy compared to Brownian motion.
\levywalks attracted significant attention after experimental work in the mid 90s showed that albatross birds follow \levywalk-like trajectories~\cite{viswanathan_levy_1996}, a pattern that was subsequently observed for various other organisms as well~\cite{reynolds2017,Reynolds3763,humphries2010,sims2008,focardi2009,boyer2006,raichlen2014}.
Even though statistical and methodological flaws were later pointed out in several of these works~\cite{edwards_revisiting_2007}, there is currently ample evidence that many animals do exhibit \levywalk movements~\cite{viswanathan_physics_2011,Humphries7169}.

A possible explanation for this phenomenon is the \emph{\levy foraging hypothesis}~\cite{viswanathan_optimizing_1999,viswanathan_levy_2008}:
``According to the optimal foraging theory~\cite{Werner1974}, natural selection drives species to adopt the most economically advantageous foraging pattern.
Thus species must have adapted to follow \levywalks because \levywalks optimize search efficiency.''
The main theoretical argument in support of this hypothesis was provided in~\cite{viswanathan_optimizing_1999}, stating that a \levywalk with exponent $\alpha = 2$ (known as Cauchy walk) maximizes the number of visits to targets, when targets are sparse and uniformly distributed.
This result has been formally shown for one-dimensional spaces~\cite{Buldyrev2001}, but is not true for higher-dimensional spaces~\cite{levernier2020}, at least not without additional assumptions~\cite{comment2021,reply2021}.

Very recently, a new argument was provided in~\cite{guinard2020} supporting the optimality of \levywalks with $\alpha=2$.
In the considered setting, the space is a square torus of area $n$, and the walk must find a  single, randomly selected target.
Two critical model assumptions are that the target may have an arbitrary diameter $D$, and that the \levywalk is ``intermittent,'' i.e., cannot  detect the target during a jump, only at the end of the jump.
Under these assumptions, the Cauchy walk was shown to achieve a (near) optimal search time of $\tilde O(n/D)$, whereas exponents $\alpha\neq2$ are suboptimal.\footnote{Note that if the target has a fixed size $D = 1$ or the walk is not intermittent, then all exponents $\alpha\geq 2$ or $\alpha \leq 2$, respectively, are optimal as well.}

Our results add a new perspective to the \levy foraging hypothesis.
Unlike~\cite{viswanathan_optimizing_1999} and~\cite{guinard2020},
we consider a \emph{collective} search setting, where $k$ individuals start from the same source and move independently.
The space is two-dimensional as in \cite{guinard2020} (but discrete and unbounded), and there is a single target (of unit size).
If rough information about the target's distance $\dist$ to the source is known then letting all individuals execute identical \levywalks with a specific exponent, which depends on $k$ and $\dist$, achieves (near) optimal search time.
If no information on $\dist$ is available, then using a random exponent for each walk, sampled independently from the super-diffusive range $(2,3)$, still achieves near optimal search time, for all distances $\dist$.

In our analysis, we derive upper and lower bounds on the hitting time of a \levywalk on $\integer^2$.
Bounds on the hitting time and related quantities for \levywalks on the (one-dimensional) real line are given in \cite{palyulin_first_2019}.
Bounds for general random walks on $\integer^d$, for $d\geq 1$, in the case where the walk has bounded second (or higher) moments can be found in~\cite{uchiyama2011}.
Recall that \levywalks have unbounded second moment when $\alpha \leq 3$.

In the Ants-Nearby-Treasure-Search (ANTS) problem~\cite{feinerman_ants_2017}, $k$ identical (probabilistic) agents starting from the same location, search for an unknown target on $\integer^2$.
Agents do not know $k$, and cannot communicate (or see each other).
However, before the search begins, each agent receives a $b$-bit \emph{advice} from an oracle.
In~\cite{feinerman_ants_2017}, matching upper and lower bounds are shown for the trade-off between the expected time until the target is found, and the size $b$ of the advice.
The proposed  optimal algorithms   repeatedly execute the following steps: walk to a random location in a ball of a certain radius (chosen according to the algorithm specifics), perform a spiral movement of the same  radius as the ball's, then return to the origin.
Our results suggest a very simple algorithm for the setting where no advice is given ($b=0$): Each agent performs a \levywalk with a uniformly random exponent sampled from $(2,3)$.
The algorithm is Monte Carlo, and finds the target w.h.p.\ in time that is larger than the optimal by at most a polylogarithmic factor.

Variants of the ANTS problem have been studied, where agents are (synchronous or asynchronous) finite state machines, which can communicate during the execution whenever they meet~\cite{EmekLUW14,ELSUW15,CELU17,LenzenLNR17}.
Another variant, involving parallel search on the line by $k$ non-communicating agents, is considered in \cite{fraigniaud_parallel_2016}.

In~\cite{boczkowski_random_2018,GuinardK20}, tight bounds were shown for the cover time on the cycle of a random walk with $k$ different jump lengths.
The optimal walk in this case is one that approximates (using $k$ levels) a \levywalk with exponent $\alpha = 2$.

When $\alpha \in (3,\infty)$, a \levywalk on $\integer^d$ behaves similarly to a simple random walk, as the variance of the jump length is bounded.
In particular, as $\alpha \to \infty$, the \levywalk jump converges in distribution to that of a simple random walk.
Parallel independent simple random walks have been studied extensively on finite graphs, under various assumptions for their starting positions~\cite{alon_many_2011,efremenko_how_2009,ElsasserSauwervald2011,IvaskovicKPS17,kanade_coalescence_2019}.
A main objective of that line of work has been to quantify the ``speedup'' achieved by $k$ parallel walks on the cover time, hitting times, and other related quantities, compared to a single walk.

The following basic network model has been proposed by Kleinberg to study the small world phenomenon~\cite{kleinberg_small_world_2000}.
A square (or, more generally, $d$-dimensional) finite lattice is augmented by adding one ``long-range'' edge from each node $u$, to a uniformly random node $v$ among all nodes at lattice distance $k$, where distance $k$ is chosen independently for each $u$, from a power-law distribution with exponent $\alpha$.
That is, the distribution of long-range edges is the same as the jump distribution of a \levywalk with the same exponent.
It was shown that (distributed) greedy routing is optimized when $\alpha = 1$,
whereas for $\alpha \neq 1$ the expected routing time is slowed down by polynomial factors~\cite{kleinberg_small_world_2000}.\footnote{In the paper, the exponent considered is that of choosing the endpoint of $u$'s long-range link to be a given node $v$ at distance $k$, which is proportional to $1/k^\beta$, where $\beta = \alpha + d-1$. Thus the optimal exponent is $\beta = 2$ for the square lattice, and $\beta = d$ for the $d$-dimensional lattice.}
This result is of similar nature as our result for the hitting time of $k$ identical \levywalks, where exactly one exponent is optimal.
However, in our case, this exponent depends on the target distance.
In Kleinberg's network, exponent $\alpha=1$ ensures that the lengths of long-range links are uniformly distributed over all \emph{distance scales}, which facilitates fast routing.
In our randomized strategy, availability of a sufficient number of walks with the right exponent is achieved by choosing the exponents uniformly at random over the interval $(2,3)$.

\subsection*{Roadmap}

The Supplementary Information is organized as follows. In \cref{sec:preliminaries}, we give some basic definitions and facts, and investigate a monotonicity property that plays a key role in our analysis. In \cref{sec:levywalk}, we provide the analysis of 
the regime $\alpha\in(2,3]$
while in \cref{sec:equivpwbw,sec:equivpwrw}, we investigate 
the regimes $\alpha\in(1,2]$ and $\alpha\in(3,\infty)$,
respectively.
Finally, in \cref{sec:algorithm_analysis}, 
we use results from \cref{sec:levywalk} to analyze the efficiency of our simple distributed search algorithm.

\section{Preliminaries}
    \label{sec:preliminaries}
    
\subsection{Technical notation}

Throughout the analysis, we make use of the conventional Bachmann–Landau notation for asymptotic behaviors of function, which we now recall. 
Let $f:\real  \rightarrow \real$ and $g:\real \rightarrow \real$ be any two functions.
We write $f(x) = \mybigo{g(x)}$ if a constant $M>0$ and a value $x_0 \in \real$ exist such that $\abs{f(x)} \le M g(x)$ for any $x > x_0$. Similarly, we write $f(x) = \myOmega{g(x)}$ if a constant $ m >0 $ and a value $ x_0 \in \real $ exist such that $ \abs{f(x)} \ge m g(x) $ for any $ x > x_0 $; finally, we write $ f(x) = \myTheta{g(x)} $ if  two constants $ 0 < m < M $ and a value $ x_0 \in \real $ exist such that $ m g(x) \le \abs{f(x)} \le Mg(x) $ for all $ x > x_0 $.
Moreover, we write $ f(x) = \mylittleo{g(x)} $ if $ \lim_{x \to \infty} \frac{f(x)}{g(x)} = 0 $, and $ f(x) = \mylittleomega{g(x)} $ if $ \lim_{x \to \infty} \frac{f(x)}{g(x)} = \infty$.
We complete this subsection by mentioning the $ \polylog $ function. By writing $ f(x) = \polylog x $, we mean that a constant $ m > 0 $ exists such that $ f(x) = \myTheta{\log^m x} $.

\subsection{Main definitions and notation}
For each point $s = (x,y)\in \real^2$, we write $\|s\|_p$ to denote its $p$-norm $(|x|^p + |y|^p)^{1/p}$.
The $p$-norm distance between points $s=(x,y)$ and $s' = (x',y')$ is $\|s-s'\|_p =(|x - x'|^p + |y-y'|^p)^{1/p}$.

We consider the infinite grid graph  $G = (\integer^2,E)$, where $E = \{\{u,v\} \colon \|u-v\|_1=1\}$.
The shortest-path distance between two nodes $u,v\in \integer^2$ in $G$ equals $\|u-v\|_1$.
In the following, we will say just \emph{distance} to refer to the shortest-path distance.

We denote by $R_d(u)$ the set of all nodes $v\in\integer^2$ that are at distance exactly $d$ from $u$, i.e.,
$R_d(u)
=
\{v\in\integer^2\colon \|u-v\|_1 = d \}
$.
We also define
$B_d(u) =
\{v\in\integer^2\colon \|u-v\|_1 \leq d \}$
and
$Q_d(u)
=
\{v\in\integer^2\colon \|u-v\|_\infty \leq d \}$.
See \cref{fig:thetworhombus} for an illustration.

\begin{figure}[t]
    \centering
    \includegraphics[scale=0.75]{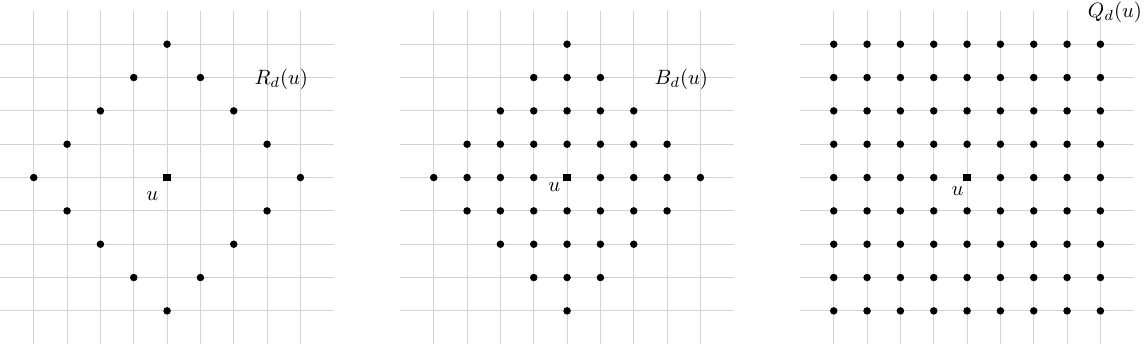}
     \caption{Illustrations of $R_d(u)$, $B_d(u)$, and $Q_d(u)$, for $d=4$.}
    \label{fig:thetworhombus}
\end{figure}

By $\overline{uv}$ we denote the straight-line segment on the real plane $\real^2$ between nodes $u$ and $v$.
A \emph{direct-path} between $u$ and $v$ in $G$ is a shortest path that ``closely follows'' the real segment $\overline{uv}$.
See \cref{fig:approximatingpath} for an illustration.

\begin{figure}[t]
	\centering
	\includegraphics[scale=0.75]{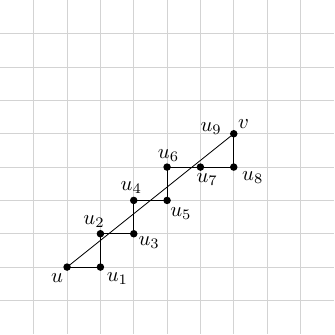}
	\caption{Example of a line segment $\overline{uv}$ and the direct-path between $u$ and $v$.}
	\label{fig:approximatingpath}
\end{figure}

\begin{definition}(Direct-path)
    \label{def:straight-path}
    A \emph{direct-path} from node $u$ to $v$ is a shortest path $u,u_1,\ldots,u_k=v$, where $k = \|u-v\|_1$, and for each $1\leq i < k$, $u_i \in R_{i}(u)$ and $\|u_i - w_i\|_2 = \min_{v'\in R_{i}(u)} \|v' - w_i\|_2$, where $w_i$ is the (unique) point $w$ in the real segment $\overline{uv}$ with $\|u-w\|_1 = i$.
\end{definition}

It is not hard to verify that $u,u_1,\ldots,u_k$ is indeed a path of $G$.
Also, unlike point $w_i$, node $u_i$ is not necessarily unique, since there may be two different nodes in $R_{i}(u)$ that are closest to $w_i$.
The next simple fact about direct-paths is proved in \cref{app:preliminaries}.

\begin{lemma}
    \label{lemma:direct-path}
    Let $u\in \integer^2$ and $d\geq 1$ be an integer.
    Suppose we sample a node $v$ uniformly at random from set $R_d(u)$, and then sample a direct-path $u,u_1,\ldots,u_d=v$ from $u$ to $v$ uniformly at random among all such paths.
    Then, for every $1\leq i < d$ and $w\in R_i(u)$, \[
        \frac{(i/d)\cdot\lfloor d/i\rfloor}{4i}
        \leq
        \pr{u_i = w}
        \le
        \frac{(i/d)\cdot\lceil d/i\rceil}{4i}
        .
    \]
\end{lemma}

A (discrete-time) \emph{jump process} on $\integer^2$ is just an infinite sequence of random variables $(J_t)_{t\geq 0}$, where $J_t \in \integer^2$ for each integer $t \geq 0$.
We say that the process \emph{visits} node $v\in \integer^2$ at step $t\geq 0$ if $J_t = u$.
Our analysis will focus on the following two jump processes.

\begin{definition}[\levyflight]
    \label{def:levy-filght}
    A jump process $L^f = (\levyflightrand{t})_{t\ge0}$ on $\integer^2$ is a \levyflight with exponent parameter $\alpha\in(1,\infty)$, and start node $s\in\integer^2$, if $\levyflightrand{0} = s$, and for each $t\geq 0$, if $\levyflightrand{t} = u$ then $\levyflightrand{t+1} = v$, where node $v\in\integer^2$ is selected as follows:
First a jump distance $d$ is chosen independently at random such that
    \begin{equation}
        \label{eq:power-law-distr}
        \pr{d = 0} = 1/2, \text{ and } \pr{d = i} = c_{\alpha}/i^{\alpha} \text{ for } i\geq 1,
    \end{equation}
    where $c_{\alpha} $ is a normalizing factor.
    Then, node $v$ is chosen independently and uniformly at random among all nodes in $R_{d}(u)$ (i.e., all nodes at distance $d$ from $u$).
\end{definition}

\begin{definition}[\levywalk]
    \label{def:levy-walk}
    A jump process $L^w = (\levyrand{t})_{t\ge0}$ on $\integer^2$ is a \levywalk with exponent parameter $\alpha\in(1,\infty)$, and start node $s\in\integer^2$, if $\levyrand{0} = s$, and the process consists of a infinite sequence of \emph{jump-phases}, where each jump-phase is defined as follows:
    Suppose that the jump-phase begins at step $t+1$ (the first jump-phase begins at the first step), and suppose also that $\levyrand{t} = u$.
    First a distance $d$ and a node $v$ at distance $d$ from $u$ are chosen, in exactly the same way as in the \levyflight.
    If $d = 0$ then the jump-phase has length 1, and $\levyrand{t+1} = u$, i.e., the process stays put for one step.
    If $d \neq 0$ then the jump-phase has length $d$, and in the next $d$ steps the process follows a path $u,u_1,\ldots,u_d=v$ chosen uniformly at random among all direct-paths from $u$ to $v$, i.e., $\levyrand{t+i} = u_i$, for all $1\leq i\leq d$.\footnote{Our analysis works also if an arbitrary direct-path between $u$ and $v$ is selected, instead of a random one.}
\end{definition}

We observe that a \levyflight is a Markov chain, while a \levywalk is not.

\begin{remark}\label{remark:alpha>1+eps}
	Throughout the analysis, we assume $ \alpha $ to be a (not necessarily constant) real value greater than $ 1+\epsilon $ for some arbitrarily small constant $ \epsilon>0 $.
\end{remark}
We will often use the following  bound on the  tail distribution of the jump length $d$ chosen according to \eqref{eq:power-law-distr}:
\begin{equation}
    \label{eq:jumpatleast}
    \pr{d \geq i}
    =
    \myTheta{{1}/{i^{\alpha - 1}}}
\end{equation}
The next statement follows immediately from \cref{lemma:direct-path}.

\begin{corollary}
    \label{cor:visit-direct-path}
    Let $u,v\in \integer^2$, and $d = \|u-v\|_1 > 0$.
    If a \levywalk is at node $u$ at the beginning of a jump-phase, then the probability it visits $v$  during the jump-phase is $\myTheta{1/ d^\alpha}$.
\end{corollary}

\begin{definition}[Hitting Time]
    \label{def:search-problem}
    The \emph{hitting time} for node $\trespoint\in \integer^2$ of a jump process is the first step $t\geq 0$ when the process visits u.
    For a set of $k$ independent jump processes that run in parallel, their \emph{parallel hitting time} for $\trespoint$ is the first step in which  some (at least one) of the $k$ processes visits $\trespoint$.
\end{definition}

We will denote by $\tau_\alpha(\trespoint)$ the hitting time for  $\trespoint$ of a single \levywalk processes with exponent $\alpha$ starting from the origin $\origin = (0,0)$; and by $\tau^k_\alpha(\trespoint)$ the parallel hitting time for  $\trespoint$ of $k$ independent copies of the above  \levywalk.
Unless stated otherwise, we will always assume that the starting node of a jump processes  is the origin $\origin = (0,0)$.

For a \levyflight $L^f = (\levyflightrand{t})_{t\ge0}$, we denote by $\levyflightvisits{u}{t}$ the number of times the process visits node $u\in\integer^2$ until step $t$,  i.e., $\levyflightvisits{u}{t} = |\{i \colon \levyflightrand{i} = u\}\cap\{1,\ldots,t\}|$.
We define $\levywalkvisits{u}{t}$ similarly for a \levywalk.

\subsection{Bounds via Monotonicity}
\label{sec:monotonicity-bounds}

We will now describe an intuitive monotonicity property that applies to a family of jump processes that includes \levyflights (but not \levywalks).
We then use this property, and the similarity between \levy flights and walks, to show upper bounds on the probability that a \levywalk visits a given target.
We start by defining the family of \emph{monotone radial} jump processes.

\begin{definition}[Monotone radial process]
    \label{def:radial-monotone}
    A jump process $(J_t)_{t\geq 0}$
    is \emph{monotone radial} if,
for any pair of nodes $u,v\in \integer^2$,
    and any $t\geq 0$,
$\pr{J_{t+1} = v \mid J_{t} = u} = \rho(\|u-v\|_1)$,
    for some non-increasing function $\rho$.
\end{definition}

Clearly, \levyflights are monotone radial processes.
For all such processes, we use geometric arguments to prove the following property. 
The proof is deferred to \cref{app:monotonicity}.

\begin{lemma}[Monotonicity property]
    \label{lemma:monotonicity}
   Let  $(J_i)_{i\geq 0}$ be any monotone radial jump process.    For  every pair $u,v \in \integer^2$ and any $t\geq 0$, if $\|v\|_\infty \geq \|u\|_1$  then
    $\pr{J_t = u} \geq \pr{J_t = v}$.
\end{lemma}

Next, we use \cref{lemma:monotonicity} to upper bound the probability that a target node is visited during a given jump-phase of a \levywalk, and then bound the probability that the target is visited during \emph{at least one} jump-phase.

\begin{lemma}
    \label{lemma:prob_anyjumpfindstreasure}
Let $\trespoint$ be an arbitrary node with  $\dist = \|\trespoint\|_1$. Let $ \epsilon > 0 $ be an arbitrarily small constant.
	For any $\alpha \geq 1 + \epsilon$, the probability that a \levywalk visits $\trespoint$ during its $i$-th jump-phase is
	$\mybigo{\mu\cdot(\dist^{-2} + \dist^{-\alpha})}$ if $ \alpha \neq 2$, where $ \mu = \min \{\log \dist, \abs{\frac{1}{2 - \alpha}}\} $, and $\mybigo{\log\dist /\dist^2}$
	if $\alpha=2$.
\end{lemma}
\begin{proof}
For any $v\in B_{\dist/4}(\trespoint)$, the probability that the $i$-th jump starts in $v$ is at most $\bigo(1/\dist^2)$ due to \cref{lemma:monotonicity}, since the process restricted only to the jumps endpoints is a \levyflight.
Moreover, for any  $1 \le d \le \dist/4$, there are at most $4d$ nodes in $B_{\dist/4}(\trespoint)$ located at distance $d$ from $\trespoint$. Then, from the chain rule and \cref{cor:visit-direct-path}, the probability that the $i$-th jump starts from $B_{\dist/4}(\trespoint)$ and the agent visits the \treasure during the jump-phase is bounded by
\[
\bigo\left(\frac{1}{\dist^2}\right)\sum_{d = 1}^{\dist/4} 4d\cdot\bigo\left(\frac{1}{d^\alpha}\right) + \bigo\left(\frac{1}{\dist^2}\right),
\]
the term $\bigo(1/\dist^2)$ being the contribution of $\trespoint$ itself. The above expression equals $ \mybigo{\mu \dist^{2 -\alpha} /\dist^2} $ if $ \alpha \neq 2$, and $ \mybigo{\log \dist / \dist^2} $ if $ \alpha = 2 $.
For any fixed node $v $, denote by $F_i$ event that that, during a jump-phase starting in $v$, the agent   visits the \treasure. We now prove that $\pr{F_i} = \mybigo{\mu / \dist^\alpha}$.

Let $V_i$ be the event that the starting point of the $i$-th jump is in $B_{\dist/4}(\trespoint)$. Notice that $\pr{F_i \mid \overline{V_i}}$ is at most $\mybigo{1 /  \dist^{\alpha} }$ for \cref{cor:visit-direct-path}. Then,
\begin{align*}
	\pr{F_i} & \le \pr{F_i \mid V_i}\pr{V_i} + \pr{F_i \mid \overline{V_i}} \le \pr{F_1 \mid V_i} + \mybigo{\frac{1}{\dist^\alpha}} .
\end{align*}
Then, if $ \alpha > 2 $,
\[
\pr{F_i} = \mybigo{\frac{\mu}{\dist^2} + \frac{1}{\dist^\alpha}}  = \mybigo{\frac{\mu}{\dist^2}}.
\]
If $ \alpha = 2 $, $ \pr{F_i} = \mybigo{\log \dist / \dist^2} $. And if $ 1 < \alpha < 2 $,
\[
\pr{F_i} = \mybigo{\frac{\mu}{\dist^\alpha} + \frac{1}{\dist^\alpha}} = \mybigo{\frac{\mu}{\dist^\alpha}}.
\qedhere
\]
\end{proof}

\begin{lemma}
	\label{lemma:prob_never_find_target}
	Let $\trespoint$ be an arbitrary node with  $\|\trespoint\|_1 = \dist$. Let $ \epsilon > 0 $ be an arbitrarily small constant.
	The probability that a \levywalk with exponent $1 + \epsilon \le \alpha < 3$ visits $\trespoint$ at least once (at any step $t$) is
	$ \mybigo{\mu \log \dist(\dist^{-1} + \dist^{3 - \alpha})} $ if $ \alpha \neq 2 $, where $ \mu = \min\{\log \dist, \abs{\frac{1}{2-\alpha}}\} $, and $ \mybigo{\frac {\log^2 \dist} {\dist}} $ if $\ \alpha = 2 $.
\end{lemma}
\begin{proof}

	For each $i\ge 0$, consider  the first time $t_i$ the agent is at distance at least $\lambda_i = 2^i \dist$ from the origin. From \cref{eq:jumpatleast}, the probability any jump has length no less than $2\lambda_i$ is at least $ c/\lambda_i^{\alpha - 1} $, for some constant $c>0$. Define, for $i\ge 1$, the values $\tau_i = 2c^{-1}\lambda_i^{\alpha-1}\log \lambda_i	$. The probability that $t_i \ge n\tau_i$ is bounded from above by the probability that no jump between the first $n\tau_i$ jumps has length at least $2\lambda_i$. 
	Since the jump lengths are mutually independent, from \cref{eq:jumpatleast} we get 
	\begin{align*}
		\pr{t_i \ge n\tau_i} & \le  \left[1- \frac{c}{\lambda_i^{\alpha-1}}\right]^{\frac{2}{c}n\lambda_i^{\alpha-1}\log\lambda_i} \\
		& \le \exp \left (-2n \log \lambda_i\right ) \\
		& = \frac{1}{\lambda_i ^{2n}} = \frac{1}{2^{2ni}\dist^{2n}},
	\end{align*}
	where we have used the well known inequality
	$ 1 - x \le e^{-x} $ that holds for every real  $ x  $.
	We next bound   the expected number of visits to the \treasure node from time $t_i$ to time $t_{i+1}$ as follows:
	\begin{align}
		\label{eq:expression}
		& \expect{\levywalkvisits{\trespoint}{t_{i+1}} - \levywalkvisits{\trespoint}{t_i}}\notag\\
		\le \ & \expect{\levywalkvisits{\trespoint}{t_{i+1}} - \levywalkvisits{\trespoint}{t_i} \mid t_{i+1} \le \tau_{i+1}}\pr{t_{i+1} \le \tau_{i+1}} \notag\\
		+ \ & \sum_{n \ge 1} \expect{\levywalkvisits{\trespoint}{t_{i+1}} - \levywalkvisits{\trespoint}{t_i} \mid n\tau _{i+1} < t_{i+1} \le (n+1)\tau_{i+1}}\cdot \notag\\
		\cdot \ & \pr{t_{i+1} \ge n\tau_i}.
	\end{align}
	
	We now proceed by analysing  three different ranges for the exponent $\alpha$ and use \cref{lemma:prob_anyjumpfindstreasure}. Notice that the agent starts at distance $\Omega\left(\lambda_i\right)$ from the target.
	
	If $ \alpha > 2 $, the expression in \cref{eq:expression} is equal to
	\begin{align*}
		& \mybigo{\frac{\mu \tau_{i+1}}{\lambda_i^2}} + \sum_{n \ge 1} \mybigo{\frac{\mu(n+1)\tau_{i+1}}{2^{2ni}\dist^{2n}\lambda_i^2}} = \mybigo{\frac{\mu\tau_i}{\lambda_i^2}}
	\end{align*}
	since the sum $ \sum_{n\ge 1 } (n+1)\left (2^i \dist\right )^{-2n} $ is less than a constant.
If $ \alpha = 2 $, the expression in \cref{eq:expression} is
	\begin{align*}
		& \mybigo{\frac{\tau_{i+1}\log \lambda_i}{\lambda_i^2}} + \sum_{n \ge 1} \mybigo{\frac{(n+1)\tau_{i+1}\log \lambda_i}{2^{2ni}\dist^{2n}\lambda_i^2}} = \mybigo{\frac{\tau_i \log \lambda_i}{\lambda_i^2}}.
	\end{align*}
	And if $ 1 < \alpha < 2 $, the expression is
	\begin{align*}
		& \mybigo{\frac{\mu \tau_{i+1}}{(2 - \alpha)\lambda_i^\alpha}} + \sum_{n \ge 1} \mybigo{\frac{\mu (n+1)\tau_{i+1}}{2^{2ni}\dist^{2n}\lambda_i^\alpha}} = \mybigo{\frac{\mu\tau_i}{\lambda_i^\alpha}}.
	\end{align*}
	The same bounds with $ i = 1 $ hold for the expected number of visits to the \treasure until time $t_1$.
	From the facts above, for $\alpha > 2$,   the expected total number of visits to the \treasure is bounded by
	\begin{align*}
		& \mybigo{\frac{\mu\tau_1}{ \lambda_1^2}} + \sum_{i \ge 1}\mybigo{\frac{\mu\tau_i}{ \lambda_i^2}} \\
		= \ & \mybigo{\frac{\mu \log \dist}{\dist^{3 - \alpha}}} + \sum_{i \ge 1} \mybigo{\mu \cdot \frac{\log(2^i) + \log \dist}{(2^{2i(3 - \alpha)} \dist^{3 - \alpha})}} \\
		= \ & \mybigo{\frac{\mu\log \dist}{\dist^{3 - \alpha}}}.
	\end{align*}
	Similarly, for $ \alpha = 2 $, we obtain the bound $ \mybigo{\log^2 \dist / \dist^{3 - \alpha}} $, while,  for $ 1 < \alpha < 2 $, we get
	\begin{align*}
		& \mybigo{\frac{\mu\log \dist}{\dist}} + \sum_{i \ge 1} \mybigo{\mu \cdot \frac{\log(2^i) + \log \dist}{(2^{2i(3 - \alpha)} \dist)}}
		= \mybigo{\frac{\mu\log \dist}{\dist}}.
	\end{align*}
	Finally, we bound the probability the agent visits the \treasure at least once using the   Markov's inequality.
\end{proof}

 \section{The Case \texorpdfstring{$\alpha\in(2,3]$}{alpha in (2,3]}}
\label{sec:levywalk}

In this section, we analyze the hitting time of \levywalks when the  exponent
parameter $\alpha$ belongs to the range $(2,3]$.
In this case, the jump length has bounded mean and unbounded variance.

Recall that $\tau_\alpha(\trespoint)$ is the hitting time for target $\trespoint$ of a single \levywalk with exponent $\alpha$, and $\tau^k_\alpha(\trespoint)$ is the parallel hitting time for  $\trespoint$ of $k$ independent copies of the above  \levywalk.
All walks start from the origin.

We will prove the following bounds on the hitting time $\tau_\alpha$.

\begin{theorem}
	\label{thm:lw23}
    Let $\alpha\in (2,3)$, $\trespoint\in\integer^2$, and $\dist = \|\trespoint\|_1$.
    Let $\mu = \min\{\log\dist, \frac{1}{\alpha - 2}\}$, $ \nu = \min \{\log \dist, \frac{1}{3 - \alpha}\}$, and $\gamma = \frac{(\log\dist)^{\frac{2}{\alpha - 1}}} {(3-\alpha)^2}$.
    Then:
    \begin{enumerate}[(a)]
        \item \label{thm:lw23:a}
            $ \pr{\tau_\alpha(\trespoint) = \mybigo{\mu\cdot \dist^{\alpha - 1}}}
            =
            \myOmega{1/(\gamma \cdot \dist^{3-\alpha})}$,
            if $3 - \alpha  = \omega(1 / \log \dist)$;

        \item \label{thm:lw23:b}
            $ \pr{\tau_\alpha(\trespoint) \leq t } = \mybigo{\mu \nu \cdot {t^2} /{\dist^{\alpha+1}} }$, for any step $\dist \le t = \mybigo{\dist^{\alpha - 1} / \nu}$;

        \item \label{thm:lw23:c}
            $\pr{\tau_\alpha(\trespoint) < \infty} = \mybigo{\mu \cdot{\log\dist}/  {\dist^{3-\alpha}}}$.
    \end{enumerate}
\end{theorem}

Using \cref{thm:lw23}, we can easily obtain the following bounds on the parallel hitting time $\tau^k_\alpha$.
The proof is given in \cref{sec:cor:plw23}.

\begin{corollary}
	\label{cor:plw23}
    Let  $\trespoint\in\integer^2$ and  $\dist = \|\trespoint\|_1$, and  let $k$ be any integer such that  $\log^6\dist \leq k \leq \dist \log ^4 \dist$.
    Let $\alpha^\ast = 3 - \frac{\log k}{\log \dist}$, and $ \overline{\alpha} = \max\{2,\alpha^\ast - 4 \frac{\log \log \dist}{\log \dist}\} $. Then:
    \begin{enumerate}[(a)]
        \item \label{cor:plw23:a}
            For $\alpha = \alpha^\ast + 5\frac{\log\log\dist}{\log \dist}$,
            $\pr{\tau_{\alpha}^k(\trespoint) = \mybigo{\dist^2 \log^6 \dist / k}}
            =
            1 - e^{-\mylittleomega{\log \dist}}$;
        \item \label{cor:plw23:b}
            For $\overline{\alpha} < \alpha <3$,
            $\pr{\tau_\alpha^k(\trespoint) \leq (\dist^2/k)\cdot \dist^{(\alpha-\alpha^\ast)/2} / \log^4 \dist}
            =
            o(1)$;  

        \item \label{cor:plw23:c}
            For $2 < \alpha \leq  \alpha^\ast$,
            $\pr{\tau_\alpha^k(\trespoint) < \infty} =  \mybigo{{\log^2\dist}/  {\dist^{\alpha^\ast-\alpha}}}$.
    \end{enumerate}
\end{corollary}

\cref{cor:plw23}\cref{cor:plw23:a} states that $\tau_{\alpha}^k(\trespoint) \leq (\dist^2 / k)\cdot \polylog \dist $, w.h.p., when $\alpha \approx \alpha^\ast = 3 - {\log k}/{\log \dist}$.
One the other hand, \cref{cor:plw23}\cref{cor:plw23:b} says that the lower bound $\tau_{\alpha}^k(\trespoint) \geq (\dist^2/k) \cdot \dist^{(\alpha-\alpha^\ast)/2} / \polylog \dist$ holds with probability $1-o(1)$
for any $\alpha \geq \alpha^\ast - \mybigo{\log\log\ell/\log\ell}$, and \cref{cor:plw23}\cref{cor:plw23:c} says that the target is never hit with probability at least $1 - \polylog\ell /  {\dist^{\alpha^\ast-\alpha}}$, if $\alpha \leq \alpha^\ast$.
Therefore, the optimal hitting time of $\dist^2 / k$ is achieved (modulo $\polylog\ell$ factors) only for values of $\alpha$ very close to $\alpha^\ast$, precisely only if $|\alpha - \alpha^\ast| = \mybigo{\log\log\ell/\log\ell}$.

For the threshold case  $\alpha = 3$, similar bounds to those of  \cref{thm:lw23} apply,  modulo some $\polylog\dist$ factors, as stated in the next theorem.

\begin{theorem}
	\label{thm:lw3}
	 Let $\trespoint\in \integer^2$ and $\dist = \|\trespoint\|_1$. Then:
	\begin{enumerate}[(a)]
		\item \label{thm:lw3:a}
		$\pr{\tau_3(\trespoint) = \mybigo{\dist^2}} = \myOmega{1 / \log^4 \dist}$;
		\item \label{thm:lw3:b}
		$ \pr{\tau_3(\trespoint) \leq t} = \mybigo{t^2 \log \dist /\dist^4}$, for any step $t$ with $\dist \le t = \mybigo{\dist^2/ \log \dist}$.
	\end{enumerate}
\end{theorem}

The next corollary gives bounds on the parallel hitting time $\tau_{3}^k$.
The proof is similar to that of \cref{cor:plw23}.

\begin{corollary}
	\label{cor:plw3}
	   Let $\trespoint\in \integer^\ast$ and $\ell = \|\trespoint\|_1=\dist$, and  let $k$ be any integer such that $  \mylittleomega{\log^5 \dist} \le k \le \dist^2 / \log^2 \dist $.
Then:
	\begin{enumerate}[(a)]
		\item \label{cor:plw3:a}
		$\pr{\tau_3^k(\trespoint) = \mybigo{\dist^2}} = 1 - e^{-\mylittleomega{\log \dist}}$.
\item \label{cor:plw3:b}
		$ \pr{\tau_3^k(\trespoint) \leq \dist^2/ \sqrt k} = \mylittleo{1}$.
	\end{enumerate}
\end{corollary}

\cref{cor:plw3}\cref{cor:plw3:a} says that $\tau_3^k(\trespoint) = \mybigo{\dist^2}$, w.h.p., for any $k \geq \polylog\ell$, and \cref{cor:plw3}\cref{cor:plw3:b} provides a very crude lower bound indicating that increasing $k$ beyond $\polylog\ell$, can only result in sublinear improvement.

\subsection{Proof of Theorems \ref{thm:lw23} and \ref{thm:lw3}}
\label{ssec:mainlevy}

We will use the following key lemma, which is shown in  \cref{ssec:analysislevyflight}.
The lemma provides an upper bound on the hitting time of a \levyflight, assuming the maximum jump length is capped to some appropriate value.

\begin{lemma}
    [\levyflight with {$\alpha \in (2,3]$}]
    \label{lemma:hittinglevyflight}
    Let $h_f$ be the hitting time of a \levyflight
for target $\trespoint$ with $\|\trespoint\|_1=\dist$.
    Let $\EE_t$ be the event that each of the first $t$ jumps has length less than $(t\log t)^{1/(\alpha-1)}$.
    Then, there is a $t = \myTheta{\dist^{\alpha - 1}}$ such that:
    \begin{enumerate}[(a)]
        \item $ \pr{h_f \leq t \Mid \EE_t} = \myOmega{1/(\gamma \dist^{3 - \alpha})}$,
        if  $2 < \alpha  \leq 3-\omega(1 / \log \dist)$, where $ \gamma = \frac{(\log \dist)^{\frac{2}{\alpha - 1}}}{(3 - \alpha)^2} $;
        \item $ \pr{h_f \leq t \Mid \EE_t} = \myOmega{1/\log^4 \dist}$, if $\alpha =3$.
    \end{enumerate}
\end{lemma}

The second lemma we need is an upper bound on the hitting time of a \levywalk in terms of the hitting time of the capped \levyflight considered above.
The proof proceeds by coupling the two processes, and is given in \cref{ssec:tech_contrib_coupling}.

\begin{lemma}
\label{lemma:couplevyflightintowalk}
    Let $h_f$  and $\EE_t$ be defined as in \cref{lemma:hittinglevyflight}, and let $\tau_\alpha (\trespoint)$ be the hitting time of a \levywalk with the same exponent $\alpha$,
    for the same target $\trespoint$. Then, for any $2 <\alpha \leq 3$ and step $t$, and for $ \mu = \min\{\log \dist, \frac{1}{\alpha - 2}\} $,
    \[
        \pr{\tau_\alpha (\trespoint) = \mybigo{\mu t}}
        \geq
        \left (1 - \mybigo{{1}/{\log t}}\right )
        \cdot\left (
            \pr{h_f \leq t \Mid \EE_t}
            - e^{-t^{\myTheta{1}}}
        \right )
        .
    \]

\end{lemma}

The last lemma we need is the following lower bound on the hitting time of a \levywalk, proved in \cref{sec:lowbound_superdiffusive}.

\begin{lemma}
    \label{lemma:lowbound_superdiffusive}
Let $\trespoint\in\integer^2$ and $\ell= \manhattan{\trespoint}$.
    For any step $t\geq \dist$,
	\begin{enumerate}[(a)]
        \item
        $\pr{\tau_\alpha (\trespoint) \leq t}
        =
        \mybigo{ \frac{\nu \mu t^2}{\dist^{\alpha + 1}}}$ if $ \alpha \neq 3 $ and $ t = \mybigo{\dist^{\alpha - 1}/\nu} $, where $ \nu = \min \{\log \dist, \frac{1}{3 - \alpha}\} $ and $ \mu = \min\{\log \dist, \frac{1}{\alpha - 2}\} $;
        \label{lemma:lowbound_superdiffusive:a}

        \item
        $\pr{\tau_\alpha (\trespoint) \leq t}
        =
	    \mybigo{\frac{t^2\log\dist}{\dist^{\alpha+ 1}}}$
        if $\alpha = 3$ and $ t = \mybigo{\dist^2/\log\dist} $.
        \label{lemma:lowbound_superdiffusive:B}
	\end{enumerate}
\end{lemma}

\begin{proof}[\bf Proof of \cref{thm:lw23}]
    \cref{lemma:hittinglevyflight,lemma:couplevyflightintowalk} imply that, for some $ t = \myTheta{\dist^{\alpha - 1}} $,
    \[
    	\pr{\tau_\alpha(\trespoint) = \mybigo{\mu t}} = \myOmega{\frac{(3-\alpha)^2}{\dist^{3-\alpha}\log^{\frac{2}{\alpha - 1}}\dist}},
    \]
    obtaining (a).
Parts (b) and (c) follow from \cref{lemma:lowbound_superdiffusive} and \cref{lemma:prob_never_find_target}, respectively.
\end{proof}

\begin{proof}[\bf Proof of \cref{thm:lw3}]
    The proof proceeds in exactly the same way as the proof of \cref{thm:lw23}, using \cref{lemma:hittinglevyflight,lemma:couplevyflightintowalk} to show  \cref{thm:lw3:a}, and \cref{lemma:prob_never_find_target} to show \cref{thm:lw3:b}.
\end{proof}

\subsection{Proof of Lemma \ref{lemma:hittinglevyflight}} 

\label{ssec:analysislevyflight}

We first define some notation.
Then we give an overview of the analysis, before we provide the detailed proof.

Recall that $(\Lf{i})_{i\geq0}$ denotes the \levyflight process, and $\Zf{u}{i} = |\{j \colon \levyflightrand{j} = u\}\cap\{1,\ldots,i\}|$ is the number of visits to node $u$ in the first $i$ steps.
Let $t = \myTheta{\dist^{\alpha - 1}}$
be a step to be fixed later.
For each $i\geq 1$, let $S_i$ be the length of the $i$-th jump of the \levyflight, i.e.,
\[
    S_i = \|\Lf{i}-\Lf{i-1}\|_1.
\]
Define also the events
\[
    E_i = \{ S_i \le (t\log t)^{{1}/{(\alpha-1)}}\},
\]
and let $\EE_i = \bigcap_{j=1}^i E_j$.
For each node $u$ and $i\geq 0$, let
\[
    p_{u,i} \ = \pr{\Lf{i} = u \ \middle|\ \EE_i},
\]
and note that
$
\expect{\Zf{u}{i} \Mid \EE_i}
    =
    \sum_{j=0}^i p_{u,j}.
$
We  partition the set of nodes into disjoint sets $\AA_1,\AA_2,\AA_3$ defined as follows:
\[
  \begin{gathered}
    \firstregion = \{v \colon \|v\|_\infty \le \dist\}
\\
    \AA_2 =
    \begin{cases}
        \{v \colon \manhattan{v}
        \le 2(t\log t)^{{1}/{(\alpha-1)}}\}
        \setminus \AA_1
            &\mbox{if } \alpha \in (2,3) \\
        \{v \colon \manhattan{v}
        \le 2\sqrt t\log t\}
        \setminus \AA_1
            & \mbox{if } \alpha = 3
    \end{cases}
\\
    \thirdregion = \integer^2\setminus (\AA_1\cup\AA_2).
  \end{gathered}
\]

\subsubsection{Proof Overview}
\label{sec:overview}

We discuss just the case of $\alpha\in(2,3)$; the case of $\alpha= 3$ is similar.
We assume all probability and expectation quantities below are conditional on the event $\EE_t$, and we omit writing this conditioning explicitly.

First, we show a simple upper bound on the mean number of visits to $\AA_1$ until step $t = \myTheta{\dist^{\alpha - 1}}$, namely,
\[
\sum_{v\in \AA_1} \expect{\Zf{v}{t}} \le ct
,
\]
for a constant $c < 1$: with constant probability the walk visits a node outside $\AA_1$ in the first $t/2$ steps, and after that at most a constant fraction of steps visit nodes in $\AA_1$, by symmetry.

To bound the mean number of visits to $\AA_2$, we use the monotonicity property from \cref{sec:monotonicity-bounds} and the fact that
$\|v\|_\infty \geq \dist = \|\trespoint\|_1$ for all $v\in\AA_2$, to obtain
\[
\sum_{v\in \AA_2}
\expect{\Zf{v}{t}}\le
|\AA_2|\cdot
\expect{\Zf{\trespoint}{t}}\leq
4(t\log t)^{1/(\alpha-1)}\cdot
\expect{\Zf{\trespoint}{t}}.
\]

For the number of visits to $\AA_3$ we obtain the following bound using Chebyshev's inequality, for a constant $c'$,
\[
\sum_{v\in \AA_3} \expect{\Zf{v}{t}}\leq
c' {t}/{\left((3-\alpha)\log t\right)}
.
\]

From the above results, and the fact that the total number of visits to all three sets is $t$, we get
\[
ct
+
4(t\log t)^{1/(\alpha-1)}\cdot
\expect{\Zf{\trespoint}{t}}+
c'{t}/{[(3-\alpha)\log t]}
\geq
t,
\]
which implies
\[
\expect{\Zf{\trespoint}{t}}=
\Omega\left(
t^\frac{\alpha-3}{\alpha-1}\cdot (\log t)^{-\frac{2}{\alpha-1}}
\right)
\]
if $ 3 - \alpha = \omega(1 / \log t) $.
We can express the probability of $h_f\leq t$ in terms of the above mean as

\begin{align*}
	\pr{h_f \leq t} & = \pr{\Zf{\trespoint}{t}>0} \\
	& = \expect{\Zf{\trespoint}{t}}/\,
	\expect{\Zf{\trespoint}{t} \ \middle|\ \Zf{\trespoint}{t}>0}.
\end{align*}
We have
\[
\expect{\Zf{\trespoint}{t} \ \middle|\ \Zf{\trespoint}{t}>0}
\leq
\expect{\Zf{\origin}{t}}+1
.
\]
We also compute
\[
\expect{\Zf{\origin}{t}} = \bigo(1/(3-\alpha)^2).
\]
Combining the last four equations
yields
\[
\pr{h_f \leq t}
=
\myOmega{t^\frac{\alpha-3}{\alpha-1}\cdot (\log t)^{-\frac{2}{\alpha-1}}\cdot (3-\alpha)^2},
\]
and substituting $t = \myTheta{\dist^{\alpha - 1}}$ completes the proof.

\subsubsection{Detailed Proof}
We give now the details of the analysis.
Throughout the section $2 < a \leq 3-\omega(1 / \log \dist)$.

\begin{lemma}[Bound on the visits to $\firstregion$]
    \label{lemma:boundclose}
    For any $t= \Theta(\dist^{\alpha-1})$ large enough, there is a constant $c\in (0,1)$ such that
    \[
        \sum_{v\in Q_{\dist}(\origin)} \expect{\levyflightvisits{v}{t}\Mid \EE_t} \le ct .
    \]
\end{lemma}

\begin{proof}We bound  the probability the walk has moved to distance $\frac{5}{2} \dist$ at least once, within time $t=\Theta\left(\dist^{\alpha-1}\right)$,    by the probability that at least one of the performed jumps  is no less than $5 \dist $ (we denote this
    latter event by $H$). Indeed, if there is a jump of length at least $5\lambda$, the walk moves necessarily to distance no less than $\frac{5}{2}\dist$. Then,
    \begin{align*}
        \pr{S_j \ge 5 \dist \mid S_j \le (t\log t)^\frac{1}{\alpha-1} } = & \ \sum_{k=5 \dist}^{(t\log t)^\frac{1}{\alpha-1}} \frac{\levyconst}{k^\alpha} \\
        \stackrel{(\ast)}{\ge} & \ \frac{\levyconst}{\alpha-1}\left(\frac{1}{(5\dist)^{\alpha-1}} - \frac{1}{t\log t}  \right) \\
        \stackrel{(\star)}{\ge} & \ \frac{\levyconst}{2(\alpha-1)(5\dist)^{\alpha-1}},
    \end{align*}
    where $(\ast)$ follows  for the integral test (\cref{fact:integraltest}), while $(\star)$  easily  holds for a large enough
      $\dist$  since $t = \Theta(\dist^{\alpha-1})$.
    Thanks to the mutual  independence among the random destinations chosen by the agent,   the probability of the event ``the desired jump takes place within time $c'\cdot 2(\alpha-1)(5\dist)^{\alpha-1}/\levyconst$''  is bounded by
    \[
        1-\left[1- \frac{\levyconst}{2(\alpha-1)(5\dist)^{\alpha-1}}\right]^{c'\frac{2(\alpha-1)(5\dist)^{\alpha-1}}{\levyconst}} \ge \frac{3}{4},
    \]
    for some constant  $c'>0$ and for $\dist$ large enough.
     Hence, by choosing  $t \ge 4c'\cdot 2(\alpha-1)(5\dist)^{\alpha-1}/\levyconst$, the desired jump takes place with probability $\frac 3 4$, within time $\frac t 4$. Once reached such a distance (conditional on the previous  event), \cref{fig:otherzones01} shows there are at least other 3 mutually disjoint regions which are  at least as equally likely as $Q_{\dist}(\origin)$ to be visited at any future time.

    \begin{figure}[t]

        \centering
        \includegraphics[scale=0.75]{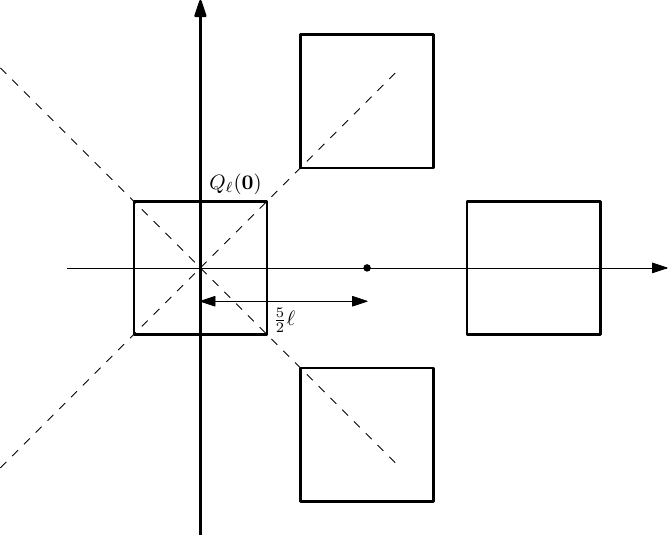}
        \caption{The disjoint zones at least as equally likely as $Q_{\dist}(\origin)$ to be visited.}
        \label{fig:otherzones01}
    \end{figure}
    Thus, the probability to visit $Q_{\dist}(\origin)$ at any future time step is at most $\frac 1 4$. Observe that
    \begin{align*}
        \expect{\sum_{v\in Q_{\dist}(\origin)} \levyflightvisits{v}{t} \mid \EE_t} = & \expect{\sum_{v\in Q_{\dist}(\origin)} \levyflightvisits{v}{t} \mid H, \EE_t }\pr{H \mid \EE_t} + \expect{\sum_{v\in Q_{\dist}(\origin)} \levyflightvisits{v}{t} \mid H^C, \EE_t }\pr{H^C \mid \EE_t}  \\
        \le & \ \left(\frac{1}{4}t + \frac{1}{4}\cdot \frac{3}{4}t\right)\frac{3}{4} + t\cdot \frac{1}{4} \\
        = & \ \frac{t}{4}\left(1 + \frac{3}{4}+\frac{9}{16}\right)
        =  \ \frac{37}{64}t ,
    \end{align*}
    and    the proof is completed.
\end{proof}

For the rest of \cref{ssec:analysislevyflight}, let $ t = \myTheta{\dist^{\alpha - 1}} $ as in \cref{lemma:visitsorigin}
\begin{remark}\label{remark:conditional_monotonicity}
	The monotonicity property (\cref{lemma:monotonicity}) holds despite the conditional event $ \EE_t $. The proof is exactly the same.
\end{remark}

Notice that, from $\expect{\levyflightvisits{v}{t}\mid \EE_t} = \sum_{i=0}^t p_{v,i}$ and the monotonicity property, we easily get  the following bound.

\begin{corollary}\label{corollary:induction}
    $\expect{\levyflightvisits{u}{t}\mid \EE_t} \ge \expect{\levyflightvisits{v}{t}\mid \EE_t}$ for all $v\notin Q_{\distu}(\origin)$.
\end{corollary}

Namely, the more the node is ``far'' (according to the sequence of squares $\{Q_{d}(\origin)\}_{d\in \nat}$) from the origin, the less it is visited on average. Thus, each node is visited at most as many times as the origin, on average. This easily gives an upper bound on the total number of visits to $\secondregion$ until time $t$, namely, by taking $u = \trespoint$ and by observing that each $v \in \secondregion$ lies outside $Q_{\dist}(\origin)$, we get that the average number of visits to $\secondregion$ is at most    the expected number of visits to the \treasure $\trespoint$ (i.e. $\expect{\levyflightvisits{\trespoint}{t}\mid \EE_t}$) times (any upper bound of) the size  of $\secondregion$:
in formula,  it is    upper bounded by $\expect{\levyflightvisits{\trespoint}{t}\mid \EE_t} \cdot 4(t\log t)^{\frac{2}{\alpha-1}}$ if $\alpha\in (2,3)$, and by $\expect{\levyflightvisits{\trespoint}{t}\mid \EE_t} \cdot 4t\log^2 t$ if $\alpha=3$.

The next lemma considers    $\thirdregion$.

\begin{lemma}[Bound on visits to $\thirdregion$]\label{lemma:boundfaraway}
    It holds that
     \begin{align}
        \sum_{\substack{v = (x,y)\colon \\ |x|+|y| \geq 2(t\log t)^\frac{1}{\alpha-1}}} & \expect{\levyflightvisits{v}{t} \mid \EE_t} = \bigo\left( \frac{t}{(3-\alpha)\log t}\right) \ \ \text{if $\alpha\in(2,3)$;} \label{eq:boundfaraway1} \\
        \sum_{\substack{v = (x,y)\colon \\ |x|+|y| \geq 2\sqrt{t}\log t}} & \expect{\levyflightvisits{v}{t} \mid \EE_t} = \bigo\left( \frac{t}{\log t}\right) \ \ \text{if $\alpha = 3$}.\label{eq:boundfaraway2}
    \end{align}
\end{lemma}

\begin{proof}Let $\levyflightrand{t'}$ be the two dimensional random variable representing the coordinates of the node the agent performing the \levyflight is located in at time $t'$.
    Consider the projection of the \levyflight on the $x$-axis, namely the random variable $X_{t'}$ such that $\levyflightrand{t'} = (X_{t'}, Y_{t'})$.
    The random variable $X_{t'}$ can be expressed as the sum of $t'$ random variables $S^x_j$, $j=1,\dots,t'$, representing the jumps (with sign) that the projection of the walk takes at each of the $t'$ rounds.
    The partial distribution of the jumps along the $x$-axis, conditional on the event $\EE_t$, can be derived as follows.\footnote{We remark that in \cref{app:projection} we estimate the unconditional distribution of the jump projection length on the $x$-axis (\cref{lemma:projection}) for any $\alpha > 1$. Nevertheless, in this case we are conditioning on the event the original two dimensional jump is bounded, and thus we cannot make use of \cref{lemma:projection}.} For any $0 \le d \le (t\log t)^\frac{1}{\alpha-1}$,
    \begin{align}\label{eq:partialdistribution}
        & \ \pr{S^x_j = \pm d \mid S_j \le (t\log t)^\frac{1}{\alpha-1}} \nonumber \\
        = & \left[\frac{1}{2} + \sum_{k= 1}^{(t\log t)^\frac{1}{\alpha-1}} \frac{\levyconst}{2k^{\alpha+1}} \right]\ind_{d=0} + \left[ \frac{\levyconst}{2d^{\alpha+1}} + \sum_{k= 1+d}^{(t\log t)^\frac{1}{\alpha-1}} \frac{\levyconst}{k^{\alpha+1}} \right] \ind_{d\neq 0} ,
    \end{align}
    where: $\ind_{d\in A}$ returns  $1$ if $d\in A$ and 0 otherwise, the term
    \[
       \frac{\ind_{d=0}}{2} + \frac{\levyconst}{2d^{\alpha+1}}\ind_{d\neq 0}
    \]
    is the probability that the original jump lies along the horizontal axis and has ``length'' exactly $d$ (there are two such jumps if $d > 0$), and, for $k\ge 1+d$, the terms
    \[
        \frac{\levyconst}{2k^{\alpha+1}}\ind_{d = 0} +  \frac{\levyconst}{k^{\alpha+1}}\ind_{d\neq 0}
    \]
    are the probability that the original jump has ``length'' exactly $k$ and its projection on the horizontal axis has ``length'' $d$ (there are two such jumps if $d =0$, and four such jumps if $d > 0$).
    Observe that  $\eqref{eq:partialdistribution}$ is of the order of
    \[
        \myTheta{\frac{1}{d^{\alpha+1}} +\sum_{k= 1+d}^{(t\log t)^\frac{1}{\alpha-1}} \frac{1}{k^{\alpha+1}}}.
    \]
    By the integral test (\cref{fact:integraltest} in \ref{app:tools}), we know that this probability is
    \[
        \pr{S^x_j = \pm d \mid \EE_j} = \Theta\left(\frac{1}{d^{\alpha}}\right).
    \]
    Due to symmetry, it is easy to see that $\expect{X_{t'} \mid \EE_t} = 0$ for each time $t'$, while
    \[
        \variance{X_{t'} \mid \EE_t} = \sum_{i=1}^{t'} \variance{S^x_j \mid \EE_j} = t' \variance{S^x_1 \mid E_1}
    \]
    since $S^x_1, \dots, S^x_{t'}$ are i.i.d.

    As for  the case  $\alpha\in (2,3)$, the variance of $S^x_1$ conditioned to the event $E_1 = \left\{S_1 \le (t\log t)^\frac{1}{\alpha-1}\right\}$, can be bounded   as follows
    \begin{align*}
        \variance{S^x_1 \mid E_1} \le & \ \sum_{k=1}^{(t\log t)^\frac{1}{\alpha-1}} \bigo\left(\frac{k^2}{k^{\alpha}}\right) \\
        \stackrel{(\ast)}{=} & \ \bigo\left(\frac{1}{3-\alpha}\left[(t\log t)^\frac{3-\alpha}{\alpha-1}-1\right]\right)  \\
        = & \ \mybigo{\frac{\left(t \log t\right)^\frac{3-\alpha}{\alpha-1}}{3 - \alpha}},
    \end{align*}
    where, in $(\ast)$, we used the integral test (\cref{fact:integraltest}). Observe that the event $\EE_t=\bigcap_{i=1}^t E_i$ has probability
    \[
        \pr{\EE_t} = 1-\bigo\left(\frac{1}{\log t}\right).
    \]
    Then, for each $t'\le t$, from the Chebyshev's inequality and the fact that $\expect{X_{t'} \mid \EE_t} = 0$,
    \begin{align*}
        \pr{\abs{X_{t'}} \ge (t\log t)^\frac{1}{\alpha-1} \mid \EE_t} \le \frac{t'\variance{S^x_1 \mid E_1}}{(t\log t)^\frac{2}{\alpha-1}} \le \frac{t\variance{S^x_1 \mid E_1}}{(t\log t)^\frac{2}{\alpha-1}}   = \bigo\left( \frac{1}{(3 - \alpha)\log t} \right),
    \end{align*}
    which implies that
    \begin{align*}
        \pr{\abs{X_{t'}} \ge (t\log t)^\frac{1}{\alpha-1}} \le \pr{\abs{X_{t'}} \ge (t\log t)^\frac{1}{\alpha-1} \mid \EE_t} + \pr{\EE_t^C} = \bigo\left(\frac{1}{(3-\alpha)\log t}\right).
    \end{align*}
    Then, the probability that both $X_{t'}$ and $Y_{t'}$ are less than $(t\log t)^\frac{1}{\alpha-1}$ (call the corresponding events $A_{x,t'}$ and $A_{y,t'}$, respectively) is
    \begin{align*}
        \pr{A_{x,t'} \cap A_{y,t'}} = \pr{A_{x,t'}} + \pr{A_{y,t'}} - \pr{A_{x,t'} \cup A_{y,t'}} \ge 1-\bigo\left(\frac{1}{(3-\alpha)\log t}\right),
    \end{align*}
    for any $t'\le t$.
    Then, let $Z'(t)$
    be the random variable indicating the number of times the \levyflight visits the set of nodes whose coordinates are both no less than $(t\log t)^\frac{1}{\alpha-1}$, until time $t$. Then,
    \[
        \expect{Z'(t) \mid \EE_t} \le \sum_{\substack{v=(x,y) \\ \abs{x} + \abs{y} \ge 2(t\log t)^\frac{1}{\alpha-1}}} \expect{\levyflightvisits{v}{t} \mid \EE_t},
    \]
    and
    \begin{align*}
        \expect{Z'(t) \mid \EE_t} = & \ \sum_{i=0}^t \expect{Z'(i) \mid A_{x,i} \cap A_{y,i}, \EE_t}\pr{A_{x,i} \cap A_{y,i} \mid \EE_t} \\
        + & \ \sum_{i=0}^t \expect{Z'(i) \mid (A_{x,i} \cap A_{y,i})^C, \EE_t}\pr{(A_{x,i} \cap A_{y,i})^C \mid \EE_t}\\
        = & \ \sum_{i=0}^t\expect{Z'(i) \mid (A_{x,i} \cap A_{y,i})^C, \EE_t}\pr{(A_{x,i} \cap A_{y,i})^C \mid \EE_t} \\
        \le & \ t\cdot \bigo\left(\frac{1}{(3-\alpha)\log t}\right) =\bigo\left(\frac{t}{(3-\alpha)\log t}\right),
    \end{align*}
    which proves    \cref{eq:boundfaraway1}.

    As for the case  $\alpha=3$, the variance of $S_1^x$ conditional on $E_1$ is  $\mybigo{\log (t\log t)}$. Then, we look at the probability that $\abs{X_{t'}}$ is at least $\sqrt{t}\cdot \log t$ conditional on $\EE_t$, which is, again, $\mybigo{1/\log t}$. Finally,  the proof proceeds in exactly the same way of the previous case, obtaining \cref{eq:boundfaraway2}.
\end{proof}

\begin{lemma}\label{lemma:goodtrick}
    For $t = \Theta( \dist^{\alpha-1})$,
\begin{align}
        & ct + \expect{\levyflightvisits{\trespoint}{t}\mid \EE_t} \cdot 4(t\log t)^{\frac{2}{\alpha-1}} + \bigo\left(\frac{t}{(3-\alpha)\log t}\right) \geq t \ \ \text{if $\alpha \in(2,3)$;} \label{eq:goodtrick1} \\
        & ct + \expect{\levyflightvisits{\trespoint}{t}\mid \EE_t} \cdot 4t\log^2 t + \bigo\left(\frac{t}{\log t}\right) \geq t \ \ \text{if $\alpha = 3$.}\label{eq:goodtrick2}
    \end{align}
\end{lemma}

\begin{proof}Suppose the agent has made $t$ jumps for some $t=\Theta( \dist^{\alpha-1})$  (the same $t$ of \cref{lemma:boundclose}), thus visiting exactly $t$ nodes. Then,
    \[
        \expect{\sum_{v\in \integer^2}\levyflightvisits{v}{t} \mid \EE_t} = t.
    \]
    As for  \cref{eq:goodtrick1}, we observe that,  from \cref{lemma:boundclose}, the number of visits to $\firstregion=Q_{\dist}(\origin)$ until time $t$ is at most $ct$, for some constant $c\in (0,1)$. From \cref{lemma:boundfaraway}, the number of visits to $\thirdregion$ is at most $\bigo\left(t/\left ((3-\alpha)\log t\right )\right)$. Thanks to \cref{corollary:induction}, each of the remaining nodes, i.e., the nodes in $\secondregion$ (whose size  is   at most $4(t\log t)^\frac{2}{\alpha-1}$), is visited by the agent at most $\expect{\levyflightvisits{\trespoint}{t} \mid \EE_t}$ times. It follows that
    \[
        ct + \expect{\levyflightvisits{\trespoint}{t} \mid \EE_t}\cdot 4(t\log t)^\frac{2}{\alpha-1}+\bigo\left(\frac{t}{(3-\alpha)\log t}\right) \ge t.
    \]

    As for  \cref{eq:goodtrick2},  we proceed as for the first case above, by noticing that
the number of visits to $\secondregion$ is at most $\expect{\levyflightvisits{\trespoint}{t} \mid \EE_t} \cdot  (4t\log^2 t)$. This gives \cref{eq:goodtrick2}.
\end{proof}

The next two lemmas  provide a clean relationship between the probability to hit a node $u$ within time $t$ to the average number of visits to the origin and to the average number of visits to $u$ itself. In particular, the first lemma  estimate the average number of visits to the origin.
For any $t \ge 0$ and $\alpha\in(2,3]$, let $\expect{\levyflightvisits{\origin}{t}\mid \EE_t} = a_t(\alpha)$.

\begin{lemma}[Visits to the origin]\label{lemma:visitsorigin}\
    \begin{enumerate}[(a)]
         \item If $\alpha \in(2,3)$, then $a_t(\alpha) = \mybigo{1/(3-\alpha)^2}$. \label{item:ballvisits1}
        \item If $\alpha = 3$, then $a_t(3) = \mybigo{\log^2 t}$. \label{item:ballvisits2}
    \end{enumerate}
\end{lemma}

\begin{proof}For the case $\alpha \in(2,3)$, we proceed as follows. Since $\expect{\levyflightvisits{\origin}{t} \mid \EE_t} = \sum_{k=1}^t p_{\origin,k}$, it suffices to accurately bound the probability $p_{\origin,k}$ for each $k=1, \dots, t$. Let us make a partition of  the natural numbers in the following way
    \[
        \nat = \bigcup_{t'=1}^\infty \bigg[\nat \cap \left[2t'\log t', 2(t'+1)\log(t'+1)\right)\bigg].
    \]
    For each $k\in \nat$, there exists $t'$ such that $k\in \left[2t'\log t', 2(t'+1)\log(t'+1)\right)$. Then, within $2t'\log t'$ steps, we claim that the walk has moved to distance $\lambda = \frac{(t')^\frac{1}{\alpha-1}}{2}$ at least once, with probability $\Omega\left(\frac{1}{(t')^2}\right)$.
    Indeed, if there is one jump of length at least $2\lambda$,  then the  walk has necessarily moved to a distance at least $\lambda$ from the origin. We now bound  the probability that one jump is at least $2\lambda$. For the integral test and for $ \lambda > 0 $, we get
    \begin{align*}
        \pr{S_j \ge 2\lambda \mid S_j \le (t\log t)^\frac{1}{\alpha-1} } \ge & \ \frac{1}{\pr{S_j \le (t\log t)^\frac{1}{\alpha-1}}}\left[\int_{2\lambda}^{(t\log t)^\frac{1}{\alpha-1}} \frac{\levyconst}{s^\alpha}ds\right] \\
        \ge & \ \frac{\levyconst}{\alpha-1}\left(\frac{1}{t'}-\frac{1}{t\log t}\right)\\
        \ge & \ \frac{\levyconst}{\alpha-1}\left(\frac{1-\frac{t'}{t\log t}}{t'}\right) \\
        \ge & \ \frac{\levyconst}{\alpha-1}\left(\frac{1-\frac{1}{2\log (t')\log t}}{t'}\right) \\
        = & \ \Omega\left(\frac{1}{t'}\right),
        \end{align*}
    where the last inequality holds since $2t'\log t' \le t$. Thus, the probability that the first $2t'\log t'$ jumps are less than $2\lambda$ is
    \begin{align*}
        \pr{\cap_{j=1}^{2t'\log t'} \{S_j < 2\lambda\} \mid \EE_t} \stackrel{(\ast)}{=} & \ \left[1-\pr{S_1 < 2\lambda \mid S_1 \le (t\log t)^\frac{1}{\alpha-1}}\right]^{2t'\log t'} \\
        \ge & \ \left[1-\Omega\left(\frac{1}{t'}\right)\right]^{2t'\log t'} \\
        = & \ \bigo\left(\frac{1}{(t')^2}\right),
    \end{align*}
    where in $(\ast)$ we used the  independence among the agent's jumps.
    Once the agent reaches such a distance,   \cref{lemma:monotonicity} implies that there are at least $\lambda^2 = \Omega\left( (t')^\frac{2}{\alpha-1}\right) $ different nodes that are at least as equally likely as $\origin$ to be visited at any given future time.
    Thus, the probability to reach the origin at any future time is at most $\bigo\left(\frac{1}{(t')^\frac{2}{\alpha-1}}\right)=\bigo\left(\frac{1}{(t')^{1+\epsilon}}\right)$ with $\epsilon = (3-\alpha)/(\alpha-1)>0$: in particular the bound holds for $p_{\origin,k}$. Observe that in an interval $\left[2t'\log t', 2(t'+1)\log(t'+1)\right)$ there are
    \[2(t'+1)\log(t'+1)-2t'\log t' = 2t'\left[\log\left(1+\frac{1}{t'}\right)\right]+2\log(t'+1) = \bigo\left(\log t'\right)\]
    integers. Let $\levyflightrand{t}$ be the two-dimensional random variable denoting the node visited at time $t$ by an agent which  started from  the origin, and let $H_{t'}$ be the event $\cup_{j=1}^{2t'\log t'} \{S_j \ge 2\lambda\}$. Observe that,  by the law of total probability,
    \[
        p_{\origin,k} = \pr{\levyflightrand{t} = \origin \mid H_{t'}, \EE_t}\pr{H_{t'} \mid \EE_t} + \pr{\levyflightrand{t} = \origin \mid H_{t'}^C, \EE_t}\pr{H_{t'}^C \mid \EE_t} .
    \]

    Thus, if $I_{t'}=\left[2t'\log t',2(t'+1)\log(t'+1)\right)$, we get
    \begin{align*}
        \sum_{k=1}^t p_{\origin,k} \le & \ \sum_{t' = 1}^t \sum_{k\in I_{t'}} p_{\origin,k} \\
        \le & \ \sum_{t' = 1}^t \left[\pr{\levyflightrand{t} = \origin \mid H_{t'}, \EE_t}\pr{H_{t'} \mid \EE_t} + \pr{\levyflightrand{t} = \origin \mid H_{t'}^C, \EE_t}\pr{H_{t'}^C \mid \EE_t}\right]\bigo(\log t')\\
        \le & \ \sum_{t' = 1}^t \left[\bigo\left( \frac{1}{(t')^{1+\epsilon}}\right)+\bigo\left(\frac{1}{(t')^2}\right)\right]\bigo(\log t') \\
        = & \ \sum_{t' = 1}^t \bigo\left( \frac{\log t'}{(t')^{1+\epsilon}}\right) \stackrel{(\star)}{=} \mybigo{\frac{1}{\epsilon^2}} = \mybigo{\frac{1}{(3-\alpha)^2}},
    \end{align*}
    where for $ \star $ we used the integral test and partial integration. In particular, it holds that
    \[
    	\int_{1}^{t} \frac{\log (x)}{x^{1 + \epsilon}} \diff x  = \frac{1}{\epsilon^2}\left [ \frac{\epsilon(1 - \log x)}{x^{\epsilon}} - \frac{\epsilon+1}{x^{\epsilon}}\right ]_{1}^t.
    \]

    For the case $\alpha = 3$, we can consider the same  argument above for  the previous case where we  fix  $\lambda = \sqrt{t'}$. Then   the  proof proceeds as in the previous case  by observing    that the average number of visits until time $t$ is, now, of magnitude  $\mybigo{\log^2 t}$.
\end{proof}

\begin{lemma}\label{corollary:visitsorigin}
    Let $u\in \integer^2$ be any node. Then,
    \begin{itemize}
        \item[(i)] $\expect{\levyflightvisits{u}{t} \mid \EE_t} \le a_t(\alpha)$,
        \item[(ii)] $1 \le \expect{\levyflightvisits{u}{t} \mid \levyflightvisits{u}{t} > 0, \EE_t} \le a_t(\alpha)$,
        \item[(iii)] $\expect{\levyflightvisits{u}{t} \mid \EE_t}/a_t(\alpha) \le \pr{\levyflightvisits{u}{t}>0 \mid \EE_t} \le \expect{\levyflightvisits{u}{t} \mid \EE_t} $.
    \end{itemize}
\end{lemma}
\begin{proof}Claim (i) is a direct  consequence of (ii), since $ \expect{\levyflightvisits{u}{t} \mid \levyflightvisits{u}{t} > 0, \EE_t} \ge \expect{\levyflightvisits{u}{t} \mid \EE_t}$. As for Claim   (ii),   let $\tau$ be the first time the agent visits $u$. Then, conditional on $\levyflightvisits{u}{t}>0$, $\tau$ is at most $t$, and
    \[
        \expect{\levyflightvisits{u}{t} \mid \levyflightvisits{u}{t} >0, \EE_t} = \expect{\levyflightvisits{\origin}{t-\tau} \mid \tau \le t, \EE_t} \le \expect{\levyflightvisits{\origin}{t} \mid \EE_t} = a_t(\alpha).
    \]
    Notice that   this expectation is at least 1 since  we have  the conditional event.
    As for Claim (iii), let us explicitly write the term $\expect{\levyflightvisits{u}{t}\mid \levyflightvisits{u}{t} >0, \EE_t }\cdot \pr{\levyflightvisits{u}{t}>0 \mid \EE_t}$:
    \begin{align*}
        & \ \sum_{i=1}^t i\pr{\levyflightvisits{u}{t} = i \mid \levyflightvisits{u}{t} > 0, \EE_t}\cdot \pr{\levyflightvisits{u}{t}>0 \mid \EE_t} \\
        = & \ \sum_{i=1}^t i\frac{\pr{\levyflightvisits{u}{t} = i , \levyflightvisits{u}{t} > 0, \EE_t}}{\pr{\levyflightvisits{u}{t} > 0, \EE_t}}\cdot \frac{\pr{\levyflightvisits{u}{t} > 0, \EE_t}}{\pr{\EE_t}} \\
        = & \ \sum_{i=1}^t i\frac{\pr{\levyflightvisits{u}{t} = i , \levyflightvisits{u}{t} > 0, \EE_t}}{\pr{\EE_t}} \\
        = & \ \sum_{i=1}^t i\pr{\levyflightvisits{u}{t} = i \mid \EE_t} \\
        = & \ \expect{\levyflightvisits{u}{t} \mid \EE_t}.
    \end{align*}
    Then,
    \[
        \expect{\levyflightvisits{u}{t} \mid \EE_t} \ge \pr{\levyflightvisits{u}{t}>0 \mid \EE_t} = \frac{\expect{\levyflightvisits{u}{t} \mid \EE_t}}{\expect{\levyflightvisits{u}{t}\mid \levyflightvisits{u}{t} >0, \EE_t }} \ge \frac{\expect{\levyflightvisits{u}{t} \mid \EE_t}}{a_t(\alpha)},
    \]
    since, from Claim (ii),  $\expect{\levyflightvisits{u}{t}\mid \levyflightvisits{u}{t} >0, \EE_t } \le a_t(\alpha)$.
\end{proof}

We can now complete the proof of \cref{lemma:hittinglevyflight}, as follows.
From \cref{lemma:goodtrick} we have that $ \expect{\levyflightvisits{\trespoint}{t} \Mid \EE_t} = \myOmega{(3 - \alpha )^2 / \left (t^{(3-\alpha)/(\alpha-1)}(\log t)^{2/(\alpha - 1)}\right )}  $ if $ \alpha \in (2,3) $ and $ 3 - \alpha = \omega(1/\log t) $ and $ \expect{\levyflightvisits{\trespoint}{t} \Mid \EE_t} = \myOmega{1 / \left (\log t\right )^{2}} $ if $ \alpha = 3 $. Then, \cref{lemma:visitsorigin} and claim (iii) of \cref{corollary:visitsorigin} give the results by substituting $ t = \myTheta{\dist^{\alpha -1}} $.

 \subsection{Proof of Lemma \ref{lemma:couplevyflightintowalk}}
\label{ssec:tech_contrib_coupling}

Let $S_j$
be the random variable denoting the $j$-th jump-length. 
From \cref{eq:jumpatleast}, we get
\[\prob\left(S_j > (t\log t)^\frac{1}{\alpha-1}\right) = \myTheta{\frac{1}{t\log t}}.\]
Let  $E_j$ be the event $\left\{S_j \le (t\log t)^\frac{1}{\alpha-1}\right\}$, and let $\EE_t$ be the intersection of $E_j$ for $j=1, \dots, t$.
Notice that, by the union bound,  the probability of $\EE_t $ is $1-\bigo(1/\log t)$. We next apply the multiplicative form of the Chernoff bound to the sum of $S_j$, conditional on the event $\EE_t$. This is possible since the variable $S_j/(t\log t)^\frac{1}{\alpha-1} $ takes values in $[0,1]$.
To this aim, we first  bound  the expectation of the sum of the random variables $S_j$, for $j=1,\dots, t$ conditional on $\EE_t$.
\begin{align*}
    \mean\left[\sum_{j=1}^t S_j \ \bigm|\  \EE_t\right] = & \ \sum_{j=1}^t \mean[S_j \mid \EE_t] = \myTheta{t} + t\frac{\levyconst}{\prob(\EE_t)}\sum_{d = 1}^{(t\log t)^\frac{1}{\alpha-1}}\frac{d}{d^\alpha} \\
    \le & \ \myTheta{t} + 2\levyconst t\sum_{d = 1}^{(t\log t)^\frac{1}{\alpha-1}}\frac{1}{d^{\alpha-1}} \\
    \stackrel{(a)}{\le} & \mybigo{\mu t},
\end{align*}
where in $ (a) $ we have $ \mu = \min \{\log \dist, \frac{1}{\alpha - 2}\} $ for the integral test (\cref{fact:integraltest}).
We now  use the Chernoff bound (\cref{lem:chernoff:multiplicative}) on the normalized sum of all jumps, to show that such a  sum is at most linear in $\mybigo{\mu t}$ with probability $1-\exp(-t^{\Theta(1)})$, conditional on $\EE_t$. In formula,
\begin{align*}
    \prob\left(\sum_{j=1}^t S_j \ge  \myTheta{\mu t} \ \bigm|\ \EE_t\right) = & \ \prob\left(\frac{\sum_{j=1}^t S_j}{(t\log t)^\frac{1}{\alpha-1}} \ge \frac{2\myTheta{\mu t}}{(t\log t)^\frac{1}{\alpha-1}} \ \bigm|\ \EE_t\right)
    \\
    \le &\ \exp\left(-\frac{2\myTheta{\mu t}}{3\left((t\log t)^\frac{1}{\alpha-1}\right)}\right) \\
    \le & \ \exp\left(-\myTheta{\frac{\mu t^\frac{\alpha-2}{\alpha-1}}{(\log t)^\frac{1}{\alpha-1}}}\right)\le \exp\left(-\Theta\left(t^\frac{\alpha-2}{2(\alpha-1)}\right)\right).
\end{align*}
Then, define
\begin{align*}
    & A = \{\text{the \levywalk finds the \treasure within time } \myTheta{\mu t} \text{,}  \\
    & A_1 = \left\{\sum_{j=1}^t S_j = \myTheta{\mu t} \right\} \text{, and } \\
    & A_2 = \{\text{the \levyflight finds the \treasure within } t \text{ jumps}\}.
\end{align*}
Observe that the event $A_1 \cap A_2$ implies that the \levywalk finds the \treasure within $t$ jumps, which, in turn, implies the event $ A $.
Indeed, $A_1 \cap A_2$ implies the \treasure is found in one of the $ t $ jump endpoints, and the overall amount of  steps  is $\myTheta{\mu t}$. Let $ p(t) = \pr{h_f \le t \Mid \EE_t} $. Then
\begin{align*}
    \prob(A) \ge & \ \prob(A_1,A_2)
    \ge \ \prob(A_1,A_2,\EE_t) \\
{=} & \ \prob(\EE_t)\left[\prob(A_1\mid \EE_t)+\prob(A_2 \mid \EE_t)-\prob(A_1 \cup A_2 \mid \EE_t)\right] \\
{\ge} & \ \left(1-\bigo\left(\frac{1}{\log t}\right)\right)\left[1-exp(-t^{\Theta(1)})+p(t)-1\right] \\
    = & \ \left(1-\bigo\left(\frac{1}{\log t}\right)\right)\left(p(t) - exp\left(-t^{\Theta(1)}\right)\right),
\end{align*}
where in the second line we used the definition of conditional probability and the inclusion-exclusion principle, and in the third line we used that $\pr{\EE_t} = (1-\bigo(1/\log t))$, $\pr{A_1 \mid \EE_t} \ge 1 - \exp(-t^\Theta(1))$, and $\pr{A_1\cup A_2 \mid \EE_t} \le 1$.
 \subsection{Proof of Lemma \ref{lemma:lowbound_superdiffusive}}

\label{sec:lowbound_superdiffusive}

Let $X_i$
    be the $x$-coordinate of the agent at the end of the $i$-th jump. For any $i\le t$, we bound the probability that $X_i > \dist / 4$. The probability that there is a jump whose length is at least $\dist$ among the first $i$ jumps is $\myTheta{i/\dist^{\alpha-1}}$.
	We first  consider the case  $\alpha \in (2,3)$.
	Conditional on the event that the first $i$ jump-lengths are all smaller than $\dist$ (event $C_i$), the expectation of $X_i$ is zero and its variance is
	\[
	\myTheta{1} + i\cdot\sum_{d=1}^{\dist/4} \myTheta{\frac{d^2}{d^\alpha}} = \myTheta{i \cdot \nu \dist^{3 - \alpha}},
	\]
	for the integral test (\cref{fact:integraltest}), where $ \nu = \min \{\log \dist, \frac{1}{3 - \alpha}\} $.
	Chebyshev's inequality implies that
	\[
	\pr{\abs{X_i} \ge \dist/4 \mid C_i} \le \frac{\myTheta{i \cdot \nu \dist^{3 - \alpha}}}{\Theta(\dist^2)} = \myTheta{\frac{i\nu}{\dist^{\alpha-1}}}.
	\]
	Since the conditional event has probability $1-\myTheta{i/\dist^{\alpha-1}}$, then the ``unconditional'' probability  of the event $\abs{X_i} \le \dist/4$ is
	\[
	\left [1 - \myTheta{\frac{i}{\dist^{\alpha - 1}}}\right ]\left [1 - \mybigo{\frac{i  \nu}{\dist^{\alpha - 1}}}\right ] = 1 - \mybigo{\frac{\nu t}{\dist^{\alpha - 1}}},
	\]
	since $i\le t$, for  $ t $ which is some $ \mybigo{\dist^{\alpha - 1}/\nu} $.
	The same result holds analogously for $Y_i$ (the $y$-coordinate of the agent after the $i$-th jump), thus obtaining   $\abs{X_i} + \abs{Y_i} \le \dist/2$, with probability $1-\mybigo{\nu t/ \dist^{\alpha -1}}$ by the union bound.
	
	During the first jump-phase,  thanks to \cref{cor:visit-direct-path}, the probability the agents visits the \treasure is $\mybigo{1 / \dist^\alpha }$. Let $2 \le i \le t$. We want to estimate the probability that during $i$-th jump-phase the agents visits the \treasure, having the additional information that $t=\mybigo{\dist^{\alpha - 1}}$. As in the proof of \cref{lemma:prob_anyjumpfindstreasure}, we consider the node $\trespoint$  where the \treasure is located on, and   the rhombus centered in $\trespoint$ that contains the nodes within distance  $\dist / 4$ from $\trespoint$, namely	$ B_{\dist/4}(\trespoint) $ .
	Let $F_i$ be  the event that during the $i$-th jump-phase the agent visits the \treasure; let $V_{i-1}$ be the event that the $(i-1)$-th jump ends in $B_{\dist/4}(\trespoint)$, and let $W_{i-1}$ be  the event that the $(i-1)$-th jump ends at distance farther than $\dist/2$ from the origin.
	Finally, let $\levyflightrand{i}$ be the two-dimensional
	random variable denoting the coordinates of the node the agent is located on at the end of the $i$-th jump-phase. Then,
	\begin{align*}
	\pr{F_i \mid V_{i-1}}\pr{V_{i-1}\mid W_{i-1}} = & \ \sum_{v \in B_{\dist/4}(\trespoint)} \pr{F_i \mid \levyflightrand{i} = v}\pr{\levyflightrand{i} = v \mid W_{i-1}} \\
	\le & \ \bigo\left(\frac{1}{\dist^2}\right)\sum_{v \in B_{\dist/4}(\trespoint)} \pr{F_i \mid \levyflightrand{i} = v},
	\end{align*}
	where in the above inequalities we used the monotonicity property (\cref{lemma:monotonicity}, which holds since the process restricted to the jump endpoints is a \levyflight
), and the fact that, for each $v\in B_{\dist/4}(\trespoint)$, there are at least $\Theta\left(\dist^2\right)$ nodes at distance at least $\dist/2$ from the origin which are more likely to be the destination of the $i$-th jump than $v$. Then, we proceed as  in the proof of \cref{lemma:prob_anyjumpfindstreasure} and   obtain
	\begin{align}
	\pr{F_i \mid V_{i-1}}\pr{V_{i-1}\mid W_{i-1}} = \mybigo{\frac{\mu}{\dist^2}}, \label{eq:lflightlow2}
	\end{align}
	where $ \mu = \min \{\log \dist, \frac{1}{\alpha - 2}\} $.
	By the law of total probabilities, we get
	\begin{align}
	\pr{F_i} = & \ \pr{F_i \mid W_{i-1}}\pr{W_{i-1}} + \pr{F_i \mid W_{i-1}^C}\pr{W_{i-1}^C} \nonumber\\
	= & \ \left[\pr{F_i \mid W_{i-1}, V_{i-1}}\pr{V_{i-1}\mid W_{i-1}}+\pr{F_i \mid W_{i-1}, V_{i-1}^C}\pr{V_{i-1}^C\mid W_{i-1}}\right]\pr{W_{i-1}} \nonumber\\
	+ & \ \pr{F_i \mid W_{i-1}^C}\pr{W_{i-1}^C} \nonumber\\
{\le} & \ \left[\pr{F_i \mid V_{i-1}}\pr{V_{i-1}\mid W_{i-1}}+\pr{F_i \mid W_{i-1}, V_{i-1}^C}\right]\pr{W_{i-1}} + \pr{F_i \mid W_{i-1}^C}\pr{W_{i-1}^C} \nonumber\\
{\le} & \  \left[\mybigo{\frac{\mu}{\dist^2}} + \bigo\left(\frac{1}{\dist^\alpha}\right)\right]\bigo\left(\frac{\nu t}{\dist^{\alpha - 1}}\right) + \bigo\left(\frac{1}{\dist^\alpha}\right) = \bigo\left(\frac{\nu \mu t}{\dist^{\alpha + 1}}\right),
    \label{eq:lowbound1}
	\end{align}
	where for second-last inequality we used that $V_{i-1} \subset W_{i-1}$ and that $\pr{V_{i-1}^C \mid W_{i-1}}\le 1$, while for the last inequality we used \cref{eq:lflightlow2}, and
	that $\pr{F_i \mid W_{i-1},V_{i-1}^C} = \bigo\left(1/\dist^\alpha\right)$, which is true because the jump starts in a node whose distance form the \treasure is $\Omega(\dist)$, and that $\pr{F_i \mid W_{i-1}^C} =\bigo\left(1/\dist^\alpha\right)$, which is true for the same reason.
	
	Thus, by the union bound and by  \cref{eq:lowbound1}, the probability that during at least one between the $t$ jump-phases, the agent finds the \treasure is
	\begin{align*}
	\mybigo{\frac{1}{\dist^\alpha}} + (t-1)\bigo\left(\frac{\nu \mu t}{\dist^{\alpha + 1}}\right) &  = \mybigo{\frac{\nu \mu t^2}{\dist^{\alpha + 1}}}
	\end{align*}
	since $ t \ge \dist $,
	which gives the first claim of the lemma by observing that within time $ t $ at most $ t $ jumps can be performed.
	
	Consider now the case  $\alpha = 3$. The proof proceeds exactly as in the first case,
	with the only  key difference that  the variance of $X_i$ is $\myTheta{i \log \dist}$. This means that the probability that $\abs{X_i}$ is at least $\dist/4$ conditional to $C_i$ is $\mybigo{\log \dist / \dist^2}$, and the ``unconditional'' probability that $\abs{X_i}$ is less than $\dist/4$ is $1-\mybigo{t\log \dist / \dist^{2}}$. It thus follows that
	\[
	\pr{F_i} = \mybigo{{t \log \dist}/{\dist^{4}}}.
	\]
	Then we get the second claimed bound of the lemma:
	$
	\mybigo{{t^2\log \dist}/{\dist^{4}}}.
	$
	
\subsection{Proof of Corollary \ref{cor:plw23}}
\label{sec:cor:plw23}

From \cref{thm:lw23}\cref{thm:lw23:a} and the independence among the agents, we get that
\[
	\pr{\tau_\alpha^k(\trespoint) = \mybigo{\mu \dist^{\alpha - 1}}} = 1 - \left[ 1 - \myOmega{\frac{1}{\gamma\dist^{3 - \alpha}}}\right ]^k \ge 1 - e^{\myOmega{\frac{k}{\gamma\dist^{3 - \alpha}}}},
\]
where we have used the inequality $ 1 - x \le e^{-x} $ for all $ x\in \real $. Then, if $ \alpha = \alpha^\ast + 5\frac{\log\log \dist}{\log \dist} $,
\[
	\pr{\tau_\alpha^k(\trespoint) = \mylittleo{\frac{\dist^2 \log^6 \dist}{k}}} = 1 - e^{-\mylittleomega{\log \dist}},
\]
since $ \mu \le \frac{1}{\alpha - 2} \le \frac{\log\dist}{\log\log \dist}$, $ \dist^{\alpha - 1} = \frac{\dist^2\log^5\dist}{k} $, $ \gamma = \mylittleo{\log^4 \dist}  $ and $ \dist^{3 - \alpha} = \frac{k}{\log^5 \dist} $, thus giving Claim (a).
From \cref{thm:lw23}.\cref{thm:lw23:b} and the independence among the agents, we get
	\[
		\pr{\tau_\alpha^k(\trespoint) > t} = \left [1 - \mybigo{\frac{\mu \nu t^2}{\dist^{\alpha + 1}}}\right ]^k ,
	\]
	for $ \dist \le t = \mylittleo{\dist^{\alpha - 1}/ \nu} $.
	Let $ t = \frac{\dist^2 \cdot \dist^{\frac{\alpha - \alpha^\ast}{2}}}{k\log^4 \dist } $ which is a function in $ \mylittleo{\dist^{\alpha - 1}/ \nu} $ since $ \alpha > \overline{\alpha} $. If $ t \ge \dist $, we get
	\[
		\pr{\tau_\alpha^k ( \trespoint) > t} = \left [1 - \mybigo{\frac{\mu \nu \dist^4}{k^2\dist^{\alpha^\ast + 1}\log^8 \dist  }}\right ]^k \ge e^{-\mybigo{\frac{1}{\log^6 \dist}}} = 1 - \mybigo{\frac{1}{\log^6 \dist}},
	\]
	since $ \mu\nu \le \log^2 \dist $, $ \dist^{\alpha^\ast + 1} = \dist^4 / k $. Notice that, in the inequality we have used that $ 1 - x \ge e^{-\frac{x}{1 - x}} $ if $ x < 1 $, and in the last equality we have used the Taylor's expansion of the exponential function. If $ t < \dist $, we get $ \pr{\tau_\alpha^k(\trespoint) > t} = 1$ (at least $ \dist $ steps are needed to reach the \treasure). Therefore, Claim (b) follows.
Finally, from \cref{thm:lw23}.\cref{thm:lw23:c} and the independence among the agents, we get
	\[
		\pr{\tau_\alpha^k (\trespoint) =  \infty} = \left [1 - \mybigo{\frac{\mu \log \dist}{\dist^{3 - \alpha}}}\right ]^k = \left[1 - \mybigo{\frac{\log^2 \dist}{k \dist^{\alpha^\ast - \alpha}}}\right]^k \ge \exp \left (-\mybigo{\frac{\log^2 \dist}{\dist^{\alpha^\ast - \alpha}}}\right ),
	\]
	since $ \mu \le \log^ \dist $, $ \dist^{3- \alpha} = \dist^{3 - \alpha^\ast}\cdot \dist^{\alpha^{\ast} - \alpha} = k\dist^{\alpha^\ast - \alpha} $. Notice that in the last inequality we have used that $ 1 - x \ge e^{-\frac{x}{1-x}} $ for $ x < 1 $, which is our case since $ k \ge \log^6 \dist $. Hence, we have Claim (c).
 
\section{The Case \texorpdfstring{$\alpha\in(1,2]$}{alpha in (1,2]}}
\label{sec:equivpwbw}

We now analyze the hitting time of \levywalks with parameter $\alpha \in (1,2]$, which is the exponent range for which the jump length has unbounded mean and unbounded variance.
We show the following theorems.

\begin{theorem}
\label{thm:lw12}
	Let $\alpha\in[1+\epsilon,2)$, where $ \epsilon>0 $ is an arbitrarily small constant.
Let $ \trespoint \in \integer^2$, and $\ell = \|\trespoint\|_1$.
Let $ \mu = \min \{\log \dist, \frac{1}{2 - \alpha}\}  $.
	Then:
	\begin{enumerate}[(a)]
		\item \label{thm:lw12:a}
		$ \pr{\tau_\alpha(\trespoint) = \mybigo{\dist}} = \myOmega{{1}/{\mu\dist}}$;
\item \label{thm:lw12:b}
		$\pr{\tau_\alpha(\trespoint) < \infty} = \mybigo{{\mu\log\dist}/{\dist}} $.
	\end{enumerate}
\end{theorem}

\begin{theorem}
\label{thm:lw2}
	Let $ \trespoint \in \integer^2$ and $\ell = \|\trespoint\|_1$.
	Then:
	\begin{enumerate}[(a)]
		\item \label{thm:lw2:a}
		$ \pr{\tau_2(\trespoint) = \mybigo{\dist}} = \myOmega{{1}/{\dist \log \dist}}$;
\item \label{thm:lw2:b}
		$\pr{\tau_2(\trespoint) < \infty} = \mybigo{{\log^2\dist}/{\dist}} $.
	\end{enumerate}
\end{theorem}

The above theorems imply the following bounds  on the parallel hitting time.

\begin{corollary}\label{cor:lw12}
	Let $\alpha\in[1+\epsilon,2]$, where $ \epsilon>0 $ is an arbitrarily small constant.
    Let $\trespoint\in \integer^2$ and  $\dist = \|\trespoint\|_1$.
    Let $\mu = \min \{\log \dist, \frac{1}{2 - \alpha}\}  $. Then:
	\begin{enumerate}[(a)]
		\item \label{cor:lw12:a}
		$ \pr{\tau_{\alpha}^k(\trespoint) = \mybigo{\dist}} = 1 - e^{-\mylittleomega{\log \dist}}$, if $ k = \mylittleomega{\dist \log^2 \dist} $;
		\item \label{cor:lw12:b}
		$ \pr{\tau_\alpha^k (\trespoint) < \infty}  = o(1) $, if $ k = \mylittleo{\dist / \log^2 \dist} $.
	\end{enumerate}
\end{corollary}
Recall also that the trivial lower bound $\tau_{\alpha}^k(\trespoint) \geq \ell$ holds.

\subsection{Proof of Theorems \ref{thm:lw12} and \ref{thm:lw2}}

We first show
the following lemma, which bounds the hitting time for a single \levywalk.

\begin{lemma}
\label{lemma:hittingparetoballistic}
Let $\alpha \in (1,2]$ and $\trespoint \in \integer^2$ with $\|\trespoint\|_1=\dist$. Then,
    \begin{enumerate}[(a)]
        \item $ \pr{\tau_\alpha(\trespoint) = \bigo(\dist)} = \myOmega{\frac{1}{\mu \dist}}$, if $\alpha\in[1+\epsilon,2)$ for an arbitrarily small constant $ \epsilon>0 $, where $ \mu = \min\{\frac{1}{\alpha - 2}, \log \dist\} $.
        \item $ \pr{\tau_\alpha(\trespoint) = \bigo(\dist)} = \myOmega{\frac{1}{\dist \log \dist}}$, if $\alpha =2$.
    \end{enumerate}
\end{lemma}

\begin{proof}Consider a single agent moving according the \levywalk with parameter $\alpha \in (1,2]$ .
By Equation \eqref{eq:jumpatleast} in \cref{sec:preliminaries}, the probability the agent chooses a jump of length at least $\dist$ is of the order of $\myTheta{1/d^{\alpha - 1}}$. Let $ c $ be some constant to be fixed later, and let $ \mu_\alpha $ be equal to $ \min\{\frac{1}{\alpha - 2}, \log \dist\} $ if $ \alpha < 2 $, and to $ \log \dist $ if $ \alpha = 2 $. Then, the probability that all the first $c \dist^{\alpha-1}/\mu_\alpha$ jumps have length less than $  \dist $ is
\begin{align*}
    \left(1-\myTheta{\frac{1}{\dist^{\alpha-1}}}\right)^{\frac{c \dist^{\alpha-1}}{\mu_\alpha}}
\end{align*}
which is greater than positive constant strictly less than 1 thanks to the inequality $ \exp(-x / (1-x)) < (1 - x) $ for $ x < 1 $.
Let $E_i$ be the event that the $ i $-th jump-length is less than $ \dist $ and $ \EE_i = \cap_{1 \le j \le i} E_j $. By what has been said before, we have
\[
    \pr{\EE_i} \ge \myTheta{1} \ \ \text{for all $i\le c\dist^{\alpha-1} / \mu_\alpha$}.
\]
Conditional on $\EE_i$, the sum of the first $i$ jumps is at most $3\dist /4$ with constant probability. Indeed, if $j<i$, the expected value of $S_j$ is, for the integral test (\cref{fact:integraltest})
\[
\expect{S_j \mid \EE_i} = \mybigo{1} + \sum_{d=1}^{\dist  - 1} \frac{\levyconst d}{d^\alpha} = \mybigo{\mu_\alpha \dist^{2-\alpha}}.
\]
Thus,
\[
    \expect{\sum_{j=1}^{i} S_j \mid  \EE_i} \le \sum_{j=1}^{c\dist^{\alpha-1} / \mu_{\alpha}} \expect{S_j \mid \EE_i} = \mybigo{c\dist}.
\]
We choose $c $ small enough so that this expression is less than $\dist / 2$.
Conditional on $\EE_i$, the $\{S_j\}_{j\le i}$ random variables are non negative and we can use the Markov's inequality to get that their sum is bounded by $ 3\dist / 4 $ with constant probability.
Indeed
\begin{align*}
	\pr{\sum_{j=1}^{i} S_j \ge \frac{3\dist}{4} \Mid  \EE_i} \le \pr{\sum_{j=1}^{i} S_j \ge \frac{3\expect{\sum_{j=1}^{i} S_j \Mid  \EE_i}}{2} \Mid  \EE_i} \le \frac{2}{3}.
\end{align*}

The latter implies there is at least constant probability the agent has displacement at most $3\dist/4$ from the origin during the first $ c\dist^{\alpha-1} / \mu_\alpha$ jumps and in time $ \mybigo{\dist} $ (since the sum of all jumps is at most linear), without any conditional event:
\begin{align*}
   \pr{\sum_{j\le i} S_j \le 3\dist/4}
    \ge & \ \pr{\sum_{j\le i} S_j \le 3\dist/4 \Mid  \EE_i}\pr{\EE_i} \\
    \ge & \ \Theta(1)
\end{align*}
for each $i\le c\dist^{\alpha-1} / \mu_\alpha$.

Define the event $W_i=\{\sum_{j\le i} S_j \le 3\dist/4\}$. We now compute the probability that, given $i\le c\dist^{\alpha-1} / \mu_\alpha$, in the first $ i-1 $ jumps the displacement has been at most $ 3\dist/4 $ and during the $i$-th jump-phase the agent finds the \treasure. Let $F_i$ be such the latter event. Since
\[
    \pr{F_i, W_{i-1}} = \pr{F_i \Mid W_{i-1}}\pr{W_{i-1}},
\]
we estimate $\pr{F_i \mid W_{i-1}}$. Let $\levyrand{t}$ be the
two-dimensional random variable representing the coordinates of the nodes the \levywalk visits at time $t$. If $t_i$ is the time the agent ends the $i$-th jump-phase, we have
\[
    \pr{F_i \Mid W_{i-1}} \ge \sum_{v \in Q_{3\dist/4}(\origin)} \pr{F_i \Mid \levyrand{t_{i-1}} = v, W_{i-1}}\pr{\levyrand{t_{i-1}} = v \Mid W_{i-1}}.
\]
By \cref{cor:visit-direct-path}, the term $\pr{F_i \Mid \levyrand{t_{i-1}} = v, W_{i-1}}$ is $\myTheta{1/\dist^{\alpha}}$, and, since $ \levyrand{t_{i-1}} \in Q_{3\dist /4}(\origin) $ is implied by $ W_{i-1} $, we have
\[
    \sum_{v \in Q_{3\dist/4}(\origin)} \pr{F_i \Mid \levyrand{t_{i-1}} = v, W_{i-1}}\pr{\levyrand{t_{i-1}} = v \Mid W_{i-1}} \ge \Theta\left(\frac{1}{\dist^\alpha}\right)\cdot \pr{\levyrand{t_{i-1}} \in Q_{3\dist /4}(\origin) \Mid W_{i-1}} = \Theta\left(\frac{1}{\dist^\alpha}\right),
\]
implying $\pr{F_i,W_{i-1}} = \Omega\left(\frac{1}{\dist^\alpha}\right)$ for all $i\le c\dist^{\alpha-1} / \left (1 + \log \dist \cdot \ind _ {[\alpha = 2]}\right )$.
Then, for the chain rule, the probability that none of the events $F_i\cap W_{i-1}$ holds for each $i\le \dist^{\alpha-1} / \log(c\dist)$ is
\begin{align*}
    \pr{\bigcup_{i\le \frac{c\dist^{\alpha-1} }{ \mu_\alpha}}(F_i\cap W_{i-1})} =  & \ 1 - \pr{\bigcap_{i\le \frac{c\dist^{\alpha-1} }{ \mu_\alpha}}(F_i^C\cup W_{i-1}^C)} \\
    = & \ 1 - \prod_{i\le \frac{c\dist^{\alpha-1} }{\mu_\alpha}} \pr{F_i^C\cup W_{i-1}^C \bigm | \bigcap_{j\le i-1}(J_j^C\cup W_{j-1}^C)}
    \\
    = & \ 1 - \prod_{i\le \frac{c\dist^{\alpha-1} }{ \mu_\alpha}}\left(1 - \pr{F_i\cap W_{i-1} \bigm | \bigcap_{j\le i-1}(J_j^C\cup W_{j-1}^C)}\right) \\
    \\
    \stackrel{(\ast)}{\ge} & \ 1 - \prod_{i\le \frac{c\dist^{\alpha-1} }{\mu_\alpha}}\left(1 - \pr{F_i\cap W_{i-1}, \bigcap_{j\le i-1}(J_j^C\cup W_{j-1}^C)}\right)
    \\
    \stackrel{(\star)}{=} & \
     1 - \prod_{i\le \frac{c\dist^{\alpha-1} }{\mu_\alpha}}\left(1 - \pr{F_i\cap W_{i-1}}\right)
    \\
= & \ 1 -\left(1 - \Omega\left(\frac{1}{\dist^\alpha}\right)\right)^{\frac{c\dist^{\alpha-1} }{ \mu_\alpha}} \\
    \ge & \
    1 - e^{-\Omega\left(\frac{c}{\mu_\alpha \dist }\right)} = \Omega\left(\frac{c}{\mu_\alpha \dist}\right),
\end{align*}
where, $(\ast)$ holds since $\pr{A \mid B} \ge \pr{A,B}$, $(\star)$ holds since $W_{i-1} \subseteq (W_{j-1}^C \cup J_{j}^C)$ for $j\le i-1$, and the last equality holds by the inequality $ e^{-x} \ge 1 - x $ for all $ x $, and by the Taylor's expansion of $f(x) = e^x$. Then, there is probability at least $\Omega\left(c/\left (\mu_\alpha\dist\right )\right)$ to find the \treasure within time $\bigo(\dist)$.
\end{proof}

\cref{lemma:hittingparetoballistic} gives part (a) of \cref{thm:lw12,thm:lw2}, while part (b) comes from \cref{lemma:prob_never_find_target}.

\subsection{Proof of Corollary~\ref{cor:lw12}}

	First, suppose $ \alpha \in [1+\epsilon,2) $.
	From \cref{cor:lw12}.\cref{cor:lw12:a} and the independence between agents, we get that
	\[
	\pr{\tau_\alpha^k(\trespoint) = \mybigo{\dist}} = 1 - \left [ \myOmega{\frac{1}{\mu\dist}}\right ]^k \ge 1 - e^{-\myOmega{\frac{k}{\dist \log \dist }}},
	\]
	where we  used the inequality $ 1 - x \le e^{-x} $ for every real $ x $ and  the bound  $ \mu \le \log \dist $.
From \cref{cor:lw12}.\cref{cor:lw12:b} and the independence between agents, we get that
	\[
	\pr{\tau_\alpha^k(\trespoint) = \infty} = \left[1 - \mybigo{\frac{\mu \log \dist}{\dist}}\right] ^ k \ge \exp\left (-\mybigo{\frac{k \log^2 \dist}{\dist}}\right ),
	\]
	where we   used  again  $ \mu \le \log \dist $ and the inequality $ 1 - x \ge e^{-\frac{x}{1-x}} $ for every real $x$. Thus, if $ k = \mylittleo{\dist / \log^2 \dist} $, for the Taylor's expansion of the exponential function, we get hitting time $ \infty $ with probability $ 1 - o(1) $.
For $ \alpha = 2 $ the proof proceeds exactly in the same way.

\section{The Case \texorpdfstring{$\alpha\in(3,\infty)$}{alpha in (3,infty)}}
\label{sec:equivpwrw}

We analyze now the hitting time of \levywalks with parameter $\alpha \in (3,\infty)$, which is the exponent range for which the jump length has bounded mean and bounded variance.

\begin{theorem}
	\label{thm:lw3+}
	Let $\alpha\in (3, \infty)$ and  $\trespoint\in \integer^2$ with $\dist = \|\trespoint\|_1$.
	Let $\nu = \min\{\log\dist, \frac{1}{\alpha - 3}\}$, and $\gamma = \frac{\alpha^2}{(\alpha-3)^2}$. Then:
	\begin{enumerate}[(a)]
		\item \label{thm:lw3+:a}
		$ \pr{\tau_\alpha(\trespoint) = \mybigo{\dist^2 \log ^ 2 \dist}}
		= \myOmega{{1}/{(\gamma \log^4 \dist)}}
		$,
		if $\alpha \ge  3 + \mylittleomega{\log \log \dist / \log \dist} $;
		
		\item \label{thm:lw3+:b}
		$ \pr{\tau_\alpha(\trespoint) \leq t } = \mybigo{ \nu \cdot {t^2} /{\dist^{\alpha+1}} }$, for any step $\dist \le t = \mybigo{\dist^{2} / \nu}$.
	\end{enumerate}
\end{theorem}

From the above result, we easily obtain the following bounds on the parallel hitting time.

\begin{corollary}
	\label{cor:lw3+}
	Let $\alpha\in (3, \infty)$,  $\trespoint\in \integer^2$, $\dist = \|\trespoint\|_1$,  and $ 1 \le k \le \mylittleo{\dist^2} $.
	Let $\nu = \min\{\log\dist, \frac{1}{\alpha - 3}\}$ and $\gamma = \frac{\alpha^2}{(\alpha-3)^2}$. Then:
	\begin{enumerate}[(a)]
		\item \label{cor:lw3+:a}
		$ \pr{\tau_\alpha^k(\trespoint) = \mybigo{\dist^2 \log ^ 2 \dist}}
		=  1 - e^{ - \mylittleomega{\log \dist} }$, if $ k \ge \myOmega{\log^6\dist} $ and $ \alpha \ge 3 + \mylittleomega{\frac{\log \log \dist}{\log \dist}}  $;
		
		\item \label{cor:lw3+:b}
		$ \pr{\tau_\alpha^k(\trespoint) \leq \dist^2 / \sqrt{k}} = 1 - o(1)$.
	\end{enumerate}
\end{corollary}

\cref{cor:lw3+}\cref{cor:lw3+:a} says that $\tau_\alpha^k(\trespoint) = \mybigo{\dist^2\log^2\dist}$, w.h.p., for $k \geq \polylog\ell$ and $ \alpha \ge 3 + \mylittleomega{\frac{\log\log \dist}{\log \dist}} $, and \cref{cor:lw3+}\cref{cor:lw3+:b} provides a crude lower bound indicating that increasing $k$ beyond $\polylog\ell$, can only result in sublinear improvement.

\subsection{Proof of Theorem \ref{thm:lw3+}}
    \label{ssec:diffusive}

The structure of the proof is similar to that for \cref{thm:lw23,thm:lw3}.
We will use the next three lemmas, which are analogous to \cref{lemma:hittinglevyflight,lemma:couplevyflightintowalk,lemma:lowbound_superdiffusive}, respectively

\begin{lemma}
    [\levyflight with {$\alpha \in (3,\infty)$}]
    \label{prop:pfrwhittingtime}
    Let $h_f$ be the hitting time of a \levyflight for target $\trespoint\in\integer^2$, and let $\ell = \|\trespoint\|_1$.
    If $ \alpha - 3 = \mylittleomega{\log \log \dist / \log \dist} $, then
    \[
        \pr{h_f = \mybigo{ \dist^2 \log ^2 \dist}} = \myOmega{\frac{(\alpha - 3)^3}{\alpha^2\log^4\dist}}.
    \]
\end{lemma}

\begin{lemma}
    \label{lemma:coupling_levyflight_levywalk_diffusive}
    Let $h_f$ be defined as in \cref{prop:pfrwhittingtime}, and
    let $\tau_\alpha (\trespoint)$ be the hitting time of a \levywalk with the same $\alpha \in (3,\infty)$,
    for the same \treasure. Then, for every step $t$,
    \[
        \pr{\tau_\alpha (\trespoint) = \mybigo{t}}
        \geq
            \pr{h_f \leq t}
            - \mybigo{\frac{\alpha}{(\alpha - 3)t}}
        .
    \]
\end{lemma}

\begin{lemma}
    \label{lemma:pwrwlowerb2}
    Let $\alpha \in (3,\infty)$, $\trespoint\in\integer^2$, and $\ell =\|\trespoint\|_1$.
    For any step $t$ such that $\dist \leq t = \mybigo{\dist^2 / \nu}$,

    \[\pr{\tau_\alpha (\trespoint) \leq t}
    =
    \mybigo{ {\nu t^2}/{\dist^{\alpha + 1}}},
    \]
    where $ \nu = \min\{\log \dist, \frac{1}{\alpha - 3}\} $.
\end{lemma}

The proofs of the above lemmas are given in \cref{ssec:pfrw,sec:proof-coupling-lf3infty,ssec:pwrw-lowbounds}, respectively.
Using these lemmas can now prove our main result as follows.
From \cref{prop:pfrwhittingtime,lemma:coupling_levyflight_levywalk_diffusive}, by substituting $ t = \myTheta{\dist^2 \log^2 \dist} $, we get
\[
	\pr{\tau_\alpha (\trespoint) = \mybigo{\dist^2\log^2\dist}} = \myOmega{\frac{(\alpha - 3)^2}{\alpha^2\log^4\dist}} - \mybigo{\frac{\alpha}{(\alpha - 3)\dist^2 \log^2 \dist}} = \myOmega{\frac{(\alpha - 3)^2}{\alpha^2\log^4\dist}}
\]
if $3 + \mylittleomega{\log \log \dist / \log \dist} \le \alpha $, which is part (a) of \cref{thm:lw3+}.
Also, by applying \cref{lemma:pwrwlowerb2}, we get part (b) of \cref{thm:lw3+}.

 \subsection{Proof of Lemma \ref{prop:pfrwhittingtime}} \label{ssec:pfrw}

The proof is similar to that of \cref{lemma:hittinglevyflight}, in \cref{ssec:analysislevyflight}.
We reuse some of the notation defined there.
Namely, $(\Lf{i})_{i\geq0}$ denotes the \levyflight process, $S_i$ is the length of the $i$-th jump, and $\Zf{u}{i}$ is the number of visits to $u$ in the first $i$ steps.
We also use the following modified definitions:
For each node $u$ and $i\geq 0$, we let
\[
    p_{u,i} \ = \pr{\Lf{i} = u},
\]
thus
$
    \expect{\Zf{u}{i}}
    =
    \sum_{j=0}^i p_{u,j}.
$
We partition $\integer^2$ into sets $\AA_1,\AA_2,\AA_3$ as follows.
Let $\delta > 0$ be some value to be fixed later.
Then,
\[
  \begin{gathered}
    \AA_1 = \{v \colon \|v\|_\infty \le \dist\}
\\
    \AA_2
    =
    \{v \colon \manhattan{v} \le 4\sqrt{2(1+\delta)t}\log t\}
    \setminus \AA_1
\\
    \AA_3 = \integer^2\setminus (\AA_1\cup\AA_2).
  \end{gathered}
\]

\subsubsection{Proof Overview}
\label{sssec:roadmap-pfrw}

First, we bound the mean number of visits to $\AA_1$ until a given step $t$.
We show that
\[
    \expect{\Zf{\origin}{t}}
    =
    \mybigo{\log^2 t}.
\]
The monotonicity property from \cref{sec:monotonicity-bounds} implies that
\[
    \sum_{v\in \AA_1} \expect{\Zf{v}{t}}
    \le |\AA_1|\cdot \expect{\Zf{\origin}{t}}
    \leq
    c(3-\alpha)(\dist \log t)^2/\alpha
    ,
\]
where $c$ is a constant.
To bound the mean number of visits to $\AA_2$, as before, we use the monotonicity property again to obtain
\[
    \sum_{v\in \AA_2}
    \expect{\Zf{v}{t}}
    \le
    |\AA_2|\cdot
    \expect{\Zf{\trespoint}{t}}
    \leq
    32(1+\delta)t\log^2 t\cdot
    \expect{\Zf{\trespoint}{t}}
    .
\]
For the number of visits to $\AA_3$, using a Chernoff-Hoeffding bound we show that
\[
    \sum_{v\in \AA_3} \expect{\Zf{v}{t}}\leq
    c'\left (t^{1-(\alpha-3)/{2}} + 1\right ),
\]
for some constant $c'$ and for $ \delta = \myTheta{1/(\alpha - 3)^2} $ large enough.
Combining the above we obtain
\[
    c\dist^2\log^2t
    +
    32(1+\delta)t\log^2 t\cdot
    \expect{\Zf{\trespoint}{t}}
    +
    c'\left(t^{1-(\alpha-3)/{2}} + 1\right)
    \geq
    t.
\]
By choosing $t = \myTheta{\dist^2 \log^2 \dist }$ and $ \alpha - 3 = \mylittleomega{\log \log \dist / \log \dist} $, the above inequality implies
\[
    \expect{\Zf{\trespoint}{t}}=
    \Omega\left(
        \frac{1}{(1 + \delta)  \log^2 \dist}
    \right)
    .
\]
Since
\[
    \pr{h_f \leq t}
    =
    \pr{\Zf{\trespoint}{t}>0}
    =
    \expect{\Zf{\trespoint}{t}}/\,
    \expect{\Zf{\trespoint}{t} \ \middle|\ \Zf{\trespoint}{t}>0},
\]
and
$
    \expect{\Zf{\trespoint}{t} \ \middle|\ \Zf{\trespoint}{t}>0}
    \leq
    \expect{\Zf{\origin}{t}}+1
    =
    \mybigo{\log ^2 t},
$
we obtain
$
    \pr{h_f \leq t}
    =
    \myOmega{(\alpha - 3)^2/(\alpha^2  \log^4 \dist)}.
$
\subsubsection{Detailed Proof}

\begin{lemma}\label{lemma:pfrworiginvisits}
    For any $t \ge 0$, $\expect{\levyflightvisits{\origin}{t}} = b_t = \bigo\left  (\log^2 t\right )$.
\end{lemma}
\begin{proof}First, we show the following.  Let $\levyflightrand{t'}$ be the two dimensional random variable representing the coordinates of the agent performing the \levyflight at time $t'$.
    Let $\levyflightrand{t'} = (X_{t'}, Y_{t'})$ and consider the projection of the \levyflight on the $x$-axis $X_{t'}$: it can be expressed as the sum of $t'$ random variables $S^x_j$, $j=1,\dots,t'$, representing the projection of the jumps (with sign) of the agent on the $x$-axis at times $j=1,\dots,t'$. The partial distribution of the jumps along the $x$-axis is given by
    \cref{lemma:projection} in \cref{app:projection}, and states that, for any given $d\ge 1$, we have
    \[
         \pr{S^x_j = \pm d} = \myTheta{\frac{1}{d^{\alpha}}}.
    \]
Since $\expect{\levyflightvisits{\origin}{t} } = \sum_{k=1}^t p_{\origin,k}$, it suffices to accurately bound the probability $p_{\origin,k}$ for each $k=1, \dots, t$. Let us partition the natural numbers in the following way
    \[
        \nat = \bigcup_{t'=1}^\infty \bigg[\nat \cap \left[2t'\log t', 2(t'+1)\log(t'+1)\right)\bigg].
    \]
    For each $k\in \nat$, there exists $t'$ such that $k\in \left[2t'\log t', 2(t'+1)\log(t'+1)\right)$. Then, within $2t'\log t'$ steps the walk has moved to distance $\myTheta{\sqrt{t'}}$ at least once, with probability $\Omega\left(\frac{1}{(t')^2}\right)$.
    Indeed, the sequence $\{S^x_j\}_{1\le j \le t'}$ consists of i.i.d.\ r.v.s with zero mean and finite variance $ \myOmega{1} $. Thus, the central limit theorem (\cref{thm:CLT} in \cref{app:tools}) implies that the variable
    \[
        \frac{S^x_1 + \dots + S^x_{t}}{\sigma \sqrt{t}}
    \]
    converges in distribution to a standard normal random variable $Z$, with $ \sigma = \myOmega{1} $. Let $\epsilon>0$ be a small enough constant, then there exists a $\overline{t'}$ large enough, such that for all $t' \ge \overline{t'}$ it holds that
    \[
        \pr{ S^x_1 + \dots + S^x_{t'}\ge \sigma \sqrt{t'}} \ge \pr{Z \ge 1} - \epsilon = \frac{c}{2} >0,
    \]
    for some suitable constant $ c \in (0,1) $. The symmetrical results in which the normalized sum is less than $-\sigma\sqrt{t'
    }$ holds analogously.
    Thus, for all $t' \ge \overline{t'}$, we have that $\abs{\sum_{j=1}^{t'} S^x_j} \ge \sigma\sqrt{t'}$ with constant probability $c>0$. In $2t'\log t'$ jumps, we have $2\log t'$ sets of $t'$ consequent i.i.d.\ jumps. For independence, the probability that at least in one round before round  $2t'\log t'$ the \levyflight has displacement $\myTheta{\sqrt{t'}}$ from the origin is at least
    \[
        1-(1-c)^{2\log t'} = 1 - \mybigo{\frac{1}{(t')^2}}.
    \]
    Once reached such a distance, there are at least $\lambda^2 = \myTheta{t'}$ different nodes that are at least as equally likely as $\origin$ to be visited at any given future time for the monotonicity property (\cref{lemma:monotonicity}).
    Thus, the probability to reach the origin at any future time is at most $\mybigo{1/t'}$.
Let $H_{t'}$ be the event that in any instant before time $2t' \log t'$ the \levyflight has displacement at least $\myTheta{\sqrt{t'}}$. Observe that
    \[
        p_{\origin,k} = \pr{\levyflightrand{t} = \origin \mid H_{t'}}\pr{H_{t'}} + \pr{\levyflightrand{t} = \origin \mid H_{t'}^C}\pr{H_{t'}^C},
    \]
    by the law of total probability.
    We remark that in an interval $\left[2t'\log t', 2(t'+1)\log(t'+1)\right)$ there are
    \[2(t'+1)\log(t'+1)-2t'\log t' = 2t'\left[\log\left(1+\frac{1}{t'}\right)\right]+2\log(t'+1) = \bigo\left(\log t'\right)\]
    integers.
    Thus, if $I_{t'}=\left[2t'\log t', 2(t'+1)\log(t'+1)\right)$, we have
    \begin{align*}
        \sum_{k=1}^t p_{\origin,k} \le & \ \sum_{t' = 1}^t \sum_{k\in I_{t'}} p_{\origin,k} \\
        \le & \ \sum_{t' = 1}^t \left[\pr{\levyflightrand{t} = \origin \mid H_{t'}}\pr{H_{t'}} + \pr{\levyflightrand{t} = \origin \mid H_{t'}^C}\pr{H_{t'}^C }\right]\bigo(\log t')\\
        \le & \ \overline{t'}+ \sum_{t' = \overline{t'}}^t \left[\bigo\left( \frac{1}{ t'}\right)+\bigo\left(\frac{1}{(t')^2}\right)\right]\bigo(\log t') = \bigo\left (\log^2 t\right ),
    \end{align*}
    since $ \overline{t'} $ is a constant.
\end{proof}

We have also the following.

\begin{restatable}{lemma}{corollarypfrwvisitsorigin}\label{corollary:pfrwvisitsorigin}
    For any node  $u\in \integer^2$,   it holds that
    \begin{itemize}
        \item[(i)] $\expect{\levyflightvisits{u}{t}} \le b_t$;
        \item[(ii)] $1 \le \expect{\levyflightvisits{u}{t} \mid \levyflightvisits{u}{t} > 0} \le b_t$;
        \item[(iii)] $\expect{\levyflightvisits{u}{t}} / b_t \le \pr{\levyflightvisits{u}{t} > 0} \le \expect{\levyflightvisits{u}{t}}$.
    \end{itemize}
\end{restatable}
\begin{proof}The proof is exactly as that of \cref{corollary:visitsorigin}.
\end{proof}
Thus, the total number of visits to  $\rwfirstregion$ is upper bounded by $m_{\trespoint} b_t $, where $m_{\trespoint} = \cardinality{Q_{\dist}(\origin)}$.

Furthermore, from the monotonicity propert (\cref{lemma:monotonicity}), the following holds.
\begin{restatable}{corollary}{corollarypfrwinduction}\label{corollary:pfrwinduction}
    For any node $u$ in $\integer^2$, we have $\expect{\levyflightvisits{u}{t}} \ge \expect{\levyflightvisits{v}{t}}$ for all $v\notin Q_{\distu}(\origin)$. \end{restatable}
Namely, almost all the nodes that are ``further'' than $u$ from the origin are less likely to be visited at any given future time. This easily gives an upper bound on the total number of visits to $\rwsecondregion$ until time $t$, namely, by taking $u = \trespoint$ and by observing that each $v \in \secondregion$ lies outside $Q_{\dist}(\origin)$, we get that the average number of visits to $\secondregion$ is at most    the expected number of visits to the \treasure $\trespoint$ (i.e. $\expect{\levyflightvisits{\trespoint}{t}}$) times (any upper bound of) the size  of $\secondregion$:
in formula,  it is    upper bounded by $\expect{\levyflightvisits{\trespoint}{t}} \cdot 32(1+\delta)t \log^2 t$.

We also give a bound to the average number of visits to nodes that are further roughly $\sqrt{t} \cdot \log t$ from the origin.

\begin{lemma}\label{lemma:pfrwboundfaraway}
    A sufficiently large positive real $\delta$ exists such that  $\delta = \myTheta{1/(\alpha-3)^2}$  and
    \[
        \sum_{\substack{v \in \integer^2 \ : \\
        		\manhattan{v} \ge 4\sqrt{2(1+\delta)t}\log t}} \expect{\levyflightvisits{v}{t}} = \mybigo{t^{1-\frac{\alpha-3}{2}} + 1}.
    \]
\end{lemma}
\begin{proof}Since $ \alpha > 3 $, the expectation and the variance of a single jump-length are finite. By Equation \eqref{eq:jumpatleast} in the preliminaries (\cref{sec:preliminaries}), the probability a jump length is at least $\sqrt{t}$ is $\myTheta{1/t^\frac{\alpha-1}{2}}$. Let us call $A_j$ the event that the $j$-th jump-length is less than $\sqrt{t}$.
    Let us recall that $\levyflightrand{j}$ is the random variable denoting the coordinates of the nodes the corresponding \levyflight visits at the end of the $j$-th jump.
    We can write $\levyflightrand{j} = (X_{j}, Y_{j})$, where $X_{j}$ is $x$-coordinate of the \levywalk after the $j$-th jump, and $Y_{j}$ is the $y$-coordinate.
    Then, $X_{j}$ can be seen as the sum $\sum_{i=1}^j S^x_{i}$ of $j$ random variables representing the projections of the jumps along the $x$-axis.
    For symmetry, $\expect{X_j} = 0$ for each $j$, while $\variance{X_j} = j\variance{S^x_1} = \mybigo{j/(\alpha - 3) + j} = \mybigo{\alpha j / (\alpha - 3)}$ since $S^x_1$ has finite variance $ \mybigo{1 + 1 / (\alpha - 3)} $. This comes by observing that $S^x_1 \le S_1$.
    Then, conditional on $A = \cap_{i=1}^{t} A_i$, we can apply the Chernoff bound (\cref{lem:chernoffbound:variance}) on the sum of the first $j$ jumps, for $j\le t$. We have
    \begin{align*}
        \pr{\abs{X_{t}} \ge 2\sqrt{2(1+\delta)t}\log t \bigm| A} & \le 2\exp\left(-\frac{8(1+\delta)t\cdot \log^2 t}{\mybigo{\frac{\alpha t}{\alpha - 3}}+\myTheta{\sqrt{(1+\delta)t}\cdot \log t}\sqrt{t}}\right)\\
        & \le 2\exp\left(-\myTheta{\frac{\alpha-3}{\alpha}\sqrt{1+\delta}\cdot \log t}\right)\\
        & \le \frac{2}{t^{\myTheta{\frac{\alpha - 3}{\alpha}\sqrt{1 + \delta}}}},
    \end{align*}
    which is less than $1/t ^2$ if we choose $\delta = \myTheta{ 1/(\alpha - 3)^2}$ large enough. The same result holds for the random variable $X_j$ for each $j<t$, since the variance of $X_j$ is smaller than the variance of $X_t$. Notice that
    \begin{align*}
        \pr{\cap_{j=1}^t \{\abs{X_j} < 2\sqrt{2(1+\delta)t}\log t \} \mid A} & = 1 - \pr{\cup_{j=1}^t \{\abs{X_t} \ge 2\sqrt{2(1+\delta)t}\log t \} \mid A} \\
        & \ge 1 - \frac{t}{t^ 2} = 1 - \frac{1}{t },
    \end{align*}
    and that
    \begin{align*}
        \pr{A} = 1 - \pr{A^C} = 1 - \pr{\cup_{j=1}^tA_j^C} \ge 1 - \mybigo{\frac{t}{t^\frac{\alpha-1}{2}}} = 1 - \mybigo{\frac{1}{t^\frac{\alpha-3}{2}}}.
    \end{align*}
    An analogous argument holds for the random variable $Y_t$ conditioned to the event $A$. Then,
    \begin{align*}
        & \ \pr{\cap_{j=1}^t \{\manhattan{X_j} < 2\sqrt{2(1+\delta)t}\cdot\log t\}, \cap_{j=1}^t \{\manhattan{Y_j} < 2\sqrt{2(1+\delta)t}\cdot\log t\}} \\
        \ge & \ \pr{\cap_{j=1}^t \{\manhattan{X_j} < 2\sqrt{2(1+\delta)t}\cdot\log t\}, \cap_{j=1}^t \{\manhattan{Y_j} < 2\sqrt{2(1+\delta)t}\cdot\log t\}\mid A}\pr{A}\\
        \ge & \  \left(2\pr{\cap_{j=1}^t \{\manhattan{X_j} < 2\sqrt{2(1+\delta)t}\cdot\log t\} \mid A} - 1\right)\pr{A} \\
        \stackrel{(\ast)}{\ge} & \ \left[2\left(1 - \frac{1}{t }\right) - 1\right]\left( 1 - \mybigo{\frac{1}{t^\frac{\alpha-3}{2}}}\right)  \\
        \ge & \ 1 - \mybigo{\frac{1}{t^\frac{\alpha-3}{2}} + \frac{1}{t}},
    \end{align*}
    where $(\ast)$ holds for symmetry (the distribution of $Y_t$ is the same as the one of $X_t$) and for the union bound. Thus, in $t$ jumps (which take at least time $t$), the walk has never reached distance $4\sqrt{2(1+\delta)t}\cdot \log t$, w.h.p. The average number of visits until time $t$ to nodes at distance at least $4\sqrt{2(1+\delta)t}\cdot \log t$ is then less than $t\cdot \mybigo{1/ t^\frac{\alpha-3}{2} + 1/t} = \mybigo{t^{1-\frac{\alpha-3}{2}} + 1}$.
\end{proof}

The following puts together the previous estimations in order to get a lower bound on the average number of visits the \treasure $\trespoint$. Let $ \delta > 0 $ be as given in \cref{lemma:pfrwboundfaraway} for the rest of the section.

\begin{lemma}\label{lemma:pfrwgoodtrick}
    For every node  $\trespoint \in   \integer^2 $ and     every time $t \geq 1$,
\[
        m_{\trespoint} b_t + \expect{\levyflightvisits{\trespoint}{t}}\cdot 32(1+\delta)(t\log^2 t)+ \mybigo{t^{1-\frac{\alpha-3}{2}}+1} \ge t.
    \]
\end{lemma}
\begin{proof}Suppose the agent has made $t$ jumps, thus visiting $t$ nodes. Then,
    \[
        \expect{\sum_{v\in \integer^2}\levyflightvisits{v}{t} } = t.
    \]
    We divide the plane in different zones, and we bound the number of visits over each zone in expectation.
    From \cref{corollary:pfrwvisitsorigin}, the number of visits inside $\firstregion=Q_{\dist}(\origin)$ until time $t$ is at most $m_{\trespoint} b_t$, where $m_{\trespoint} = \abs{Q_{\dist}(\origin)} = 4\dist^2$. From \cref{lemma:pfrwboundfaraway}, the number of visits $\thirdregion$ is at most $\mybigo{t^{1-\frac{\alpha-3}{2}}}$. Each of the remaining nodes, i.e.\ the nodes in $\secondregion$, which are at most $32(1+\delta)(t\log^2 t)$ in total, is visited by the agent at most $\expect{\levyflightvisits{u}{t} }$ times, for \cref{corollary:pfrwinduction}. Then, we have that
    \[
        m_{\trespoint} b_t + \expect{\levyflightvisits{\trespoint}{t}}\cdot 32(1+\delta)(t\log^2 t)+ \mybigo{t^{1-\frac{\alpha-3}{2}}+1} \ge t.
        \qedhere
    \]
\end{proof}

We can now complete the proof of \cref{prop:pfrwhittingtime} as follows.
 	\cref{lemma:pfrwgoodtrick} implies that
    \[
    	\expect{\levyflightvisits{\trespoint}{t}} = \myOmega{\frac{t - t^{1 - \frac{\alpha - 3}{2}} - 1 - m_{\trespoint}b_t}{(1 + \delta) t\log^2t}},
    \]
    while \cref{corollary:pfrwvisitsorigin} implies that
    \[
    	\pr{h_f \le t} = \myOmega{\frac{t - t^{1- \frac{\alpha - 3}{2}} - 1 - m_\trespoint b_t}{(1+\delta) t\log^2 t \cdot b_t}}.
    \]
    \cref{lemma:pfrworiginvisits} gives $ b_t = \mybigo{\log^2 t} $, while $ \cref{lemma:pfrwboundfaraway} $ gives $ \delta = \myTheta{1/(\alpha-3)^2 } $. If $ t = \myTheta{\dist^2 \log^2 \dist} $ is large enough and $ \alpha - 3 = \mylittleomega{\log\log \dist / \log \dist} $, so that $t - t^{1- \frac{\alpha - 3}{2}} - m_\trespoint b_t = \myTheta{t} $, we get the result.

\subsection{Proof of Lemma \ref{lemma:coupling_levyflight_levywalk_diffusive}}
    \label{sec:proof-coupling-lf3infty}

     If $S_i$ is the random variable yielding the $i$-th jump length, then it has expectation $ \myTheta{1} $ and variance. This means that the sum $\Bar{S}_t = \sum_{i=1}^t S_i$ has expectation $\Theta(t)$ and variance $\mybigo{t + t/(\alpha - 3)} = \mybigo{\alpha t / \alpha-3}$. Then, from  Chebyshev's inequality,
    \[
        \pr{\Bar{S}_t \ge \Theta(t) +  t} \le \frac{\variance{\Bar{S}_t}}{t^2} = \mybigo{\frac{\alpha}{(\alpha - 3)t}}.
    \]
    Hence,
    \begin{align*}
    	\pr{h_w = \mybigo{t}} & \ge \pr{h_f \le t, \Bar{S}_t \le \Theta(t) + t } = \pr{h_f \le t} - \mybigo{\frac{\alpha}{(\alpha -3) t}},
    \end{align*}
    where the latter equality is obtained using the union bound.
 \subsection{Proof of Lemma \ref{lemma:pwrwlowerb2}}
\label{ssec:pwrw-lowbounds}

Let $X_i$ be the $x$-coordinate of the agent at the end of the $i$-th jump-phase. For any $i\le t$, we bound the probability that $X_i > \dist / 4$. The probability that there is a jump whose length is at least $\dist$ among the first $i$ jumps is $\bigo(i/\dist^{\alpha-1})$ for the union bound. Conditional on the event that the first $i$ jump lengths are all smaller than $\dist$ (event $\EE_i$), the expectation of $X_i$ is zero and its variance is
    \[
        i\cdot\sum_{d=1}^{\dist/4}\Theta\left( d^2 / d^\alpha\right) = \mybigo{i\nu \dist^{3-\alpha}},
    \]
    for the integral test (\cref{fact:integraltest}),  where $ \nu = \min \{\log \dist, \frac{1}{\alpha - 3}\} $.
    Chebyshev's inequality implies that
    \[
        \pr{\abs{X_i} \ge \dist/4 \mid \EE_i} \le \frac{\mybigo{i\nu \dist^{3-\alpha} }}{\Theta(\dist^2)} = \mybigo{\frac{i\nu}{\dist^{\alpha - 1}}},
    \]

    Since the conditional event has probability $1-\bigo(i/\dist^{\alpha-1})$, then the ``unconditional'' probability that of the event $\abs{X_i} \le \dist/4$ is
    \[
        \left[1-\mybigo{\frac{i}{\dist^{\alpha-1}}}\right]\cdot \left[1 - \mybigo{\frac{i\nu}{\dist^{\alpha - 1}}}\right] = 1-\mybigo{\frac{\nu t }{\dist^{\alpha - 1}}},
    \]
    since $i\le t $, with $ t  $ which is some function in $\mybigo{{\dist^{\alpha -1}}/{\nu}} $.
    The same result holds analogously for $Y_i$ (the $y$-coordinate of the agent after the $i$-th jump), obtaining that $\abs{X_i} + \abs{Y_i} \le \dist/2$ with probability $1-\mybigo{\nu t/\dist^{\alpha -1}}$ by the union bound.

    Consider the first jump-phase. The probability the agents visits the \treasure during it is $\bigo(1/\dist^\alpha)$ for \cref{cor:visit-direct-path} (\cref{sec:preliminaries}). Now, let $2 \le i \le t$. We want to estimate the probability the agent visits the \treasure during the i-th jump-phase. We recall that $ B_{\dist/4}(\trespoint) $ is the rhombus centered in $\trespoint$ that contains the nodes at distance at most $\frac{\dist}{4}$ from $\trespoint$.
    We denote the event that the agent visits the \treasure during the $i$-th jump-phase by  $F_i$. Furthermore, let $V_{i-1}$ be the event that the $(i-1)$-th jump ends in $B_{\dist/4}(\trespoint)$, and $W_{i-1}$ the event that  $(i-1)$-th jump ends at distance farther than $\dist/2$ from the origin.
    Then, by the law of total probabilities, we have
    \begin{align}
        \pr{F_i} = & \ \pr{F_i \mid W_{i-1}}\pr{W_{i-1}} + \pr{F_i \mid W_{i-1}^C}\pr{W_{i-1}^C} \nonumber\\
        = & \ \left[\pr{F_i \mid W_{i-1}, V_{i-1}}\pr{V_{i-1}\mid W_{i-1}}+\pr{F_i \mid W_{i-1}, V_{i-1}^C}\pr{V_{i-1}^C\mid W_{i-1}}\right]\pr{W_{i-1}} \nonumber\\
        & + \ \pr{F_i \mid W_{i-1}^C}\pr{W_{i-1}^C} \nonumber\\
        \stackrel{(\ast)}{\le} & \ \left[\pr{F_i \mid V_{i-1}}\pr{V_{i-1}\mid W_{i-1}}+\pr{F_i \mid W_{i-1}, V_{i-1}^C}\right]\pr{W_{i-1}} + \pr{F_i \mid W_{i-1}^C}\pr{W_{i-1}^C} \nonumber\\
        \stackrel{(\star)}{\le} & \  \left[\bigo\left(\frac{1}{\dist^2}\right) + \bigo\left(\frac{1}{\dist^\alpha}\right)\right]\bigo\left(\frac{\nu t}{\dist^{\alpha-1}}\right) + \bigo\left(\frac{1}{\dist^\alpha}\right) = \bigo\left(\frac{\nu t}{\dist^{\alpha+1}}\right)\label{eq:pwrwlowbound2}
    \end{align}
    where in $(\ast)$ we used that $V_{i-1} \subset W_{i-1}$ and that $\pr{V_{i-1}^C \mid W_{i-1}}\le 1$, while in $(\star)$ we used that
    \[
        \pr{F_i \mid V_{i-1}}\pr{V_{i-1}\mid W_{i-1}} = \bigo\left(\frac{1}{\dist^2}\right), \text{ (the proof is below) }
    \]
    that $\pr{F_i \mid W_{i-1},V_{i-1}^C} = \bigo\left(1/\dist^\alpha\right)$ because the jump starts in a node whose distance form the \treasure is $\Omega(\dist)$, and that $\pr{F_i \mid W_{i-1}^C} =\bigo\left(1/\dist^\alpha\right)$ for the same reason. As for the term $\pr{F_i \mid V_{i-1}}\cdot \pr{V_{i-1}\mid W_{i-1}}$ we observe the following. Let $t_i$ be the time at the end of the $ i $-th jump phase. Then
    \begin{align*}
        \pr{F_i \mid V_{i-1}}\pr{V_{i-1}\mid W_{i-1}} = & \ \sum_{v \in B_{\dist/4}(\trespoint)} \pr{F_i \mid \levyrand{t_i} = v}\pr{\levyrand{t_i} = v \mid W_{i-1}} \\
        \le & \ \bigo\left(\frac{1}{\dist^2}\right)\sum_{v \in B_{\dist/4}(\trespoint)} \pr{F_i \mid \levyrand{t_i} = v},
    \end{align*}
    since \cref{lemma:monotonicity}
holds in a consequent way conditional on $W_{i-1}$, and since, for each $v\in B_{\dist/4}(\trespoint)$, there are at least $\Theta\left(\dist^2\right)$ nodes at distance at least $\dist/2$ from the origin which are more probable to be visited than $v$.
    Then, we proceed similarly to the proof of \cref{lemma:prob_anyjumpfindstreasure}
    to
show that $\sum_{v \in B_{\dist/4}(\trespoint)} \pr{F_i \mid \levyrand{t_i} = v} = \bigo(1)$, and we obtain $\pr{F_i \mid V_{i-1}}\pr{V_{i-1}\mid W_{i-1}} = \bigo\left(1/\dist^2\right)$.

    Thus, by the union bound and by the inequality \eqref{eq:pwrwlowbound2}, the probability that at least during one of the $ t $ jump-phases the agent finds the \treasure is, for some $ t = \mybigo{\dist^{2}/\nu} $,
    \begin{align*}
        \frac{1}{\dist^\alpha} + \bigo\left(\frac{\nu t^2}{\dist^{\alpha + 1}}\right) & =  \mybigo{\frac{\nu t^2}{\dist^{\alpha + 1}}},
    \end{align*}
    since $ t \ge \dist $.

\subsection{Proof of Corollary~\ref{cor:lw3+}}

	From \cref{thm:lw3+}.\cref{thm:lw3+:a} and the independence between agents we ge that
	\[
		\pr{\tau_\alpha^k(\trespoint) = \mybigo{\dist^2 \log^2 \dist} } = 1 - \left[1 - \myOmega{\frac{1}{\gamma \log^4 \dist}}\right]^ k \ge 1 - e^{-\myOmega{\frac{k}{\gamma \log^4 \dist}}},
	\]
	where we have used the inequality $ 1 - x \le e^{-x} $ for all $ x $. Then, part (a) follows.
	Let $ t = \dist^2/k \cdot \sqrt{k} / \log^2 \dist$; hence, we have $  t \le \mylittleo{\dist^{\alpha - 1} / \nu } $ since $ k = \mylittleo{\dist^2} $. If $ t < \dist $, then $ \pr{\tau_\alpha^k(\trespoint) > t} = 1 $, since $ \dist $ steps are needed to reach distance $ \dist $. If $ t \ge \dist $, from \cref{thm:lw3+}.\cref{thm:lw3+:b} and the independence between agents, we get that
	\[
		\pr{\tau_\alpha^k (\trespoint) > t} = \left[1 - \mybigo{\frac{\nu t^2}{\dist^{\alpha+1}}}\right]^k \ge \exp\left (-\mybigo{\frac{k t^2 \log \dist }{\dist^{\alpha+1}}}\right ),
	\]
	where we have used the inequality $ 1 - x \ge e^{-\frac{x}{1 - x}} $ for $ x < 1 $, that $ \nu t^2 / \dist^{\alpha+1} = \mylittleo{1} $, and that $ \nu \le \log \dist $.  Then, by substituting $ t = \dist^2/k \cdot \sqrt{k} / \log^2 \dist $, and by the Taylor's expansion  of the exponential function, we get
	\[
		\exp\left (-\mybigo{\frac{k t^2 \log \dist }{\dist^{\alpha+1}}}\right )	= 1 - \mybigo{\frac{1}{\log^3 \dist}}.
	\]
	\cref{cor:lw3+}.\cref{cor:lw3+:b} follows by observing that $ \dist^2/k \cdot \sqrt{k} / \log^2 \dist \le \dist^2 / \sqrt{k} $.

\section{Distributed Search Algorithm}
\label{sec:algorithm_analysis}

In this section, we prove the following theorem, which provides a simple distributed search algorithm, which allows $k$ agents to find an arbitrary, unknown target on $\integer^2$ in optimal time (modulo polylogarithmic factors).

\begin{theorem}
	\label{thm:uniform-algo}
	Consider $k$ independent \levywalks that start simultaneously from the origin, and the exponent of each walk is sampled independently and uniformly at random from the real interval $(2,3)$.
	Let $\tau^k_{\mathit{rand}}(u^\ast)$ be the parallel hitting time for a given target $u^\ast$.
	If $k \geq \log^8 \dist$, and $\ell = \manhattan{u^\ast}$ is large enough, then
	\[
	\pr{\tau^k_{\mathit{rand}}(u^\ast) = \mybigo{  (\dist^2/k)\cdot \log^7\dist + \dist\log^3\dist}}
	=
	1-e^{-\omega(\log\dist)}
	.
	\]
\end{theorem}

We need the next lemma, which is a slight generalization of \cref{cor:plw23}\ref{cor:plw23:a} that bounds the hitting time of a collection of \levywalks with different exponent values.

\begin{lemma}
	\label{thm:phtub23}
	Consider $k$ independent \levywalks that start simultaneously from the origin, and the exponent of each walk is in $[\alpha_1,\alpha_2]$.
	Let $h_{\mathit{diff}}$ be the parallel hitting time for a target $u^\ast$ with $\manhattan{u^\ast}=\dist$. If $2 < \alpha_1 \leq \alpha_2 \leq 3 - \epsilon$ and $\epsilon = \omega(1/\log\dist)$,
	then
	\[
		\pr{h_{\mathit{diff}} = \mybigo{\frac{\dist^{\alpha_2 - 1}}{\alpha_1 - 2}}}
		= 1 - e^{-\myOmega{\frac{(3-\alpha_2)^2  k} {\dist^{3-\alpha_1}  \log^2 \dist}}}.
	\]
	
\end{lemma}

\begin{proof}
	First, we recall that the $ k $ agents move independently from each other. Let $ \tau_\alpha(\trespoint) $ be the hitting time of a single walk. 
If $ \alpha \in [\alpha_1, \alpha_2] $, then from \cref{lemma:hittinglevyflight,lemma:couplevyflightintowalk},
	\[
		\pr{\tau_\alpha(\trespoint) = \mybigo{\frac{\dist^{\alpha - 1}}{\alpha - 2}}} = \myOmega{\frac{(3-\alpha)^2}{\dist^{3 - \alpha} \log^{\frac{2}{\alpha - 1}} \dist}} = \myOmega{\frac{(3 - \alpha_2)^2}{\dist^{3 - \alpha_1}\log^2 \dist}},
	\]
	provided that $ 3 - \alpha_2 = \omega\left (1 / \log \dist\right ) $.
	Observe that $ \pr{\tau_\alpha(\trespoint) = \mybigo{\frac{\dist^{\alpha_2 -1}}{\alpha_1 - 2}}} \ge \pr{\tau_\alpha(\trespoint) = \mybigo{\frac{\dist^{\alpha - 1}}{\alpha - 2}}} $. Then,
	\[
		\pr{h_{\mathit{diff}} = \mybigo{\frac{\dist^{\alpha_2 - 1}}{\alpha_1 - 2}}}
		= 1 - \left ( 1 - \myOmega{\frac{(3 - \alpha_2)^2}{\dist^{3 - \alpha_1}\log^2 \dist}}\right )^k \le 1 - e^{- \myOmega{\frac{(3 - \alpha_2)^2k}{\dist^{3 - \alpha_1}\log^2 \dist}}}.
    \qedhere
	\]
\end{proof}

We can now prove our main result.

\begin{proof}[\bf Proof of \cref{thm:uniform-algo}]
	Fix $k,\dist$ such that $k \geq \log^8 \dist$, and let $\epsilon = \log\log\dist /\log\dist$.
	Let $\alpha\in[2+\epsilon,3-2\epsilon]$, and
let $k_\alpha$ be the number of L\'evy walks whose exponent is in the interval $[\alpha,\alpha+\epsilon]$.
	Then $\expect{k_\alpha} = \epsilon k$, and by the Chernoff bound \cref{lem:chernoff:multiplicative},
	\[
	\pr{k_\alpha \geq \epsilon k/2}
	=
	1 - e^{-\myOmega{\epsilon k}} .
	\]
	Clearly, the parallel hitting time of the $k$ \levywalks is upper bounded by the parallel hitting time of the $k_\alpha$ L\'evy walks whose exponent is in $[\alpha,\alpha+\epsilon]$.
	Then, from \cref{thm:phtub23}, it follows that
	\[
		\pr{\tau^k_{\mathit{rand}}(u^\ast) = \mybigo{\frac{\dist^{\alpha +\epsilon - 1}}{\alpha - 2} }
		\Mid
		k_\alpha \geq \epsilon k/2
		}
	= 1 - e^{-\myOmega{\frac{(3-\alpha-\epsilon)^2 \epsilon k} {\dist^{3-\alpha}  \log^2 \dist}}}
	.
	\]
	Combining the last two equations, we obtain
	\begin{equation}
	\label{eq:hrand}
	\pr{
	\tau^k_{\mathit{rand}}(u^\ast) = \mybigo{\frac{\dist^{\alpha +\epsilon - 1}}{\alpha - 2}}}
	= 1 - e^{-\myOmega{\frac{(3-\alpha-\epsilon)^2  \epsilon k} {\dist^{3-\alpha}  \log^2 \dist}}}
	.
	\end{equation}
	We distinguish the following two cases.
	
	\paragraph{Case $\log^7\dist \leq k \leq  \dist\log^3\dist$.}
	Choose some $\alpha\in[2+\epsilon,3-2\epsilon]$ such that
	\[
	k
	=
	\frac{\dist^{3-\alpha}\cdot \log^2 \dist} {(3-\alpha-\epsilon)^2\cdot \epsilon}
	\cdot\log \dist\cdot \log\log\dist
	.
	\]
	Such an $\alpha$ exists because the values of function $f(\alpha) = \frac{\dist^{3-\alpha}\cdot \log^3\dist\cdot\log\log\dist} {(3-\alpha-\epsilon)^2\cdot \epsilon}$ at the extreme points of $\alpha$ are
	$
	f(2+\epsilon)
\geq
	\dist \log^3 \dist\geq k
	$
	and
	$
	f(3-2\epsilon)
\leq
	\log^8\dist
	\leq k
	.
	$
	Substituting the above value of $\alpha$ to \cref{eq:hrand}, we obtain
	\[\pr{\tau^k_{\mathit{rand}}(u^\ast) = \mybigo{\frac{\dist^{\alpha +\epsilon - 1}}{\alpha - 2} }}
	=
	1-e^{-\omega(\log\dist)}
	.\]
	Thus, with probability
	$1 - e^{-\omega(\log\dist )}$,
	\[
	k\cdot \tau^k_{\mathit{rand}}(u^\ast)
	=
	k\cdot
	\mybigo{\frac{\dist^{\alpha +\epsilon - 1}}{\alpha - 2} }
	=
	\mybigo{
	\frac{\dist^{2}\cdot \log^5 \dist} {(3-\alpha-\epsilon)^2(\alpha - 2)}
	}
	=
	o\left(
	\dist^{2}\cdot \log^7 \dist
	\right)
	.
	\]
	
	\paragraph{Case $k \geq  \dist\log^3\dist$.}
	In this case, we set $\alpha = 2 + \epsilon$ and substitute this value of $\alpha$ to \cref{eq:hrand} to obtain
	$
	\pr{\tau^k_{\mathit{rand}}(u^\ast) = \mybigo{{\dist^{1 + 2\epsilon}}/ {\epsilon} }}
	=
	1 - e^{-\myOmega{ \frac{k\log\log\dist}{\dist\log^2\dist} }}
	=
	1-e^{-\omega(\log\dist)}
	.
	$
	Thus, with probability
	$1 - e^{-\omega(\log\dist )}$,
	\[
	\tau^k_{\mathit{rand}}(u^\ast)
	=
	\mybigo{ {\dist^{1 + 2\epsilon}}/ {\epsilon} }
	=
	o\left(\dist \cdot\log^3\dist
	\right)
	.
	\]
	Combining the two cases completes the proof.
\end{proof} 
\phantomsection
\addcontentsline{toc}{section}{APPENDIX}
\appendix
\section*{APPENDIX}
\section{Tools}\label{app:tools}

Below we list some standard concentration bounds we  use, and prove some simple algebraic inequalities.

\subsection{Concentration Bounds}
\begin{theorem}[Central limit theorem~{\cite[Chapter X]{feller1}}]\label{thm:CLT}
    Let $\{X_k\}_{k\ge 1}$ be a sequence of i.i.d.\ random variables. Let $\mu = \expect{X_1}$, $\sigma^2 = \variance{X_1}$, and $S_n = \sum_{k=1}^n X_k$ for any $n \ge 1$. Let $\Phi : \real \to [0,1]$ be the cumulative distribution function of a standard normal distribution. Then, for any $\beta \in \real$,
    \[
        \lim_{n \to \infty} \pr{\frac{S_n - n\mu }{\sigma \sqrt{n}} < \beta} = \Phi(\beta).
    \]
\end{theorem}

\begin{theorem}
[Multiplicative Chernoff bounds~\cite{dubhashi2009}]\label{lem:chernoff:multiplicative}
	Let $X_1, X_2, \dots, X_n$ be independent random variables taking values in $[0,1]$. Let $X = \sum_{i=1}^n X_i$ and $\mu=\mathbb{E}[X]$. Then:
	\begin{enumerate}
		\item[(i)] For any $\delta > 0$ and $\mu \le \mu_+ \le n$,
		\begin{equation}\label{MCB+}
		    P\big(X\ge (1+\delta)\mu_+\big)\le e^{-\frac{1}{3}\delta^2\mu_+},
		\end{equation}
		\item[(ii)] For any $\delta \in (0,1)$ and $0 \le \mu_- \le \mu$,
		\begin{equation}\label{MCB-}
			P\big(X\le (1-\delta)\mu_-\big)\le e^{-\frac{1}{2}\delta^2\mu_-}.
		\end{equation}
	\end{enumerate}
\end{theorem}

\begin{theorem}
[Additive Chernoff bound using variance {\cite[Theorem 3.4]{chung2006}}]\label{lem:chernoffbound:variance}
    Let $X_1, \dots, X_n$ be independent random variables satisfying $X_i \le \expect{X_i} + M$ for some $M\ge 0$, for all $i = 1, \dots, n$. Let $X = \sum_{i=1}^n X_i$, $\mu = \expect{X}$, and $\sigma^2 = \variance{X}$. Then, for any $\lambda > 0$,
    \begin{equation}\label{ACB+}
        \pr{X \ge \mu + \lambda} \le \exp\left(-\frac{\lambda^2}{\sigma^2 + \frac{M\lambda}{3}}\right).
    \end{equation}
\end{theorem}

\subsection{Inequalities}

\begin{lemma}
    \label{fact:integraltest}
    Let $0 < d < \dmax$ be any integers.
    For any $\alpha > 1$,
    \begin{gather}
         \frac{1}{(\alpha-1)(d)^{\alpha-1}} \ \le \ \sum_{k\ge d} \frac{1}{k^\alpha} \ \le \ \frac{1}{(\alpha-1)(d)^{\alpha-1}}+ \frac{1}{d^\alpha}, \ \ \text{and} \label{eq:integraltest1.1}
        \\ \frac{1}{(\alpha-1)}\left(\frac{1}{d^{\alpha-1}} - \frac{1}{\dmax^{\alpha-1}}\right) \ \le \ \sum_{k= d}^{\dmax} \frac{1}{k^\alpha} \ \le \ \frac{1}{(\alpha-1)}\left(\frac{1}{d^{\alpha-1}} - \frac{1}{\dmax^{\alpha-1}}\right) + \frac{1}{d^\alpha}. \label{eq:integraltest1.2}
    \end{gather}
    Also,
    \begin{equation}
        \log\left(\frac{\dmax}{d}\right) \ \le \ \sum_{k=d}^{\dmax} \frac{1}{k} \ \le \ \log\left(\frac{\dmax}{d}\right) + \frac{1}{d} \label{eq:integraltest2},
    \end{equation}
    and for any $0<\alpha<1$,
    \begin{equation}
        \frac{(\dmax)^{1-\alpha}-d^{1-\alpha}}{1-\alpha} \ \le \ \sum_{k = d}^{\dmax} \frac{1}{k^\alpha} \ \le \ \frac{(\dmax)^{1-\alpha}-d^{1-\alpha}}{1-\alpha} + \frac{1}{d^\alpha}.
        \label{eq:integraltest3}
    \end{equation}
\end{lemma}
\begin{proof}
    By the integral test, it holds that
    \[
        \int_d^{\dmax} \frac{1}{k^\alpha}dk \le \sum_{k= d}^{\dmax} \frac{1}{k^\alpha} \le \int_d^{\dmax} \frac{1}{k^\alpha}dk +  \frac{1}{d^\alpha}.
    \]
    Straightforward calculations give the result for \cref{eq:integraltest1.2,eq:integraltest2,eq:integraltest3}. As for \cref{eq:integraltest1.1}, it comes from the integral test letting $\dmax \to \infty$.
\end{proof}

\section{Proofs Omitted from Section \ref{sec:preliminaries}}

\subsection{Proof of Lemma \ref{lemma:direct-path}}
\label{app:preliminaries}

\color{gray}

\begin{figure}
    \centering
    \includegraphics[scale=0.75]
        {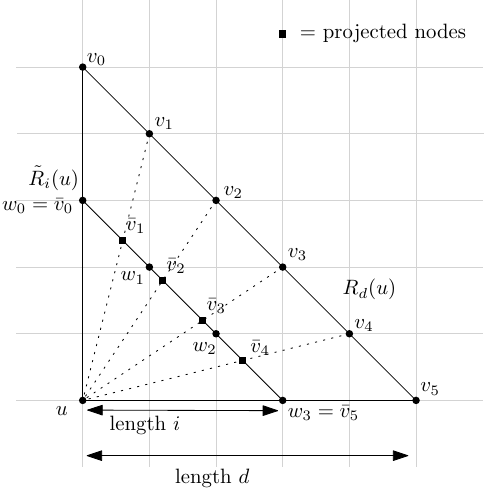}
    \caption{Projection from $ R_d(u) $ to $ \tilde R_i(u) $, with $ d = 5 $, $ i=3 $.}
    \label{fig:direct-path-projection}
\end{figure}

\color{black}

For $ 1 \le i < d $, consider the real rhombus $ \tilde R_i(u) $ which is the set $ \{v \in \real^2 : \manhattan{u-v} = i\} $. Project each element of $ R_d(u) $ on $ \tilde R_i(u) $ as in \cref{fig:direct-path-projection}. We obtain $ 4d $ equidistant points $ \bar{v}_1, \dots, \bar{v}_{4d} $ in Euclidean distance. Then, for each $ w \in R_i(u) $, consider the set $ C_w \subseteq \{\bar{v}_1, \dots, \bar{v}_{4d}\} $ of points that are closest to $ w $ than to any other node of $ R_i(u) $ in Euclidean distance.
Note that each $\bar{v}_j$ belongs to either one or two sets $C_w$; in the latter case we say that $\bar{v}_j$ is \emph{shared}.
For the cardinality of set $C_w$, we have the following cases:
(i)~if $d \equiv 0 \pmod i$ then $|C_w| = d/i + 1$, and two of the elements of $C_w$ are shared;
(ii)~if $d \not\equiv 0 \pmod i$ then either $|C_w| = \myfloor{d/i}$ and no elements of $C_w$ are shared, or $|C_w| = \myceil{d/i}$ and at most one element of $C_w$ is shared.

	Choose a node $ v \in R_d(u) $ u.a.r.\ and look at a node $ w \in R_i(u) $. The probability that $ w $ is on a direct-path chosen u.a.r.\ is exactly the probability that the projection $ \bar{v} $ of $ v $ on $ \tilde R_i(u) $ belongs to $ C_w $, where shared points contribute by $ 1/2 $.
In all cases, this probability is between $\myfloor{d/i}/(4d)$ and $\myceil{d/i}/(4d)$.
\color{black}
 
\subsection{Proof of Lemma \ref{lemma:monotonicity}}
\label{app:monotonicity}

    For any node $ w $, define $ D(w) $ as the set $ B_{\manhattan{w}}(\origin) \cup Q_{\supdist{w}}(\origin) $. Notice that $ D(u) \subseteq Q_{\manhattan{u}}(\origin) $. Then it suffice to prove that for all nodes $ v \notin D(u) $ we have
    \[
    	\pr{J_t = u} \ge \pr{J_t = v}.
    \]

    \begin{figure}[tb]
        \centering
        \includegraphics[scale=0.75]{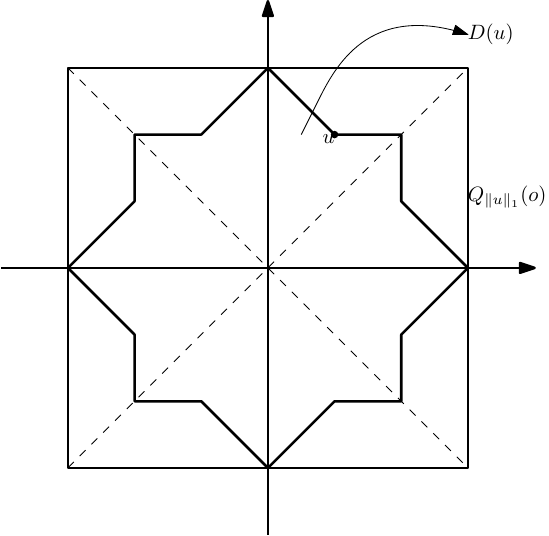}
        \caption{The set $D(u)$, consisting in all inner nodes of the ``star'', and the square $Q_{\distu}(\origin)$.}
        \label{fig:shape01}
    \end{figure}

    \begin{figure}[t]
        \centering
        \includegraphics[scale=0.75]{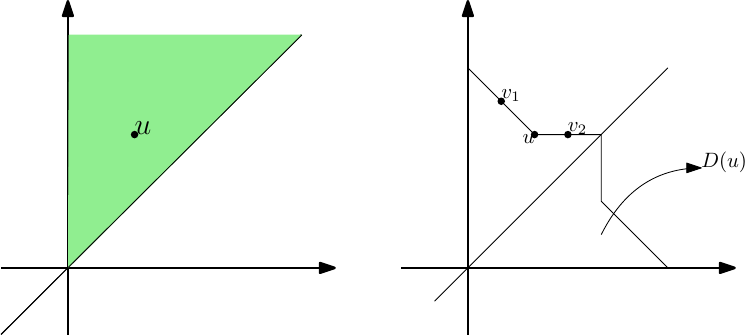}
        \caption{The ``area'' in which we take $u$, and the possible choices of $v$.}
        \label{fig:induction02}
    \end{figure}
    \begin{figure}[t]
        \centering
        \includegraphics[scale=0.75]{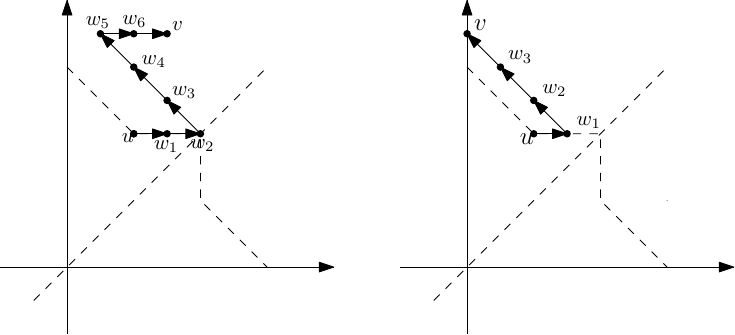}
        \caption{Two ``path'' examples.}
        \label{fig:induction03}
    \end{figure}
    \begin{figure}[t]
        \centering
        \includegraphics[scale=0.75]{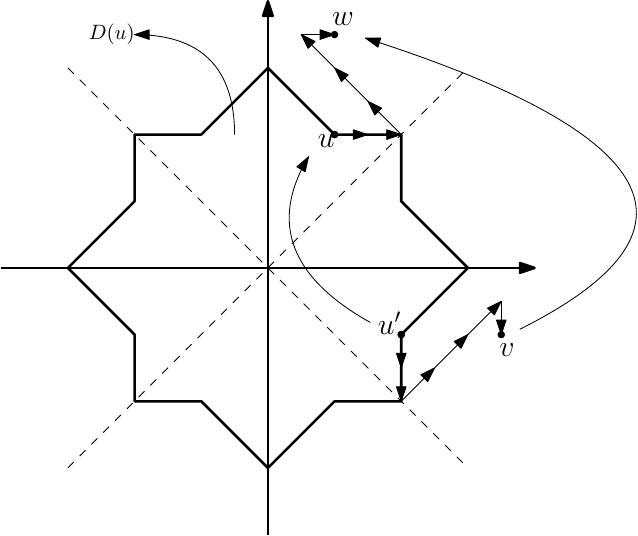}
        \caption{Symmetrical argument.}
        \label{fig:induction04}
    \end{figure}
    \begin{figure}[t]
        \centering
        \includegraphics[scale=0.75]{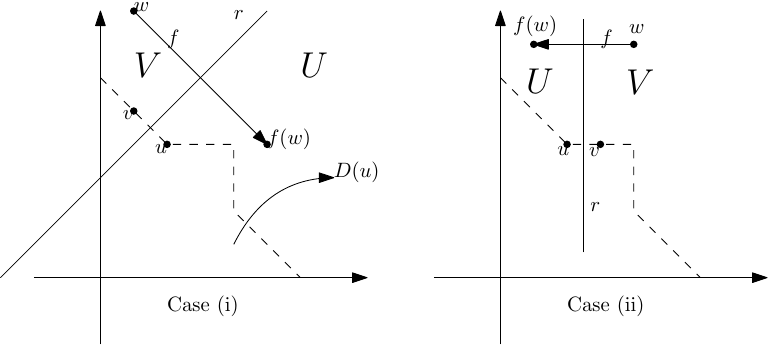}
        \caption{Geometric constructions in the two cases.}
        \label{fig:induction05}
    \end{figure}
    \begin{figure}[t]
        \centering
        \includegraphics[scale=0.75]{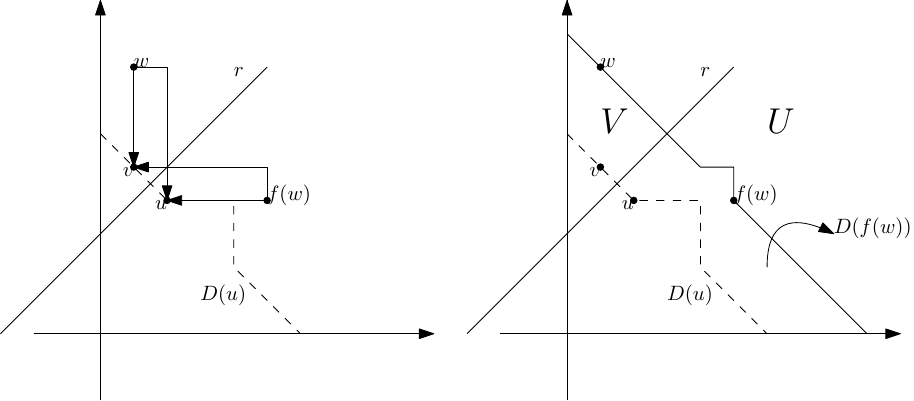}
        \caption{Left: equivalence between distances. Right: $w\notin D(f(w))$.}
        \label{fig:induction06}
    \end{figure}

    Let $u=(x_u,y_u)$ and,
    without loss of generality, suppose $u$ is in the first quadrant and not below the main bisector, i.e.\ in the set $\{(x,y) \in \integer^2 : y\ge 0, x\ge y\}$ (\cref{fig:induction02}).
    If we show that, for any $ v $ in $ \{v_1 =(x_u-1,y_u+1),v_2 = (x_u+1,y_u)\}$ (\cref{fig:induction02})., we have $ \pr{J_t = u} \ge \pr{J_t = v}$, than we have the statement.
    Indeed, for any $v\notin D(u)$ that ``lives'' in the highlighted area in \cref{fig:induction02}, there exists a sequence of nodes $u=w_0, w_1, \dots, w_k = v$ from $u$ to $v$ such that $w_{i+1}$ belongs to the set
    \[
        \{(x_{w_i}-1,y_{w_i}+1), (x_{w_i}+1,y_{w_i})\},
    \]
    where $w_i = (x_{w_i}, y_{w_i})$,
    as \cref{fig:induction03} shows. Thus, if the thesis is true for $v\in\{v_1,v_2\}$, then it is true also for all $v\notin D(u)$ in the highlighted area in \cref{fig:induction02}.
    At the same time, for any other $v\notin D(u)$, outside the highlighted area in \cref{fig:induction02}, there exists a symmetrical argument explained in \cref{fig:induction04}. Thus, if the thesis is true for all $v\notin D(u)$ in the highlighted area in \cref{fig:induction02}, then it is also true for any $v\notin D(u)$.
    We now consider some geometric constructions which will be used in the proof, one for each choice of $v$. The following description is showed in \cref{fig:induction05} .

    \begin{enumerate}[(i)]
        \item $v = (x_u-1,y_u+1)$: consider the strict line defined by $r: y=x+(y_u-x_u)+1$ (i.e.\ the line in $\real^2 $ which is the set of points that are equidistant from $u$ and $v$ in Euclidean distance). Call $V \subset \integer^2$ the set of nodes that are ``above'' this line, namely the ones that are closer to $v$ than $u$. Define $U = \integer^2\backslash (V\cup r)$ the complementary set without line $r$.
        Consider the injective function $f: V \to U$ such that $f(x,y) = (y-(y_u-x_u)-1, x+(y_u-x_u)+1)$, which is the symmetry with respect to $r$. It trivially holds that for any $w\in V$, $\manhattan{w-v}=\manhattan{f(w)-u}$ and $\manhattan{w-u}=\manhattan{f(w)-v}$. Furthermore, it holds that for each $w\in V$, either $w \notin D\left(f(w)\right)$, or $w$ lies on the ``border'' of $D\left(f(w)\right)$. All these properties are well-shown in \cref{fig:induction06}.

        \item $v = (x_u+1,y_u)$: the same construction can be done in this case. Indeed, the strict line will be $x = x_u + \frac{1}{2}$, and the injective function $f(x,y) = (2x_u+1 - x,y)$. The same properties we have seen in the previous case hold here too.
    \end{enumerate}

    Now we go for the proof. For any time $i$, and any two nodes $u',v' \in \integer^2$, define
    \[
        p^i(u',v')=\pr{\genericrand{i} = v' \mid \genericrand{0} = u'}.
    \]
    Let $v\in\{v_1,v_2\}$. We show that $p^t(\origin,u)\ge p^t(\origin,v)$ by induction on $t$.
    The base case is $t=1$. From the monotonicity, we know that
    \begin{align*}
        p^1(\origin,u)-p^1(\origin,v) \ge 0
    \end{align*}
    for any $u$ and $v$ in $\integer^2$ such that $\manhattan{u} \le \manhattan{v}$.
    We now suppose $t\ge 2$ and the thesis true for $t-1$. Fix $u$ and $v$ as in \cref{fig:induction02}; then, for the geometric construction we made above, it holds that
    \begin{align*}
        p^t(\origin,u)-p^t(\origin,v) = & \ \sum_{w\in \integer^2} p^{t-1}(\origin,w)\left(p^1(w,u)-p^1(w,v)\right)  \\
        \ge & \ \sum_{w\in U} p^{t-1}(\origin,w)\left(p^1(w,u)-p^1(w,v)\right) + \sum_{w\in V} p^{t-1}(\origin,w)\left(p^1(w,u)-p^1(w,v)\right)
    \end{align*}
    where last inequality is immediate for case (ii), indeed the line $r$ does not contain elements of $\integer^2$, while in case (i) the sum over nodes in line $r$ is zero. Then, the previous value is equal to
    \[
        \sum_{w\in V} p^{t-1}(\origin,f(w))\left(p^1(f(w),u)-p^1(f(w),v)\right) + \sum_{w\in V} p^{t-1}(\origin,w)\left(p^1(w,u)-p^1(w,v)\right)
    \]
    because of the definition of $f: V \to U$, and, changing the sign of the second sum, we obtain
    \[
        \sum_{w\in V} p^{t-1}(\origin,f(w))\left(p^1(f(w),u)-p^1(f(w),v)\right) - \sum_{w\in V} p^{t-1}(\origin,w)\left(p^1(w,v)-p^1(w,u)\right).
    \]
    Now, observe that the definition of $f$ implies that for each $w\in V$, $\manhattan{w-v}=\manhattan{f(w)-u}$ and $\manhattan{f(w)-v}=\manhattan{w-u}$ (\cref{fig:induction06}).
    Thus we can group out the term $p^1(f(w),u)-p^1(f(w),v)=p^1(w,v)-p^1(w,u)$, and we have
    \begin{equation}
    \label{eq:last_of_part_2}
        \sum_{w\in V} \left(p^{t-1}(\origin,f(w))-p^{t-1}(\origin,w)\right)\left(p^1(f(w),u)-p^1(f(w),v)\right).
    \end{equation}
    We observe that $p^{t-1}(\origin,f(w))-p^{t-1}(\origin,w)\ge 0$ by the inductive hypothesis, since either $w\notin D(f(w))$ or $w$ lies on the ``border'' of $D\left(f(w)\right)$ (\cref{fig:induction06}), and $p^1(f(w),u)-p^1(f(w),v)\ge 0$ by definition of $f$, since the distance between $f(w)$ and $u$ is no more than the distance between $f(w)$ and $v$.
    It follows that \eqref{eq:last_of_part_2} is non-negative, and, thus, the thesis.
 \section{Projection of a \levyFlight Jump}\label{app:projection}

Let $\levyflightrand{t}$ be the two dimensional random variable representing the coordinates of an agent performing an $\alpha$-\levyflight at time $t$, for any $\alpha > 1$.
Consider the projection of the \levyflight on the $x$-axis, namely the random variable $X_{t'}$ such that $\levyflightrand{t} = (X_{t}, Y_{t})$.
The random variable $X_{t}$ can be expressed as the sum of $t$ random variables $S^x_j$, $j=1,\dots,t$, representing the projection of the jumps (with sign) of the agent on the $x$-axis at times $j=1,\dots,t$.
With the next lemma, we prove that the jump projection length has the same tail distribution as the original jump length.

\begin{restatable}{lemma}{lemmaprojection}\label{lemma:projection}
    The probability that a jump $S^x_j$ has length equal to $d$ is $\myTheta{1/d^\alpha}$.
\end{restatable}
\begin{proof}
    The partial distribution of the jumps along the $x$-axis is given by the following. For any $d \ge 0$,
    \begin{equation}\label{eq:pfjumpprojection}
        \pr{S^x_j = \pm d } = \left[ \frac{1}{2}  +  \sum_{k= 1}^{\infty} \frac{\levyconst}{2k^{\alpha+1}}\right] \ind_{d=0} + \left[\frac{\levyconst}{2d^{\alpha+1}}+\sum_{k= 1+d}^{\infty} \frac{\levyconst}{k^{\alpha+1}}\right]\ind_{d\neq 0}\ ,
    \end{equation}
    where $\ind_{d\in A}$ returns $1$ if $d\in A$ and 0 otherwise, the term
    \[
        \frac{\ind_{d=0}}{2} + \frac{\levyconst}{2d^{\alpha+1}}\ind_{d\neq 0}
    \]
    is the probability that the original jump lies along the horizontal axis and has ``length'' exactly $d$ (there are two such jumps if $d > 0$), and, for $k\ge 1+d$, the terms
    \[
        \frac{\levyconst}{2k^{\alpha+1}}\ind_{d = 0} +  \frac{\levyconst}{k^{\alpha+1}}\ind_{d\neq 0}
    \]
    are the probability that the original jump has ``length'' exactly $k$ and its projection on the horizontal axis has ``length'' $d$ (there are two such jumps if $d =0$, and four such jumps if $d > 0$). 
    By the integral test (\cref{fact:integraltest}) we know that quantity $\eqref{eq:pfjumpprojection}$ is
    \[
        \pr{S^x_j = \pm d} = \myTheta{\frac{1}{d^{\alpha+1}}}.
    \]
\end{proof}

\phantomsection
\addcontentsline{toc}{section}{References}
\bibliographystyle{abbrv}\bibliography{bibliography}

\begin{thebibliography}{10}

\bibitem{alon_many_2011}
N.~Alon, C.~Avin, M.~Koucký, G.~Kozma, Z.~Lotker, and M.~R. Tuttle.
\newblock Many {random} {walks} {are} {gaster} {than} {one}.
\newblock {\em Combinatorics, Probability and Computing}, 20(4):481--502, 2011.

\bibitem{boczkowski_random_2018}
L.~Boczkowski, B.~Guinard, A.~Korman, Z.~Lotker, and M.~Renault.
\newblock Random {walks} with {multiple} {step} {lengths}.
\newblock In {\em {LATIN} 2018: {Theoretical} {Informatics}}, pages 174--186,
  2018.

\bibitem{boyer2006}
D.~Boyer, G.~Ramos-Fern{\'a}ndez, O.~Miramontes, J.~L. Mateos, G.~Cocho,
  H.~Larralde, H.~Ramos, and F.~Rojas.
\newblock Scale-free foraging by primates emerges from their interaction with a
  complex environment.
\newblock {\em Proceedings of the Royal Society B: Biological Sciences},
  273(1595):1743--1750, 2006.

\bibitem{Buldyrev2001}
S.~V. Buldyrev, S.~Havlin, A.~Y. Kazakov, M.~G.~E. da~Luz, E.~P. Raposo, H.~E.
  Stanley, and G.~M. Viswanathan.
\newblock Average time spent by {L}\'evy flights and walks on an interval with
  absorbing boundaries.
\newblock {\em Phys. Rev. E}, 64:041108, 2001.

\bibitem{comment2021}
S.~V. Buldyrev, E.~P. Raposo, F.~Bartumeus, S.~Havlin, F.~R. Rusch, M.~G.~E.
  da~Luz, and G.~M. Viswanathan.
\newblock Comment on ``inverse square {L}\'evy walks are not optimal search
  strategies for $d\ensuremath{\ge}2$''.
\newblock {\em Phys. Rev. Lett.}, 126:048901, Jan 2021.

\bibitem{chung2006}
F.~Chung and L.~Lu.
\newblock {C}oncentration inequalities and martingale inequalities: {A} survey.
\newblock {\em Internet Mathematics}, 3, 2006.

\bibitem{CELU17}
L.~Cohen, Y.~Emek, O.~Louidor, and J.~Uitto.
\newblock Exploring an infinite space with finite memory scouts.
\newblock In {\em Proceedings of the Twenty-Eighth Annual {ACM-SIAM} Symposium
  on Discrete Algorithms, {SODA}}, pages 207--224, 2017.

\bibitem{dubhashi2009}
D.~P. Dubhashi and A.~Panconesi.
\newblock {\em Concentration of measure for the analysis of randomized
  algorithms}.
\newblock Cambridge University Press, 2009.

\bibitem{edwards_revisiting_2007}
A.~M. Edwards, R.~A. Phillips, N.~W. Watkins, M.~P. Freeman, E.~J. Murphy,
  V.~Afanasyev, S.~V. Buldyrev, M.~G. E.~d. Luz, E.~P. Raposo, H.~E. Stanley,
  and G.~M. Viswanathan.
\newblock Revisiting {L}évy flight search patterns of wandering albatrosses,
  bumblebees and deer.
\newblock {\em Nature}, 449(7165):1044--1048, 2007.

\bibitem{efremenko_how_2009}
K.~Efremenko and O.~Reingold.
\newblock How well do random walks parallelize?
\newblock In {\em Approximation, Randomization, and Combinatorial Optimization.
  Algorithms and Techniques, 12th International Workshop, {APPROX}}, pages
  476--489, 2009.

\bibitem{ElsasserSauwervald2011}
R.~Elsässer and T.~Sauerwald.
\newblock Tight bounds for the cover time of multiple random walks.
\newblock {\em Theoretical Computer Science}, 412(24):2623 -- 2641, 2011.

\bibitem{ELSUW15}
Y.~Emek, T.~Langner, D.~Stolz, J.~Uitto, and R.~Wattenhofer.
\newblock How many ants does it take to find the food?
\newblock {\em Theoretical Computer Science}, 608:255 -- 267, 2015.

\bibitem{EmekLUW14}
Y.~Emek, T.~Langner, J.~Uitto, and R.~Wattenhofer.
\newblock Solving the {ANTS} problem with asynchronous finite state machines.
\newblock In {\em Automata, Languages, and Programming - 41st International
  Colloquium, {ICALP} 2014}, pages 471--482, 2014.

\bibitem{feinerman_ants_2017}
O.~Feinerman and A.~Korman.
\newblock The {ANTS} problem.
\newblock {\em Distributed Comput.}, 30(3):149--168, 2017.

\bibitem{feller1}
W.~Feller.
\newblock {\em An Introduction to Probability Theory and Its Applications},
  volume~1.
\newblock Wiley, 1968.

\bibitem{focardi2009}
S.~Focardi, P.~Montanaro, and E.~Pecchioli.
\newblock Adaptive l{\'e}vy walks in foraging fallow deer.
\newblock {\em PLoS One}, 4(8):e6587, 2009.

\bibitem{fraigniaud_parallel_2016}
P.~Fraigniaud, A.~Korman, and Y.~Rodeh.
\newblock Parallel exhaustive search without coordination.
\newblock In {\em Proceedings of the {Forty}-eighth {Annual} {ACM} {Symposium}
  on {Theory} of {Computing}, STOC}, pages 312--323, 2016.

\bibitem{guinard2020}
B.~Guinard and A.~Korman.
\newblock The search efficiency of intermittent {L}{\'{e}}vy walks optimally
  scales with target size.
\newblock {\em CoRR}, abs/2003.13041, 2020.

\bibitem{GuinardK20}
B.~Guinard and A.~Korman.
\newblock Tight bounds for the cover times of random walks with heterogeneous
  step lengths.
\newblock In {\em 37th International Symposium on Theoretical Aspects of
  Computer Science, {STACS} 2020}, pages 28:1--28:14, 2020.

\bibitem{humphries2010}
N.~E. Humphries, N.~Queiroz, J.~R. Dyer, N.~G. Pade, M.~K. Musyl, K.~M.
  Schaefer, D.~W. Fuller, J.~M. Brunnschweiler, T.~K. Doyle, J.~D. Houghton,
  et~al.
\newblock Environmental context explains {L}{\'e}vy and {B}rownian movement
  patterns of marine predators.
\newblock {\em Nature}, 465(7301):1066--1069, 2010.

\bibitem{Humphries7169}
N.~E. Humphries, H.~Weimerskirch, N.~Queiroz, E.~J. Southall, and D.~W. Sims.
\newblock Foraging success of biological {L}{\'e}vy flights recorded in situ.
\newblock {\em Proceedings of the National Academy of Sciences},
  109(19):7169--7174, 2012.

\bibitem{IvaskovicKPS17}
A.~Ivaskovic, A.~Kosowski, D.~Pajak, and T.~Sauerwald.
\newblock Multiple random walks on paths and grids.
\newblock In H.~Vollmer and B.~Vall{\'{e}}e, editors, {\em 34th Symposium on
  Theoretical Aspects of Computer Science, {STACS}}, pages 44:1--44:14, 2017.

\bibitem{kanade_coalescence_2019}
V.~Kanade, F.~Mallmann{-}Trenn, and T.~Sauerwald.
\newblock On coalescence time in graphs: When is coalescing as fast as
  meeting?: Extended abstract.
\newblock In {\em Proceedings of the Thirtieth Annual {ACM-SIAM} Symposium on
  Discrete Algorithms, {SODA}}, pages 956--965, 2019.

\bibitem{kleinberg_small_world_2000}
J.~Kleinberg.
\newblock The {small}-world {phenomenon}: {An} {algorithmic} {perspective}.
\newblock In {\em Proceedings of the {Thirty}-second {Annual} {ACM} {Symposium}
  on {Theory} of {Computing}, STOC}, pages 163--170, 2000.

\bibitem{LenzenLNR17}
C.~Lenzen, N.~A. Lynch, C.~Newport, and T.~Radeva.
\newblock Searching without communicating: {T}radeoffs between performance and
  selection complexity.
\newblock {\em Distributed Comput.}, 30(3):169--191, 2017.

\bibitem{levernier2020}
N.~Levernier, J.~Textor, O.~B\'enichou, and R.~Voituriez.
\newblock Inverse square {L}\'evy walks are not optimal search strategies for
  $d\ensuremath{\ge}2$.
\newblock {\em Phys. Rev. Lett.}, 124:080601, 2020.

\bibitem{reply2021}
N.~Levernier, J.~Textor, O.~B\'enichou, and R.~Voituriez.
\newblock Reply to ``comment on `inverse square l\'evy walks are not optimal
  search strategies for $d\ensuremath{\ge}2$'''.
\newblock {\em Phys. Rev. Lett.}, 126:048902, Jan 2021.

\bibitem{palyulin_first_2019}
V.~V. Palyulin, G.~Blackburn, M.~A. Lomholt, N.~W. Watkins, R.~Metzler,
  R.~Klages, and A.~V. Chechkin.
\newblock First passage and first hitting times of {Lévy} flights and {Lévy}
  walks.
\newblock {\em New Journal of Physics}, 21(10):103028, 2019.

\bibitem{raichlen2014}
D.~A. Raichlen, B.~M. Wood, A.~D. Gordon, A.~Z. Mabulla, F.~W. Marlowe, and
  H.~Pontzer.
\newblock Evidence of {L}{\'e}vy walk foraging patterns in human
  hunter--gatherers.
\newblock {\em Proceedings of the National Academy of Sciences},
  111(2):728--733, 2014.

\bibitem{razin2013desert}
N.~Razin, J.-P. Eckmann, and O.~Feinerman.
\newblock Desert ants achieve reliable recruitment across noisy interactions.
\newblock {\em Journal of the Royal Society Interface}, 10(82):20130079, 2013.

\bibitem{reynolds2017}
A.~Reynolds, G.~Santini, G.~Chelazzi, and S.~Focardi.
\newblock The {W}eierstrassian movement patterns of snails.
\newblock {\em Royal Society open science}, 4(6):160941, 2017.

\bibitem{reynolds_current_2018}
A.~M. Reynolds.
\newblock Current status and future directions of {Lévy} walk research.
\newblock {\em Biology Open}, 7(1):bio030106, 2018.

\bibitem{Reynolds3763}
A.~M. Reynolds, A.~D. Smith, D.~R. Reynolds, N.~L. Carreck, and J.~L. Osborne.
\newblock Honeybees perform optimal scale-free searching flights when
  attempting to locate a food source.
\newblock {\em Journal of Experimental Biology}, 210(21):3763--3770, 2007.

\bibitem{shlesinger_levy_1986}
M.~F. Shlesinger and J.~Klafter.
\newblock Lévy walks versus {L}évy flights.
\newblock In {\em On Growth and Form: {F}ractal and Non-Fractal Patterns in
  Physics}, pages 279--283. Springer Netherlands, Dordrecht, 1986.

\bibitem{sims2008}
D.~W. Sims, E.~J. Southall, N.~E. Humphries, G.~C. Hays, C.~J. Bradshaw, J.~W.
  Pitchford, A.~James, M.~Z. Ahmed, A.~S. Brierley, M.~A. Hindell, et~al.
\newblock Scaling laws of marine predator search behaviour.
\newblock {\em Nature}, 451(7182):1098--1102, 2008.

\bibitem{uchiyama2011}
K.~Uchiyama.
\newblock The first hitting time of a single point for random walks.
\newblock {\em Electron. J. Probab.}, 16:1960--2000, 2011.

\bibitem{viswanathan_levy_1996}
G.~M. Viswanathan, V.~Afanasyev, S.~V. Buldyrev, E.~J. Murphy, P.~A. Prince,
  and H.~E. Stanley.
\newblock Lévy flight search patterns of wandering albatrosses.
\newblock {\em Nature}, 381(6581):413--415, 1996.

\bibitem{viswanathan_optimizing_1999}
G.~M. Viswanathan, S.~V. Buldyrev, S.~Havlin, M.~G.~E. da~Luz, E.~P. Raposo,
  and H.~E. Stanley.
\newblock Optimizing the success of random searches.
\newblock {\em Nature}, 401(6756):911--914, 1999.

\bibitem{viswanathan_physics_2011}
G.~M. Viswanathan, M.~G. E.~d. Luz, E.~P. Raposo, and H.~E. Stanley.
\newblock {\em The {Physics} of {Foraging}: {An} {Introduction} to {Random}
  {Searches} and {Biological} {Encounters}}.
\newblock Cambridge University Press, Cambridge ; New York, 2011.

\bibitem{viswanathan_levy_2008}
G.~M. Viswanathan, E.~P. Raposo, and M.~G.~E. da~Luz.
\newblock Lévy flights and superdiffusion in the context of biological
  encounters and random searches.
\newblock {\em Physics of Life Reviews}, 5(3):133--150, 2008.

\bibitem{Werner1974}
E.~E. Werner and D.~J. Hall.
\newblock Optimal foraging and the size selection of prey by the bluegill
  sunfish (lepomis macrochirus).
\newblock {\em Ecology}, 55(5):1042--1052, 1974.

\bibitem{Zaburdaev2015lw}
V.~Zaburdaev, S.~Denisov, and J.~Klafter.
\newblock L\'evy walks.
\newblock {\em Rev. Mod. Phys.}, 87:483--530, 2015.

\end{thebibliography}

\end{document}